\documentclass{article}
\usepackage{amsmath,amsthm,amssymb,hyperref}
\usepackage{ytableau}
\usepackage{tabmac}
\usepackage{algorithm,algpseudocode}
\algrenewcommand\algorithmicrequire{\textbf{Input:}}

\usepackage[margin=1.25in]{geometry}
\usepackage{graphicx} 

\newcommand{\F}{\mathbb{F}}

\DeclareMathOperator{\stab}{stab}
\DeclareMathOperator{\tr}{tr}
\DeclareMathOperator{\Aut}{Aut}
\DeclareMathOperator{\Ind}{Ind}
\DeclareMathOperator{\Res}{Res}

\DeclareMathOperator{\mult}{mult}
\DeclareMathOperator{\GL}{GL}
\DeclareMathOperator{\Hom}{Hom}
\DeclareMathOperator{\Irr}{Irr}

\newtheorem{theorem}{Theorem}[section]
\newtheorem{lemma}[theorem]{Lemma}
\newtheorem{corollary}[theorem]{Corollary}
\newtheorem{proposition}[theorem]{Proposition}
\theoremstyle{definition}
\newtheorem{definition}[theorem]{Definition}
\newtheorem{conjecture}[theorem]{Conjecture}
\newtheorem{observation}[theorem]{Observation}
\newtheorem{fact}[theorem]{Fact}
\newtheorem{open}[theorem]{Open Question}

\theoremstyle{remark}
\newtheorem{remark}[theorem]{Remark}
\newtheorem{example}[theorem]{Example}

\title{Geometric Complexity Theory and Graph Isomorphism}
\author{Joshua A. Grochow and Jacob Urisman}
\date{September 2026}

\begin{document}

\maketitle

\begin{abstract}
    We investigate ideas from the Geometric Complexity Theory approach to separating complexity classes (Mulmuley \& Sohoni, \emph{SIAM J. Comput.}, 2001) in the setting of graph isomorphism. This provides us a playground of finite combinatorial objects on which to explore these techniques, in comparison with the high-dimensional algebraic varieties associated to complexity classes. We seek to separate non-isomorphic pairs of graphs by finding vector spaces of polynomials that are set-wise invariant under permutations (so-called separating modules, which are, in particular, representations of the symmetric group). We characterize the power of this method for distinguishing non-isomorphic graphs under several different complexity measures, relating this GCT-inspired approach to more standard approaches for distinguishing graphs:
    \begin{itemize}
        \item Degree. We show that separating modules of ``support-degree'' $k$ (each monomial touches at most $k$ vertices) are equivalent in power to the counts of $O(k)$-vertex subgraphs. This is strictly weaker than $O(k)$-dimensional Weisfeiler--Leman by a previous result of Fürer (ICALP '01).

    \item Symmetric algebraic circuit size. If we measure the complexity of a separating module by the size of the smallest multi-output symmetric algebraic circuit that computes a spanning set for it, 
    we show that separating modules of symmetric circuit size $n^{\Theta(k)}$ are equivalent to $\Theta(k)$-WL. This generalizes and strengthens the result of Dawar \& Wilsenach (CSL '18; ICALP '20; \emph{ACM Trans. Comput. Log.}, 2022; \emph{Theory Comput.}, 2025): they proved one direction of this equivalence, and only for invariant polynomials, whereas we generalize to  separating modules and prove both directions. 

    \item Multiplicities. 
    When considering only the representation-theoretic multiplicities of separating modules (as was proposed in GCT by Mulmuley \& Sohoni, \emph{ibid.}, rather than the polynomials themselves), we show that two graphs are separated by multiplicities if and only if their automorphism groups have different multiplicity of cycle types (cycle index). 
    \end{itemize}

    The latter result is notable in the analogy with GCT, as it is the only result we are aware of in which the multiplicity approach to  separating isomorphism types of objects has been given an exact ``intrinsic'' characterization in terms of the objects themselves. We use this characterization to show that for graphs, multiplicity obstructions are stronger than occurrence obstructions. 
    We also make connections between support size (from the study of Weisfeiler--Leman) and complexity measures on the symmetric group (Dafni, Filmus, Lifshitz, Lindzey, \& Vinyals, ITCS '21); as well as connections between invariant polynomials, the Graph Reconstruction Conjectures (Kelly, \emph{Pacific J. Math.}, 1957; Ulam, 1960; Harary, 1964), and Forman's ``invariants of finite type'' (\emph{Adv. Math.} 2004). 
    
\end{abstract}

\setcounter{tocdepth}{1}
\tableofcontents

\section{Introduction}
In complexity lower bounds, to show that some hard function $f_{hard}$ is not in some complexity class $\mathcal{C}_{easy}$, one often seeks a property of functions that is invariant under some natural group of symmetries (such as permuting the inputs to the functions, which often doesn't change their complexity, or sometimes even taking linear combinations of inputs), shows that every function in $\mathcal{C}_{easy}$ has the property, but $f_{hard}$ does not. The Geometric Complexity Theory (GCT) program, initiated by Mulmuley \& Sohoni \cite{MS01}, takes this idea, and proposes specific ways to use algebraic geometry (seeking invariant properties defined by the vanishing of polynomials) and representation theory (the way the relevant symmetry group acts on the ring of polynomials) to help find such invariant properties. 

An analogous situation arises in the \textsc{Graph Isomorphism} (GI) problem. The GI problem is one of the most long-standing problems in complexity theory, appearing in both Cook's and Levin's original papers on $\mathsf{NP}$ \cite{Cook1971,Levin1972} (see \cite{AllenderDas} for the history). It is one of only a handful of problems in the landscape of complexity theory that has no known polynomial-time algorithm and is also not known to be $\mathsf{NP}$-complete. Graph isomorphism has numerous practical applications, including symmetry-breaking for highly symmetric integer linear programs \cite{ostrowski,kun}, computation of Feynman periods \cite{BorinskySchnetz}, and it was used for many years in circuit layout \cite{GeminiII} until it was supplanted by SAT solver-style equivalence checkers. 
The problem is as follows: given two graphs $G,H$, decide whether there is some relabeling of nodes of $G$ which makes it equal to $H$. More formally two graphs $G,H$ are isomorphic, denoted $G \cong H$, if there exists a bijection $f:V(G)\to V(H)$ such that $(u,v)\in E(G) \Leftrightarrow (f(u),f(v))\in E(H)$, and the algorithmic problem is to detect, given $G,H$, whether an isomorphism exists.

A common way to distinguish non-isomorphic graphs is to find some \emph{(isomorphism) invariant} of the graphs that differs between them: a function $f$ on labeled graphs (or adjacency matrices) such that $f(G)=f(H)$ whenever $G \cong H$. If $f(G) \neq f(H)$, then certainly $G$ and $H$ are not isomorphic, but it is possible for some invariant $f$ that $G \not\cong H$ yet $f(G)=f(H)$. Easy-to-compute invariants include the multiset of vertex degrees, or whether the graph is connected. A harder-to-compute invariant would be the number of $3$-colorings. Generally we seek invariants that are easy to compute.

In the lens of GCT, one can also phrase the problem as showing that $G$ is not in the topological closure of the $S_n$-orbit of $H$, but because the orbits (=isomorphism classes) are finite, the orbit itself is already topologically closed. In the analogy, we may think of the functions $f_{hard}$ and each $f_{easy} \in \mathcal{C}_{easy}$ as analogous to two graphs we want to distinguish. When $\mathcal{C}_{easy}$ has a complete problem the analogy is even tighter: if $f_{easy}$ is $\mathcal{C}_{easy}$-complete, then one is typically seeking to show that $f_{hard}$ is not in the topological closure of the $\mathrm{GL}_n$-orbit of $f_{easy}$. The fact that graph orbits are finite is one of the ways in which GI is ``simpler'' than the analogous problem of separating complexity classes. See Lysikov \& Walter \cite{Lysikov-Walter} for a recent paper developing the complexity theory of orbit and orbit closure  more generally.

In this paper, we use this analogy as a testbed for the ideas and techniques of GCT, but in the ``simpler'' setting of GI (separating pairs of non-isomorphic graphs seems easier than separating complexity classes). 

Our exploration reveals a number of interesting new connections between GI, algebraic proof complexity, the Graph Reconstruction Conjectures, and the representation-theoretic approach to separations put forward in GCT. From the point of view of GI itself, the methods we develop end up being equivalent to previously-studied methods for GI (specifically, the Weisfeiler--Leman family of algorithms, about which we say more below). We view this as a conceptual strength however: perhaps the connections developed here will let us transfer our intuition and understanding of these well-studied methods for GI to the more general lower bound methods of GCT. 

In the remainder of this introduction we motivate the approach, explain the mathematical objects involved (primarily representations of the symmetric group), and put our main results in context.

In order to apply algebraic geometry in the setting of GI, we deal with polynomial equations. Although prior work by many authors \cite{Berkholz-Grohe, Berkholz-Grohe2, Derksen, SSC, Atserias-Maneva} has also dealt with a system of polynomial equations for GI, we take a different approach. In prior works (\emph{ibid.}), they typically have variables $x_{ij}$, and polynomial equations in those variables whose vanishing enforces that the $n \times n$ matrix $X$ is a permutation matrix that gives an isomorphism between $G$ and $H$. If $G$ and $H$ are not isomorphic, then the resulting system of polynomial equations is not satisfiable, and one can use various algebraic or semi-algebraic proof systems to try to refute the satisfiability of the system of equations---equivalently, refute the existence of an isomorphism between $G$ and $H$. Several complexity measures on those algebraic proof systems turn out to be equivalent to or stronger than WL (\emph{ibid.}).

In contrast, in our setting, following the analogy with GCT, we wish to examine how to use polynomials ``on the graphs themselves'' to distinguish non-isomorphic graphs. That is, although we also have variables that we denote $x_{ij}$, for us $x_{ij}$ just gives the $(i,j)$ entry of the adjacency matrix of $G$, rather than encoding part of a potential isomorphism between $G$ and $H$. We can then ask what isomorphism-invariant properties of graphs can be expressed by polynomials in the entries of the adjacency matrix. As a simple example, $\sum_{i,j} x_{ij}$ is clearly invariant, and computes the total number of edges of a graph. (Since, again following GCT, we ultimately will wish to use the vanishing or non-vanishing of polynomials to distinguish non-isomorphic graphs, we note that the invariant $-|E(G)| + \sum_{i,j} x_{ij}$ vanishes on $G$, but not on any graph with a different number of edges.)

However, there are some invariant properties of graphs that are more naturally described by the vanishing of \emph{sets} of polynomials rather than individual invariant polynomials. For example, the invariant polynomial
\[
\sum_{i,i' \in [n], i \neq i'} \left(\sum_{j \neq i} x_{ij} - \sum_{j \neq i'} x_{i'j}\right)^2
\]
vanishes if and only if all vertices have the same out-degree (we assume throughout the paper that we work over a field of characteristic zero). But in fact how this polynomial is built is as a sum of squares, and since adjacency matrices are $\{0,1\}$-valued, all terms evaluate to real values, so a sum of squares vanishes iff the individual summands vanish. It is arguably more natural to say that a graph is out-regular if every polynomial in the set
\[
\left\{\sum_{j \neq i} x_{ij} - \sum_{j \neq i'} x_{i'j} : i,i' \in [n], i \neq i'\right\}
\]
vanishes on its adjacency matrix. (Note that this set of polynomials also has lower degree than the single polynomial above; degree is one of the relevant notions of complexity in our approach, as we will see.)

What is it that makes a given set of polynomials define an isomorphism-invariant property? To see this, we can use the standard rephrasing of GI as an ``orbit problem:'' the symmetric group $S_n$---consisting of permutations of $[n]$---acts on $n \times n$ adjacency matrices by permuting the indices: $(\pi \cdot A)_{ij} = A_{\pi(i), \pi(j)}$, and then two adjacency matrices $A,B$ correspond to isomorphic graphs iff there exists $\pi \in S_n$ such that $\pi \cdot A = B$. (If $P_{\pi}$ is the corresponding permutation matrix, then the action of $\pi$ on $A$ is given by $\pi \cdot A = P_{\pi} A P_{\pi}^\top$.) There is then a ``dual action'' of $S_n$ on the variables $x_{ij}$: $\pi \cdot x_{ij} = x_{\pi^{-1}(i), \pi^{-1}(j)}$, which has the effect of, for any polynomial $f$, $(\pi \cdot f)(A) = f(\pi^{-1} \cdot A)$. If $S$ is a set of polynomials, then the vanishing of all polynomials in $S$ defines an invariant property of graphs if (but not only if! see below) the action of $S_n$ permutes the set $S$ amongst itself: for if $\{\pi \cdot f : f \in S\} = S$, and $G \cong H$, then there is some $\pi$ such that $\pi \cdot G = H$, and we have $f(H) = f(\pi \cdot G) = (\pi^{-1} \cdot f)(G)$, which is some other function in $S$ evaluated at $G$. So all the functions in $S$ vanish on $G$ iff each function in $S$ vanishes on all graphs isomorphic to $G$. (An invariant polynomial as above is then just a singleton set $S$.) 

By a similar proof, even if $\pi \cdot S$ is just contained in the linear span of the polynomials in $S$, the vanishing of $S$ will still define an invariant property. This motivates a key definition. Since we will only work with $\{0,1\}$-adjacency matrices, we need only consider multilinear polynomials (every variable appears in each monomial with degree $\leq 1$), or equivalently, work in the quotient of the polynomial ring $\F[x_{ij} : i,j \in [n]] / \langle x_{ij}^2 - x_{ij} : i,j \in [n] \rangle$. The idea of this definition in the context of complexity classes goes back to \cite{MS01}; we follow the nomenclature and definition from \cite{GrochowUnifying}.

\begin{definition}[{Test module, separating module, cf. \cite[Def.~2.7]{GrochowUnifying}}]
A \emph{test module} (for graph isomorphism) is a vector space $T \subseteq \F[x_{ij} : i,j \in [n]] / \langle x_{ij}^2 - x_{ij} : i,j \in [n] \rangle$ that is closed under the above action of $S_n$: for all $f \in T$ and $\pi \in S_n$, we have $\pi \cdot f \in T$ as well.\footnote{We will have occasion to consider test modules contained in the un-quotiented polynomial ring $\F[x_{ij} : i,j \in [n]]$ as well, and for those the definition is the same, \emph{mutatis mutandis}. We also note that the definition of test module makes sense for individual graphs $G$ and $H$; although we will eventually be interested in asymptotic complexity, in this paper we are not concerned with, for example, families of test modules that are uniform in $n$.} 

Given two graphs $G,H$, a test module $T$ is called a \emph{separating module} for $G$ and $H$, or is said to \emph{separate $G$ from $H$}, if all $f \in T$ vanish on the adjacency matrix $A_G$, and some $f \in T$ does not vanish on $A_H$.
\end{definition}

As we have argued above, the set of graphs on which a given test module $T$ vanishes is an isomorphism-invariant set of graphs. Conversely, the set of all polynomials that vanish on the set of all graphs isomorphic to $G$ is a test module. These will be our main objects of exploration in this paper. We now come to our results.

\subsection{Main Results}
\paragraph{Degree.}
We were motivated to explore the distinguishing power of separating modules using degree as a complexity measure by the work of Berkholz and Grohe. In \cite{Berkholz-Grohe}, they showed an equivalence between degree in the proof system they dubbed monomial-PC and WL dimension. 

To state the simplest version of our first result, we introduce the notion of ``support-degree'' of a monomial in the variables $x_{ij}$: the support-degree of a monomial is the number of distinct vertices referenced in its subscripts, and the support-degree of a polynomial is the maximum support-degree of any monomial in it. By the ``census'' of $d$-vertex subgraphs, we mean the multiset of isomorphism types of $d$-vertex subgraphs.

\begin{theorem}[{Simplified version of Thm.~\ref{thm:degree-new}}]
Two graphs are distinguished by separating modules of support-degree $d$ if and only if they are distinguished by separating invariants of support-degree $\Theta(d)$ if and only if they are distinguished by the census of $\Theta(d)$-vertex subgraphs.
\end{theorem}

In particular, it follows immediately from a lower bound of Fürer \cite{Furer01} that the proof systems mentioned above are strictly weaker than $\Theta(d)$-dimensional Weisfeiler--Leman, as the subgraph census is the initial coloring of WL, and Fürer shows that a linear number of rounds of refinement are required in the worst case.

We also connect invariant degree to the (in)famous Reconstruction Conjectures \cite{kellyThesis,kelly,ulam,harary}. A graph is called $k$-reconstructible (resp., edge-reconstructible) if its isomorphism type is uniquely identified by the multiset of isomorphism types of its $k$-vertex-deleted subgraphs (resp., $k$-edge-deleted). The Reconstruction Conjecture is that all $n$-vertex graphs are $(n-1)$-reconstructible, which, remarkably, remains open. In Theorems~\ref{thm:vertex-recon} and \ref{thm:edge-recon} we show that a graph is $(n-k)$-reconstructible if and only if it is identified by invariants of support-degree $k$ (and analogously for edge-reconstructible and degree). These results were previously established implicitly, in different language, by Forman \cite{forman}; see Section~\ref{sec:related} below for a more detailed comparison.

What is it that leads to the difference between monomial-PC degree \cite{Berkholz-Grohe} and degree of separating modules, in terms of their power to distinguish graphs? 
As discussed above, in the setting of Berkholz \& Grohe, the meaning of their variable  $x_{ij}$ is that vertex $i$ in graph $G$ maps to vertex $j$ in graph $H$. Therefore, a degree $k$ polynomial allows them to reason about possible partial isomorphisms from $k$ vertices of $G$ to $k$ vertices of $H$. In contrast, in our setting, the variables are $x_{ij}$ which takes value 1 if there is an edge from node $i$ to node $j$ in graph $G$. Therefore degree $k$ polynomials in our setting only let us reason about $k$-edge subgraphs of $G$, and separately $k$-edge subgraphs of $H$, but not directly the relations between them. We will see next that when we measure complexity by symmetric algebraic circuit size instead of degree, separating modules do become powerful enough to simulate WL.

\paragraph{Symmetric circuit size.}
As context for our next result, we first recall the Weisfeiler--Leman family of algorithms. The Weisfeiler--Leman (WL) algorithm for \textsc{Graph Isomorphism} was first formulated in 1968 by Weisfeiler and Leman \cite{Weisfeiler-Leman}. The $k$-dimensional WL algorithm works by coloring $k$-tuples of vertices according to the isomorphism type of their induced subgraph, and then iteratively refining the coloring based on the colors of ``neighboring'' $k$-tuples; this algorithm can be implemented in $n^{O(k)}$ time. However, Cai, Fürer, and Immmerman \cite{CFI92} showed that there are pairs of $n$-vertex graphs that can only be separated by $\Omega(n)$-dimensional WL (which would take $n^{\Omega(n)}$ time, more than the $n!$ of the trivial algorithm). Nonetheless, WL has been a useful subroutine in many other algorithms for GI, notably in Babai's quasi-polynomial-time algorithm \cite{Babai16}. WL has many characterizations in terms of coloring, first-order logical properties of graphs \cite{CFI92}, and proof complexity \cite{Atserias-Maneva, Berkholz-Grohe, Berkholz-Grohe2, Toran-Worz}.

Motivated by the results of Dawar and Wilsenach, we then look at symmetric circuit size as a complexity measure on separating modules. In our setting, a symmetric circuit in the variables $x_{ij}$ is a circuit where the action of $S_n$ on the vertices $i,j,\dotsc$ extends uniquely to an action on the underlying labeled DAG of the circuit by automorphisms (labeled by variables $x_{ij}$ for input gates, or by gate type for other gates). 
In \cite{Dawar-Wilsenach20}, they show that WL can simulate the computation of symmetric algebraic circuits. To relate separating modules to symmetric algebraic circuits, we define multi-output circuits which compute a spanning set of polynomials of a separating module. We then show that $k$-WL is equivalent in distinguishing power to separating modules of symmetric circuit size $n^{\Theta(k)}$:

\begin{theorem}[Simplified version of Cor.~\ref{col:equivalence}] \label{thm:equivalence}
    Two graphs are distinguished by separating modules (resp., separating invariants) computed by symmetric algebraic circuits of size $n^{\Theta(k)}$ if and only if they are distinguished by $\Theta(k)$-WL.\footnote{In more detail, this equivalence holds for all constant $k$; as $k$ approaches $n^c$ the factors in the big-Theta depend on $c$ in a way we make explicit in Cor.~\ref{col:equivalence}. For example, for $k$ up to $\sqrt{n}/2$, we still get an equivalence with a factor of $2$ in the big-Theta.}
\end{theorem}

In Corollary~\ref{col:equivalence}, in addition to more details, we also relate the depth of the circuits to the number of rounds of $k$-WL. An anonymous reviewer suggested that the same result might be provable using homomorphism-counting polynomials; we do not pursue this here, but see Remark~\ref{rmk:hom-count} for more details.

Given that WL is essentially a proof system for graph non-isomorphism, that is equivalent to various algebraic proof systems such as monomial-PC \cite{Berkholz-Grohe}, it is natural to wonder about algebraic proof systems based on symmetric algebraic circuits. This was studied recently by Dawar, Grädel, Kullmann, and Pago \cite{DawarIPS}. Similar to previous works, their starting point is the system of equations encoding all possible isomorphisms, rather than in our setting where our polynomials encode properties of individual (isomorphism classes of) graphs. 

\paragraph{Complexity measures on the symmetric group and beyond.}
Dafni, Filmus, Lifshitz, Lindzey, and Vinyals \cite{DFLLV21} generalized many standard complexity measures from Boolean functions to functions on the symmetric group. For functions on the symmetric group, they showed that degree, approximate degree, decision tree complexity, sensitivity, block sensitivity, and others were all polynomially related. Previously, Ellis, Friedgut, and Pilpel \cite{EFP11} had shown that the degree of such functions was essentially equivalent to a combinatorial parameter of the representations of $S_n$ on which those functions were supported. Namely, the (irreducible) representations of $S_n$ are naturally indexed by the integer partitions of $n$, $\lambda_1 \geq \lambda_2 \geq \dotsb \geq \lambda_n$ such that $\sum_i \lambda_i = n$. For a given function on $S_n$, one can look at vector space spanned by its $S_n$-orbit, and see which partitions $\lambda$ arise when breaking that vector space into irreducible representations. They showed that that $n-\lambda_1$ is essentially equivalent to degree, hence by \cite{DFLLV21} also polynomially related to the above complexity measures.

An innovative side benefit of our study is Proposition~\ref{prop:alg-circuit}, in which we generalize part of the equivalence above---between degree and the representation-theoretic parameter $n-\lambda_1$---from functions on the symmetric group itself, to polynomial functions on any space that is acted on by the symmetric group (including, but not at all limited to, the space of $n \times n$ adjacency matrices). Here, we do not quite get an equivalence, but we get relationships in both directions. Another novel aspect of this relationship is that we bring the notion of support size from the study of WL into the orbit of this circle of equivalences:

\begin{proposition}[Simplified version of Prop.~\ref{prop:alg-circuit}] Suppose $S_n$ acts on a set $\Omega$, and let $R$ be the polynomial ring on variables indexed by $\Omega$. Then irreducible $S_n$-representations in $R$ of degree $\leq d$ and indexed by partitions $\lambda$ are computed by symmetric algebraic circuits of support size at most $\max\{n-\lambda_1, O(d)\}$.
Conversely, if an irreducible $S_n$-representation in $R$ is computed by symmetric algebraic circuits of support size $k$, then the representation it spans has $k \geq n - \lambda_1$.
\end{proposition}

\paragraph{Multiplicity obstructions.}
In the first paper on GCT \cite{MS01}, they propose separating complexity classes by using only ``representation-theoretic multiplicities'' instead of separating modules or polynomials. To explain this approach in our setting,  two test modules are said to be ``abstractly isomorphic'' if $S_n$ ``acts on them in the same way,'' regardless of which actual polynomials they contain. More specifically, two test modules $T,T'$ are said to be \emph{equivalent} if there is a bijective linear map $f\colon T \to T'$ such that $f(\pi u) = \pi f(u)$ for all $u \in T, \pi \in S_n$. For any given test module $T$, one can then ask how many test modules \emph{equivalent to $T$} vanish on a given graph $G$, and this is called the multiplicity of $T$ in the vanishing ideal of $G$'s orbit. If there is some abstract equivalence class of test modules that has a different multiplicity for $G$ than for $H$, then certainly they are not in the same orbit, hence are not isomorphic. Such an equivalence class of test modules is called a \emph{multiplicity obstruction}, as its existence obstructs the existence of an isomorphism. The multiplicity approach to separating graphs thus asks whether the number of vanishing separating modules of each equivalence-class is the same between two graphs.

In this setting, we get a clean intrinsic/combinatorial characterization of what these representation-theoretic multiplicities tell us about graph isomorphism. This is the first characterization of this kind we are aware of in a GCT-style multiplicity approach to separating orbit( closure)s. To explain this result, we recall that the \emph{cycle type} of a permutation is the multiset of sizes of its cycles; given a group $\Gamma \leq S_n$, our characterization is in terms of the number of elements of each cycle type in $\Gamma$:

\begin{theorem}[=Cor.~\ref{cor:multiplicity}]
Two $n$-vertex graphs $G$ and $H$ admit a multiplicity obstruction if and only if the counts of cycle types of elements in $\Aut(G) \leq S_n$ and $\Aut(H) \leq S_n$ are distinct.
\end{theorem}

In particular, since most graphs have trivial automorphism group, most graphs are not distinguished by multiplicity obstructions (Obs.~\ref{obs:multiplicity-nogo}). More generally, if two $n$-vertex graphs have automorphism groups that are conjugate in $S_n$, then they are indistinguishable by multiplicity obstructions. We show in Proposition~\ref{prop:nonisocycleindex} that there are infinitely many pairs of graphs $G_n, G'_n$ such that $\Aut(G_n) \not\cong \Aut(G'_n)$, but they nonetheless have the same cycle indices, hence are indistinguishable by multiplicity obstructions.

In light of the preceding result, we find it natural to ask whether there are interesting families of graphs---with interesting automorphism groups---that are distinguished by multiplicity obstructions. One such possible family of graphs is suggested again by analogy with GCT. Part of the philosophy behind GCT is that the representation-theoretic, and especially multiplicity-based, approach is expected to be most helpful for objects (polynomials, tensors, graphs, etc.) that are \emph{characterized by their symmetries}, a notion introduced in \cite[Sec.~4]{MS01}. If $V$ is a vector space acted on linearly by a group $G$ (i.\,e., a representation of $G$), a (nonzero) point $v \in V$ is said to be \emph{characterized by its symmetries} if for all points $u \in V$, if $\stab_G(v) \subseteq \stab_G(u)$---that is, for all $g \in G$, if $g v = v$ then $gu = u$---then $u$ must be a scalar multiple of $v$. In GCT, notable examples of symmetry-characterized polynomials include the determinant and permanent \cite[Props.~4.2 and 4.5]{MS01}, and since then many other polynomials of complexity-theoretic interest have been shown to be characterized by their symmetries (e.\,g., \cite{CKW, GrochowThesis, GrochowLieAlgebraIso, RRR}). 

But a graph and its complement have the same automorphism group, so for a graph $G$ to be ``characterized by its symmetries'' we mean that for any graph $H$ on the same number of vertices, if $\Aut(G) \leq \Aut(H)$ (as subgroups of $S_n$), then $H$ must be isomorphic to either $G$ or its complement (or the empty or complete graphs, since every permutation group is a subgroup of their automorphism groups; these play a role much like the zero vector in a vector space).\footnote{Let $V$ be the vector space of adjacency matrices, and $J$ be the adjacency matrix of the complete graph, and $U := V / \langle J \rangle$ the quotient vector space. Then in $U$, a graph and its complement are scalar multiples of one another, for we have: $A_{G^c} = J - A_G \equiv -A_G \pmod{J}$. And among simple graphs, a graph is only a nonzero scalar multiple of itself and its complement. Furthermore, the empty graph and the complete graph both correspond to the zero vector in $U$. So in the space $U$, we can use the now-standard definition of symmetry-characterization from GCT mentioned above, and see that it agrees with the natural notion that one might want for graphs.}
 It is an interesting question whether multiplicity obstructions alone can distinguish all pairs of non-isomorphic symmetry-characterized graphs. In Prop.~\ref{prop:FI} we make a small inroad into this question: using the theory of representation stability and FI-modules \cite{ChurchFarb,CEF}, we show that \emph{constant-degree} multiplicity obstructions can distinguish at most constantly many isomorphism types. The more general question of unbounded-degree multiplicity obstructions remains open; see Section~\ref{sec:future}. 

\paragraph{Occurrences obstructions.} 
In \cite{MS01}, one of the original conjectures (since disproven in the context of the Permanent vs Determinant Conjecture by Bürgisser, Ikenmeyer, and Panova \cite{BIP}), was to not only use multiplicity obstructions to separate complexity classes, but to use so-called ``occurrence obstructions,'' which are multiplicity obstructions in which the multiplicity is \emph{zero} for one class and \emph{non-zero} for the other. In the context of GCT, Dörfler, Ikenmeyer and Panova \cite{DIP} showed that there are settings in algebraic complexity for which multiplicity obstructions are strictly more powerful than occurrence obstructions for separations. We show that the same is true for graphs:

\begin{proposition}[Simplified version of Prop.~\ref{prop:mult-vs-occurrence}]
There are infinitely many pairs of graphs that are distinguished by multiplicity obstructions but not by occurrence obstructions.
\end{proposition}

\begin{remark}[Different types of graphs]
(Directed vs undirected, self-loops) We believe all of our results work, \emph{mutatis mutandis}, in the setting of graphs without multi-edges, that are either directed or undirected, and whether or not they have self-loops. 

(Multi-graphs and weighted graphs) If we allow multiple parallel edges, or more generally weighted graphs with weights in a field (still assuming characteristic zero), many of our results still go through with appropriate changes. This uses the fact that representations of $S_n$ over any field of characteristic zero are in fact defined over the rational numbers. One notable exception is that for the non-trivial equivalence between separating modules and separating invariants, we use the fact that a sum of squares (over, say, the rationals) vanishes if and only if all its summands vanish individually. The latter condition is precisely one of the equivalent characterizations of the notion of \emph{formally real field}. However, in our proof we find polynomials with integer coefficients (regardless of what characteristic-zero field we started over), so the only additional assumption needed here is that edge weights lie in a formally real field. 

(Vertex-colored graphs) In the case of vertex-colored graphs and color-respecting isomorphisms, we sketch what we believe are the minor but nontrivial needed changes for our results to go through. The ambient group changes from $S_n$ to $S_{n_1} \times \dotsb \times S_{n_k}$, where there are $n_i$ vertices of color $i$ for each $i$. It is still the case that the characters of such groups are integer-valued. With appropriate modifications to the notion of support (Def.~\ref{def:support})---viz., a support is now a tuple of subsets $X_i \subseteq [n_i]$ for $i=1,\dotsc,k$---we believe essentially all of our results go through. One slightly more complicated change is to Prop.~\ref{prop:alg-circuit}; the irreps of $S_{n_1} \times \dotsb \times S_{n_k}$ are tensor products of irreps of the $S_{n_i}$ factors, and for an irrep indexed by a $k$-tuple of partitions $\lambda^{(1)}, \dotsc, \lambda^{(k)}$, the support size now corresponds to the $k$-tuple $(n_1 - \lambda^{(1)}_1, n_2 - \lambda^{(2)}_1, \dotsc, n_k - \lambda^{(k)}_1)$. Finally, the Grohe--Verbitsky circuits and our modified versions of them in Sec.~\ref{sec:sym-ckt-size} must change their initial coloring to take into account the colors of the vertices (but the iteration steps of those circuits remain the same).
\end{remark}

\subsection{Related work} \label{sec:related}
As far as we are aware, ours is the first study using separating modules for graph isomorphism. The only other paper we are aware of that even uses representation theory in the context of GI is Derksen \cite{Derksen}, which has quite a different flavor to what we pursue here. Indeed, as with the other prior work, its starting point is a system of equations whose solutions encode all possible isomorphisms, rather than looking at (sets of) polynomials on a single graph at a time.

\paragraph{Related work on invariants.}
We are aware of three lines of research using what we would call separating (polynomial) \emph{invariants} for graph isomorphism. The first of these three goes back implicitly to Kocay \cite{kocay}, who considered an ``algebra of subgraphs'' based on the counts of induced or general subgraphs. However, Kocay seems not to have formulated his work in terms of invariant polynomials, and the identities between such numerical invariants he considers are all linear.

The second, and closer to our work, is that of Pouzet \cite{pouzet1,pouzet2}, continued more recently with Thi{\'e}ry \cite{thiery, PouzetThieryInvariant}. In their work they consider polynomial invariants of \emph{weighted} adjacency matrices, that is, arbitrary matrices over a field of characteristic zero (rather than just $\{0,1\}$ matrices or unweighted graphs as we consider). Effectively, they are working in the (infinte-dimensional) ring $\F[x_{ij}]$ whereas we (mostly) work in the (finite-dimensional) ring $\F[x_{ij}] / \langle x_{ij}^2 - x_{ij} \rangle$. The invariant ring in that setting is, expectedly, more complicated, having a generating set in terms of counting multi-graphs (with non-negative integer weights on the edges) rather than in terms of counting ordinary graphs (as in our setting). The focus on weighted graphs makes their work connected to weighted/algebraic generalizations of the Reconstruction Conjectures, but not (as far as we are aware) to the original (unweighted) Reconstruction Conjectures.

The third work on invariants, and closest to our work, is that of Forman. Forman \cite{forman} introduced what he called ``graph invariants of (edge- or vertex-)finite type'', which had an ``order'' $k$, and used these to prove a number of results on the Graph Reconstruction Conjectures. Forman's ``invariants of finite type'' are isomorphism-invariant functions on the (infinite) set of all isomorphism types of finite simple graphs satisfying certain inclusion-exclusion-like identities. By  Observation~\ref{obs:invariants} and Forman's Theorem 3.3, Forman's invariants of edge-finite-type and order $k$, applied to $n$-vertex graphs, are precisely given by invariants (in the usual polynomial sense studied in our paper) of degree $k$. Similarly, by Observation~\ref{obs:invariants} and Forman's Theorem 10.2, Forman's invariants of vertex-finite-type and order $k$, applied to $n$-vertex graphs, are precisely given by invariants in our sense of what we call ``support-degree'' $k$ (which counts the number of vertices involved in a monomial rather than the number of variables=edges). Forman's Theorems 4.6 and 11.4 are essentially equivalent, via the preceding correspondence, to our connections between invariant polynomials and the Reconstruction Conjectures. Once we get to the point of Observation~\ref{obs:invariants} and Forman's Theorems 3.3 and 10.2, our proofs of the connections with the Reconstruction Conjectures are essentially the same as his, though we came to them independently. 

We also remark that some results in Forman's paper may have shorter proofs in the language of invariant polynomials (indeed, Observation~\ref{obs:invariants} is a pretty trivial observation, whereas Forman's Theorems 3.3 and 10.2 take some work), while the opposite may be true for other results. 

Technically, Forman also deals with invariants on the infinite set of all graphs---similar to Kocay (\emph{ibid.})---whereas we focus on invariant polynomials on graphs with $n$ vertices (asymptotically as $n \to \infty$, but for each $n$ we have a finitely generated algebra). Translating Forman's results more faithfully into the language of polynomials seems to require something like invariant polynomials on graphs with infinitely many vertices (analogous to the ring of symmetric functions on infinitely many variables) or FI-rings \cite{ChurchFarb, CEF}; we hope to pursue such development in future work.

\paragraph{Related work on Weisfeiler--Leman, symmetric algebraic circuits, and proof complexity.}
Finally, we discuss several works related to WL and proof complexity that are somewhat related but quite distinct from our work. 

Berkholz and Grohe \cite{Berkholz-Grohe} introduce an algebraic proof system they name monomial-PC, and show that its distinguishing power is equivalent to WL. In particular, they start by setting up a system of polynomials that encode all possible isomorphisms between two graphs and set those polynomials to zero. They then use the polynomial calculus proof system on that system of polynomials while restricting certain intermediate steps. They show that if a refutation exists, then one exists where the maximum degree of any equation derived in that refutation is equal to the WL dimension required to distinguish the two graphs. There are several other algebraic and semi-algebraic proof systems that are equivalent to WL when applied to this same system of polynomials (or closely related systems of inequations) \cite{SSC, Atserias-Maneva, Berkholz-Grohe2}.

Toran and Wörz \cite{Toran-Worz} work in the setting of Boolean proof complexity. They use the Boolean version of essentially the same system of equations as in \cite{Berkholz-Grohe}, and apply the Narrow Resolution proof system to refute the system of equations. They show that this proof system, under the complexity measure of narrow resolution depth and width, is equivalent in distinguishing power to WL. It is possible that there is a direct connection between symmetric circuit size and support size in our setting and narrow resolution depth and width. Indeed, one can develop a Boolean analogue of separating modules in the setting of CNFs, consisting of a set of Boolean formulas that is set-wise invariant under the action of $S_n$. It might be interesting to explore this potential connection, but we do not do so in our paper. 

Anderson \& Dawar \cite{Anderson-Dawar} show that, for any first-order structure, if a property of that structure can be computed by symmetric Boolean (threshold) circuits, then the circuit can be simulated by first-order formulas (or equivalently, the property determined by the circuit can be expressed in first-order logic). Their proof preserves circuit size as number of variables and circuit depth as number of quantifiers. This allowed us to show equivalence between symmetric algebraic circuits and the first order formulas for graph isomorphism.

Dawar \& Wilsenach \cite{Dawar-Wilsenach20} show that WL can be simulated by symmetric algebraic circuits. In their paper, they rely on multiple background results proved thoroughly in \cite{Dawar-Wilsenach22} and \cite{Anderson-Dawar}. The former of those contains results for rank logic which we do not explore in this paper. We prove the converse of their results for symmetric algebraic circuits.

Theorem~\ref{thm:equivalence} (=Corollary~\ref{col:equivalence}) is connected to several other results in the literature on subgraph counts. Say that a function $f$ on isomorphism types of graphs is $k$-WL invariant if whenever $G$ and $G'$ are $k$-WL indistinguishable, $f(G)=f(G')$. Nueun \cite{neuen} proves that the counts of $H$ as a subgraph is $k$-WL-invariant if and only if $k \geq hdtw(H)$, where $hdtw$ denotes the hereditary treewidth (we refer to \emph{ibid.} for definitions). This count is captured by an invariant polynomial we denote $S_H$ (see Section~\ref{sec:invariants1}). Consistent with Corollary~\ref{col:equivalence} and Neuen's aforementioned result, Dawar, Pago, and Seppelt \cite{DPS26} proved that the invariant polynomial $S_H$ can be computed by symmetric algebraic circuits of size $n^{O(hdtw(H))}$ (combine Theorems~1.1 and 7.1 of the arXiv version of their paper).\footnote{Technically their results are for bipartite graphs and bipartition-respecting homomorphisms/subgraphs, but they directly transfer to the non-bipartite case.}

\subsection{Future directions and open questions} \label{sec:future}
The connections we develop in this paper also raise a number of intriguing new questions, some of which we raise here. We start by highlighting a question we raised above:

\begin{open}
Is every pair of graphs that are characterized by their symmetries distinguished by multiplicity obstructions? Is every graph characterized by its symmetries also characterized by its multiplicities?
\end{open}

 If a symmetry-characterized graph is not trivial (complete or empty graph) nor a star graph (complete bipartite graph $K_{n,1}$) nor its complement, then its automorphism group must have ``permutation rank 3'' (we thank Babai for this observation): the set of orbitals (=orbits on the pairs $V \times V$) must be precisely the diagonal $\{(v,v) : v \in V\}$, the edges, and the non-edges. As the rank-3 permutation groups have a complete classification \cite{Liebeck}, perhaps that classification would inspire a positive answer to the following question:

\begin{open}
Can isomorphism of symmetry-characterized graphs (to one another and/or to arbitrary graphs) be decided in polynomial time? \end{open}

The above question would be interesting to resolve even under the promise that (one or both of) the graphs are symmetry-characterized. Which also raises the question of: What is the complexity of deciding whether a graph is symmetry-characterized? (The analogous question for polynomials was raised in \cite[Open Question~3.4.8]{GrochowThesis}, and is still open as far as we are aware.)

We show, using one algorithm, that finding separating modules and computing multiplicities up to degree $d$ can be done in $n^{O(d)}$ time (Theorems~\ref{thm:algorithm} and \ref{thm:algorithm-multiplicities}). 

\begin{open}
Is there an algorithm to find a separating module in degree $\leq d$ faster than $n^{O(d)}$ time? Is there an algorithm to find a multiplicity obstruction in degree $\leq d$ faster than $n^{O(d)}$ time?
\end{open}

It is not clear to us which of these problems (if either) should be easier: on the one hand, finding a separating module requires just finding a single separating module, not computing how many of that equivalence class vanish. On the other hand, perhaps there are representation-theoretic tricks that allow one to compute a multiplicity obstruction without explicitly finding the underlying separating modules involved. It would also be interesting if there were efficient reductions in either direction between these two problems.

In addition to the result of Dawar, Pago, and Seppelt mentioned above \cite{DPS26}---that the count of $H$-subgraphs of $G$ (the value of $S_H(G)$) can be computed by symmetric algebraic circuits of size $n^{O(hdtw(H))}$---this same count can be computed in time $f(H)|G|^{O(tw(H))}$, where $tw$ denotes the treewidth and $f$ is some function that depends only on $H$ \cite{ColorCoding}. (Marx \cite{marx} showed that this essentially cannot be improved to $f(H)|G|^{o(tw(H)/\log tw(H))}$ unless the Exponential Time Hypothesis is false.) We ask whether the same holds for symmetric algebraic circuits:

\begin{open}
Can $S_H$ be computed symbolically by (symmetric) algebraic circuits of size $f(H)|G|^{O(tw(H))}$?
\end{open}

Finally, we hope to make it the subject of future work to apply these ideas in the setting of \textsc{Tensor Isomorphism} (e.\,g., \cite{GrochowQiaoTI1}). There, as in the case of GCT more generally, orbits need no longer be closed, so we are not guaranteed that they are separated by invariants at all, and the use of separating modules in that setting, rather than merely invariants, may be more of a necessity than for \textsc{Graph Isomorphism}. As with our study of \textsc{GI}, applying these ideas to \textsc{Tensor Isomorphism} has the possibility of revealing new insights both about \textsc{Tensor Isomorphism} itself, and about the representation-theoretic approach to separating orbits (and orbit closures) suggested by GCT.  

\subsection{Paper outline}
In Sec.~\ref{sec:prelim} we give preliminaries. In Sec.~\ref{sec:degree} we prove that separating modules of degree $d$ can be computed in $n^{O(d)}$ time, their equivalence with invariants of low-degree, and their equivalence with the subgraph census of small subgraphs (=the initial coloring of low-dimensional WL). In Sec.~\ref{sec:support} we connect symmetric circuit size of $S_n$-modules to other complexity measures on the symmetric group \cite{DFLLV21}.
In Sec.~\ref{sec:sym-ckt-size} we prove the equivalence between separating modules of symmetric circuit size $n^{O(k)}$ and $O(k)$-WL. In Sec.~\ref{sec:multiplicity} we prove our results on the multiplicity approach---including limitations of multiplicity and occurrence obstruction, and the proof that multiplicity obstructions are in general stronger than occurrence obstructions. Finally, in Sec.~\ref{sec:invariant} we give the connections between invariant polynomials and the Reconstruction Conjectures.

\section{Preliminaries} \label{sec:prelim}

We will need the following basic lemma.
\begin{lemma}\label{lem:binomial}
    For any $1/\log_2 n \leq \varepsilon < 1$ and $s<\varepsilon n^{1-\varepsilon}$, we have $ n^s <\binom{n}{r}$ whenever $s/\varepsilon \leq r \leq n/2$.
\end{lemma}

By setting $\varepsilon$ as an appropriate function of $n$, this lets us extend the range of $s$ to which the lemma applies almost all the way up to $\Theta(n)$: the best possible is when $\varepsilon = 1/\log_2 n$, for then the lemma applies for $s$ up to $n / (2\log_2 n)$, and $s/\varepsilon = n/2$, so it works only for exactly one value of $r$, namely $s/\varepsilon = r = n/2$. (And this explains the lower bound on $\varepsilon$: when $\varepsilon < 1/\log_2 n$, the range of $r$ to which the lemma applies becomes empty.)

\begin{proof}
    When $r\leq \frac{n}{2}$, $(\frac{n}{r})^r\le\binom{n}{r}$. Thus, it suffices to show that $n^s<(\frac{n}{r})^r$. Since $\binom{n}{r}$ is an increasing function of $r$ for $r \leq n/2$, we prove the result for $r=s/\varepsilon$:
    \begin{align*}
        n^s<\left(\frac{n}{s/\varepsilon}\right)^{s/\varepsilon} & \iff (s/\varepsilon)^{s/\varepsilon} <n^{s(1/\varepsilon - 1)} \\
        & \iff (s/\varepsilon)^{1/\varepsilon}<n^{(1-\varepsilon)/\varepsilon}, 
    \end{align*}
    and the latter holds iff $s < \varepsilon n^{1-\varepsilon}$, completing the proof of the lemma.
\end{proof}

We will use the following group theoretic definition throughout this paper.
\begin{definition}
    The stabilizer subgroup of $G$ with respect to $x$ is the set of all elements in $G$ that fix $x$, i.\,e. $\stab_G(x)=\{g\in G:gx=x\}$.
\end{definition}

\subsection{Algebraic circuits}
By an algebraic circuit over a ring $R$ (which, in this paper, we will typically take to be either the integers or a field) we mean a directed acyclic graph where each source node is labeled either by a variable or by a constant from $R$, and each non-source node is labeled either by $+$ or $\times$, which we refer to as addition gates and multiplication gates, respectively. Each gate in an algebraic circuit naturally computes a polynomial over $R$ in the variables occurring on its source nodes, which we refer to as input gates.
By the \emph{size} of an algebraic circuit we mean its number of gates.

\subsection{Representation theory} \label{sec:prelim:rep}
We use some basic definitions from any standard representation theory textbook, e.\,g., \cite{Fulton-Harris,Dummit-Foote,James-Kerber,Sagan}.

\paragraph{Group actions and $\Gamma$-sets.} If $\Gamma$ is a group, a $\Gamma$-set is a set $X$ together with an action of $\Gamma$ on $X$. If $\Sigma \leq \Gamma$ is a subgroup, then $\Gamma$ acts on the set of cosets $\Gamma/\Sigma := \{g\Sigma : g \in \Gamma\}$ by left multiplication. 

A function $f \colon X \to Y$ between two $\Gamma$-sets is $\Gamma$-equivariant if $f(\gamma x) = \gamma \cdot (f(x))$ for all $x \in x$, where on the left-hand side $\gamma x$ is the action of $\Gamma$ on $X$, and on the right-hand side $\gamma \cdot (f(x))$ is the action of $\Gamma$ on $Y$. 
Two $\Gamma$-sets $X,Y$ are said to be isomorphic if there is a $\Gamma$-equivariant bijection between them. 

If $\Sigma \leq \Gamma$ is a subgroup, then the set of cosets $\Gamma/\Sigma$ is naturally a $\Gamma$-set under the action $\gamma \cdot (\gamma' \Sigma) = (\gamma \gamma') \Sigma$. 

The \emph{orbit} of $x \in X$ is $\Gamma x := \{\gamma x : \gamma \in \Gamma$. The \emph{stabilizer} of $x$ is $\stab_\Gamma(x) := \{\gamma \in \Gamma : \gamma x = x\}$. The stabilizer is always a subgroup, and the Orbit-Stabilizer Theorem tells us that $\Gamma x \cong \Gamma/\stab_{\Gamma}(x)$ (isomorphism of $\Gamma$-sets).

A $\Gamma$-set $X$ is said to be a \emph{transitive} if it consists of a single orbit. Since every $\Gamma$-set is a disjoint union of its orbits, the Orbit-Stabilizer Theorem implies that every $\Gamma$-set is isomorphic to a disjoint union of coset spaces.

There is a $\Gamma$-equivariant function $\Gamma / \Sigma \to \Gamma / \Delta$ if and only if $\Sigma$ is conjugate to a subgroup of $\Delta$. Any $\Gamma$-equivariant homomorphism from a non-empty $\Gamma$-set to a transitive $\Gamma$-set $X \cong \Gamma / \Sigma$ is automatically surjective (exercise).

\paragraph{Representations.}
If a $\Gamma$-set $V$ is a vector space, and the action respects the linear structure---that is, $\gamma \cdot (\alpha u+\beta v) = \alpha (\gamma \cdot u) + \beta (\gamma \cdot v)$ for all $\gamma \in \Gamma$, $\alpha,\beta \in \F$, $u,v \in V$---we call $V$ a $\Gamma$-module. In this case, the action of $\Gamma$ on $V$ corresponds to a group homomorphism $\rho \colon \Gamma \to \GL(V)$, that is, $\rho(\gamma_1 \gamma_2) = \rho(\gamma_1) \rho(\gamma_2)$. The homomorphism $\rho$ is called a \emph{(linear) representation} of $\Gamma$ ``on $V$'' or ``supported on $V$''. We sometimes use the terminology interchangeably, and when understood from context we may refer to any of $(\rho,V)$, $\rho$, or $V$ as a representation.

If $\dim_\F V=k<\infty$, by choosing a basis for $V$ we can use the 
    alternate notation $\rho:G\to \GL_k(\F)$. 

If $U,V$ are two $\Gamma$-modules, a \emph{$\Gamma$-homomorphism} between them is a linear function $U \to V$ that is also $\Gamma$-equivariant (as for $\Gamma$-sets above). We denote the set of such $\Gamma$-homomorphisms $\Hom_\Gamma(U,V)$; this set is readily checked to be a vector subspace of the space of all linear functions $U \to V$.

Given two representations $(\rho,U)$, $(\tau,V)$, we say they are 
\emph{equivalent} (or that the $\Gamma$-modules $U,V$ are \emph{isomorphic}) if their homomorphisms are conjugate, i.\,e., there exists some invertible matrix $A$ such that $\rho(g)=A\tau(g)A^{-1}$ for all $g\in G$. In other words, representations are equivalent up to change of basis in the vector space. Equivalently, two $\Gamma$-modules are isomorphic if there is a bijective $\Gamma$-homomorphism between them.

Given a representation $(\rho,V)$, a \emph{sub-representation} is a subspace $W$ of $V$ that is invariant under the action of $G$, viz., $\rho(g)W\subseteq W$ for all $g\in G$, and in this case we call $W$ a $\Gamma$-submodule of $V$. 
    An \emph{irreducible representation} or \emph{irrep} is a representation that has no nonzero proper sub-representation.\footnote{The corresponding term for $\Gamma$-modules is ``simple,'' but we will not use that terminology here...for simplicity.}

Given two representations $(\rho,U)$, $(\tau,V)$ their direct sum $(\rho \oplus \tau, U \oplus V)$ is defined by $(\rho \oplus \tau)(\gamma)(u+v) = \rho(\gamma)u + \tau(\gamma)v$, for $u \in U, v \in V, \gamma \in \Gamma$. Equivalently if we choose a basis of $U \oplus V$ that consists first of a basis for $U$ followed by a basis for $V$, then the matrices for $\rho \oplus \tau$ are the block-diagonal matrices
\[
(\rho \oplus \tau)(\gamma) := \begin{pmatrix} \rho(\gamma) & 0 \\ 0 & \tau(\gamma) \end{pmatrix}.
\]
    
    Maschke's Theorem tells us that for a finite group $G$, every representation over a field of characteristic zero (or more generally coprime to $|G|$) is equivalent to a direct sum of irreps. The \emph{multiplicity} of an irreducible $G$-module $U$ in an arbitrary $G$-module $V$ is $\dim_{\F} \Hom_G(U,V)$, which is equal to the number of times representations equivalent to $U$ occur when writing $V$ as a direct sum of irreps.

    If $U,V$ are $G$-modules and $\varphi \colon U \to V$ is a surjective $G$-homomorphism, then for each irrep, its multiplicity in $V$ is at most its multiplicity in $U$. In particular, the irreps that occur in $V$ with nonzero multiplicity are a (not necessarily proper) subset of the irreps that occur in $U$ with nonzero multiplicity.

\paragraph{Isotypic components.} Given an irrep $U$ and an arbitrary $G$-module $V$, the $U$-isotypic component of $V$ is the subspace of $V$ spanned by all irreps contained in $V$ that are equivalent to $U$. Equivalently, it is the subspace spanned by the image of all $\Gamma$ homomorphisms $U \to V$. If we break $V$ into its irreducible constituents, we write $V = U_1^{\oplus m_1} \oplus \dotsb \oplus U_k^{\oplus m_k}$, where each $U_i$ is an irrep (up to equivalence), $m_i$ is the multiplicity of $U_i$ in $V$, and $U_i$ and $U_j$ are not equivalent for all $i \neq j$. Then each $U_i^{\oplus m_i}$ is an isotypic component of $V$, and these are all the isotypic components. Furthermore, the isotypic components are characteristic subspaces of $V$, in the sense that an $\Gamma$-automorphism of $V$ (a bijective $\Gamma$-homomorphism from $V$ to itself) must map each isotypic component to itself. This is one of their key uses, because the latter is not true for individual irreps inside $V$: a $\Gamma$-automorphism of $V$ may move one copy of $U_1$ contained in $V$ to a different copy of $U_1$ contained in $V$, but it cannot move $U_1$ outside of its isotypic component. Within an isotpyic component $U_1^{\oplus m_1}$, the $\Gamma$-automorphisms are isomorphic to $\GL_{m_1}(\F)$; this generalizes the choice of basis in a vector space $\F^{m_1}$ (which one can think of as a representation of the trivial group). 

\paragraph{Relationship between linear representations and actions on finite sets.}
Given a $\Gamma$-set $X$, there is an associated representation denoted $\F[X]$, which has a basis $\{e_x : x \in X\}$ that is in bijection with the set $X$, and on which $g$ acts by $g \cdot e_x := e_{g^{-1} \cdot x}$ (as in the case of the action on functions, the inverse is needed to make it a left action). Such representations are called \emph{permutation representations}.

\begin{observation} \label{obs:perm-to-linear}
Suppose $V$ is a representation of $\Gamma$ over a field $\F$ and $\Sigma \leq \Gamma$. Then 
\[
\Hom_\Gamma(\F[\Gamma/\Sigma], V) \cong \{v \in V: \Sigma \leq \stab_\Gamma(v)\},
\]
where $\cong$ denotes a bijection of sets. 
\end{observation}

\begin{proof}
Given $f \in \Hom_\Gamma(\F[\Gamma/\Sigma], V)$, let $v = f(\Sigma)$. Then for $\sigma \in \Sigma$ we have $\sigma v = \sigma f(\Sigma) = f(\sigma \Sigma) = f(\Sigma) = v$, where the middle equality follows by equivariance, and thus $\Sigma \leq \stab_\Gamma(v)$. 

Conversely, given $v \in V$ such that $\Sigma \leq \stab_\Gamma(v)$, we claim there is a unique $\Gamma$-equivariant linear map $\F[\Gamma/\Sigma] \to V$ defined by $e_{g\Sigma} \mapsto g \cdot v$, where we use $e_{g\Sigma}$ to denote the basis vector in $\F[\Gamma/\Sigma]$ corresponding to the coset $g\Sigma$. Since the vectors $e_{g\Sigma}$ corresponding to the cosets of $\Sigma$ form a basis of $\F[\Gamma/\Sigma]$ by construction, the main thing to check is that this function is well-defined. And for the latter, we need to show that if $e_{g\Sigma} = e_{g' \Sigma}$ then $g v = g'v$. We have $e_{g\Sigma} = e_{g'\Sigma} \Leftrightarrow g\Sigma = g'\Sigma \Leftrightarrow \Sigma = g^{-1} g' \Sigma \Leftrightarrow g^{-1} g' \in \Sigma$. And in that case, we have $g^{-1} g'v = v$; multiplying both sides by $g$ we get $g'v = gv$ as desired.
\end{proof}

\paragraph{Characters, and projections onto isotypic components.} For representations of finite groups $\Gamma$ over the complex numbers, the isomorphism type of the representation $\rho \colon \Gamma \to \GL_d(\mathbb{C})$, where $d$ is the dimension, is determined (up to equivalence) by its \emph{character}, which by definition the function
\[
\chi_{\rho} \colon \Gamma \to \mathbb{C} \qquad \chi_{\rho}(g) := \tr(\rho(g)).
\]
An immediate consequence of this definition is that because the trace of a matrix remains the same after conjugation, if $g,h\in \Gamma$ are in the same conjugacy class, $\chi_\rho(g)=\chi_\rho(h)$. Particularly relevant for us will be that in case of $\Gamma=S_n$, $\chi_\rho(g)$ depends only on the cycle type of $g$ since the conjugacy classes in $S_n$ are exactly determined by cycle type. 

Because there is a bijection between characters (up to equality) and representations (up to equivalence), just as we speak of the multiplicity of an irrep in another representation, if $\chi$ is the character of an irrep $V$ and $\chi'$ is the character of another representation $W$, by the multiplicity of $\chi$ in $\chi'$'' we mean the multiplicity of $V$ in $W$. Similarly, we speak of the character $\chi$ itself being irreducible when it is the character of an irrep. 

If $\chi$ is the character of a representation $(\rho,V)$, we say that $\rho$ is a representation \emph{realizing} the character $\chi$. 

Given an irreducible character $\chi$ and an arbitrary representation $(\rho,V)$, the $\chi$-isotypic subspace is the isotypic subspace of $V$ spanned by irreps realizing $\chi$. Given a representation $(\rho, V)$ of $\Gamma$, and an irreducible character $\chi$, the matrix 
\[
\pi_\chi := \frac{1}{|\Gamma|}\sum_{g \in \Gamma} \chi(g) \rho(g)
\]
is the projection onto the $\chi$-isotypic subspace. That if $V_\chi$ denotes the $\chi$-isotpyic subspace, and $V = V_\chi \oplus V'$ (where $V'$ is the direct sum of the other isotypic subspaces), then $\pi_\chi$ is the identity on $V_\chi$ and is zero on $V'$. We will not generally use the notation $\pi_\chi$, but we will use the above projection formula frequently.

\subsection{Representations of the symmetric group} \label{sec:prelim:Sn}
We use $S_n$ to denote the symmetric group, consisting of the $n!$ permutations of the set $[n] = \{1,\dotsc,n\}$. We refer the reader to standard texts \cite{Fulton-Harris,Fulton,James-Kerber,Sagan} for more details on the representation theory of the symmetric groups $S_n$.

A partition $\lambda = (\lambda_1,\dotsc,\lambda_n)$ of $n$ is a list of non-negative integers $\lambda_1 \geq \lambda_2 \geq \dotsb \geq \lambda_n$ such that $\sum_{i=1}^n \lambda_i = n$. The irreducible representations of $S_n$ are indexed in a natural way by the partitions of $n$---there is exactly one irreducible representation associated to each partition $\lambda$ that we will describe next, and all irreducible representations of $S_n$ (in characteristic zero) arise this way.

We now partially describe the irreducible representation associated to $\lambda$, known as a \emph{Specht module (of shape $\lambda$)}. We first describe a simpler, but non-irreducible representation. A \emph{row tabloid} of shape $\lambda$ is a list of subsets of $[n]$ that partition $[n]$, and such that the $i$-th subset has size exactly $\lambda_i$. The symmetric group acts transitively on the row tabloids of shape $\lambda$, and we define the row tabloid representation as the associated permutation representation.

The Specht module of shape $\lambda$ is a quotient of the row tabloid representation of shape $\lambda$; we will mostly not need the details of this construction, so we leave the details to the standard texts cited above.

\begin{fact} \label{fact:small-characters}
The character values of $S_n$ are integers (of magnitude at most $\sqrt{n!}$).
\end{fact}

We will not need the magnitude bound except in footnotes, where it can be useful to bound the bit-size of coefficients appearing in various constructions throughout the paper. We will frequently use the fact that the characters of $S_n$ are integer-valued.

\begin{proof}
The character values of $S_n$ are all integers (e.\,g., follows directly from the Frobenius character formula \cite[Frobenius Formula~4.10]{Fulton-Harris}), and since the dot product of a character with itself is $1/n!$ times the sum of the squared norms of the character values, and the dot product of an irreducible character with itself is $1$, every character value is strictly less than $\sqrt{n!}$.
\end{proof}

The \emph{defining representation} of $S_n$ is the $n$-dimensional representation with basis $e_1, \dotsc, e_n$ where $\pi \in S_n$ acts by $\pi(e_i) = e_{\pi(i)}$. 

\begin{fact} \label{fact:Sn-reps-by-degree}
If $V$ is the defining representation of $S_n$, then $V^{\otimes k}$ contains only representations of $S_n$ corresponding to partitions $\lambda$ with $\lambda_1 \geq n-k$.
\end{fact}

\begin{proof}
It is standard that the defining representation is the direct sum of a trivial representation and the so-called standard representation, corresponding to $\lambda = (n-1,1)$. Since tensor distributes over direct sum, the trivial representations here have the effect of including (not necessarily with the same multiplicity) the irreps from $V^{\otimes k'}$ with $k' \leq k$ among those in $V^{\otimes k}$. But otherwise they do not generate any ``new'' irreps that aren't already present in some tensor power of the standard representation.

For the tensor powers of the standard representation, see, for example, \cite[Prop.~1]{GC05}, which gives a fairly explicit combinatorial formula for the multiplicities of the representations appearing. It is apparent from that formula that only representations with $\lambda_1 \geq n-k$ can appear.
\end{proof}

We will use the following classical result a few times. Given a partition $\lambda_1 \geq \lambda_2 \geq \dotsb \geq \lambda_n$ with $\sum \lambda_i = n$, we call $i$ an \emph{inner corner} of $\lambda$ if $\lambda_i > \lambda_{i+1}$ (strict inequality). 

\begin{theorem}[{Branching rule for $S_n$, see, e.\,g., \cite[\S 2.8]{Sagan}}] \label{thm:branching}
Given an irrep $V_\lambda$ of $S_n$ corresponding to the partition $\lambda$ of $n$, the $S_{n-1}$ irreps occurring in $V_\lambda$ (restricted to $S_{n-1}$) are precisely those corresponding to partitions $\mu$ gotten from $\lambda$ by taking an outer corner $i$, and reducing $\lambda_i$ by $1$; that is, $\mu_j = \lambda_j$ for all $j \neq i$, and $\mu_i = \lambda_i-1$. Each such irrep occurs in $\Res^{S_n}_{S_{n-1}} V_\lambda$ with multiplicity exactly 1.
\end{theorem}

\begin{corollary} \label{cor:branching}
$\F[S_n / S_{n-k}]$ only contains irreducible representations of shape $\lambda$ with $k \geq n - \lambda_1$.
\end{corollary}

\begin{proof}
By Observation~\ref{obs:perm-to-linear}, $\F[S_n / S_{n-k}]$ contains an irrep $V$ if and only if $V$ contains an $S_{n-k}$-invariant vector. An $S_{n-k}$-invariant vector is a 1-dimensional representation corresponding to the partition $(n-k)$ (with only one part). Because the branching rule from $S_n$ down to $S_{n-1}$ (and hence, iterated down to $S_{n-k}$) can only decrease each $\lambda_i$, in order for $\lambda_1$ to be reduced to $n-k$, we must have $\lambda_1 \geq n-k$ to begin with.
\end{proof}

The reader should be aware that throughout the paper we mostly use ``$\F[X]$'' to mean a polynomial ring in the variables corresponding to the entries of the matrix $X$ (although in the preliminaries we used this same notation to indicate the permutation representation associated to a $\Gamma$-set $X$; for specific $\Gamma$-sets such as $S_n / S_{n-k}$ as an $S_n$-set, we still use the notation $\F[S_n / S_{n-k}]$ to indicate the permutation representation).

\begin{corollary} \label{cor:Sn-reps-by-degree}
Let $\F[X]$ be the polynomial ring in $n^2$ variables $x_{ij}$, arranged into an $n \times n$ matrix $X$, acted on by $S_n$ by conjugation by ordinary permutation matrices (equivalently, $\pi \cdot x_{ij} = x_{\pi^{-1}(i), \pi^{-1}(j)}$). Then among the degree-$d$ polynomials in this ring, the $S_n$ representations that occur all correspond to partitions $\lambda$ with $\lambda_1 \geq n-2d$.
\end{corollary}

\begin{proof}
Let $V$ denote the defining $n$-dimensional representation of $S_n$; then as an $S_n$-representation, the degree-1 polynomials are $V \otimes V$. Thus the degree-$d$ polynomials are the $d$-th symmetric power $\text{Sym}^d(V \otimes V)$, which is contained in $V^{\otimes 2d}$. The result now follows from Fact~\ref{fact:Sn-reps-by-degree}.
\end{proof}

We will not need the following lemma, but it gives us a precise relationship between separating modules in the full polynomial ring $\F[x_{ij} : i,j \in [n]]$ (but evaluated on $\{0,1\}$-adjacency matrices of graphs) and those in the multilinear polynomial ring $\F[x_{ij} : i,j \in [n]] / \langle x_{ij}^2 - x_{ij} : i,j \in [n] \rangle$.

\begin{lemma} \label{lem:test-modules-of-high-degree}
Let $\F[X]$ be the polynomial ring in $n^2$ variables $x_{ij}$ arranged into an $n \times n$ matrix, acted on by $S_n$ as above. Let $I_{bool} := \langle x_{ij}^2 - x_{ij} : i,j \in [n] \rangle$. Let $\pi \colon \F[X] \to \F[X] / I_{bool}$ be the natural quotient map.

If $V$ is an irreducible representation contained in $\F[X]$ that contains some $f \in V$ such that there exists a $\{0,1\}$-matrix $A$ with $f(A) \neq 0$,  then $V$ is not contained in $I_{bool}$ and $\pi(V)$ is equivalent to $V$ as $S_n$-representations.
\end{lemma}

It follows that if $V \leq \F[X]$ is any test module (not necessarily irreducible) that is not identically zero on all $\{0,1\}$ matrices, then $\pi(V) \leq \F[X] / I_{bool}$ is a non-zero test module.

\begin{proof}[Proof of Lemma~\ref{lem:test-modules-of-high-degree}]
First, $I_{bool}$ is a sub-representation of $S_n$: if $f \in I_{bool}$, then $f = \sum_{i,j} f_{ij} \cdot (x_{ij}^2 - x_{ij})$ for some polynomials $f_{ij}$ (these $f_{ij}$ will not in general be uniquely determined by $f$, but that non-uniqueness will not be an issue here). Then for $\sigma \in S_n$, since $\sigma \cdot (x_{ij}^2 - x_{ij})$ is another polynomial of the same form, namely $x_{k\ell}^2 - x_{k\ell}$ where $\sigma(k)=i$ and $\sigma(\ell)=j$, we have $\sigma \cdot f = \sum_{i,j} f_{ij} \cdot (x_{\sigma^{-1}(i), \sigma^{-1}(j)}^2 - x_{\sigma^{-1}(i), \sigma^{-1}(j)})$, so $\sigma f$ is also in $I_{bool}$.

Thus $V \cap I_{bool}$ is also a representation of $S_n$. Since $V$ is irreducible, we must have $V \cap I_{bool}$ is either $0$ or all of $V$. 

Finally, it remains to show that if the polynomials in $V$ do not all vanish identically on all $\{0,1\}$ matrices, then $V$ is not contained in $I_{bool}$. For then it follows from the above that $V \cap I_{bool} = 0$, and thus $\pi$ maps $V$ isomorphically onto its image $\pi(V)$ (isomorphism of $S_n$-representations). 

To show the latter, it suffices to show that $I_{bool}$ is a radical ideal. Recall that an ideal $I$ is \emph{radical} if $f^k \in I$ for some $k \geq 1$ implies that $f \in I$. Hilbert's Nullstellensatz then tells us that if $f$ vanishes on all $\{0,1\}$ matrices, then some power of $f$ must be in $I_{bool}$, so if we show that $I_{bool}$ is radical, it follows that any $f$ vanishes on all $\{0,1\}$ matrices if and only if it is in $I_{bool}$.

This fact follows from some standard commutative algebra, but we include the proof here for completeness.

Since each generator of $I_{bool}$ is on a different variable, we have $\F[X]/I_{bool}\cong\bigotimes_{i,j} \F[x_{ij}]/ \langle x_{ij}^2 - x_{ij} \rangle$. From the (generalized) Chinese Remainder Theorem, we have that $\F[x] / \langle x^2 - x \rangle \cong (\F[x] / \langle x \rangle) \oplus (\F[x] / \langle x-1 \rangle) \cong \F \times \F$. Thus $\F[X] / I_{bool}$ is a tensor product over $\F$ of direct products of copies of $\F$; the result is a direct product of $2^{n^2}$ copies of $\F$. 

But if $f^k \in I_{bool}$, then the image of $f$ modulo $I_{bool}$ is nilpotent $f^k \equiv 0 \pmod{I_{bool}}$. But since $\F[X] / I_{bool}$ is a direct product of fields, the only nilpotent element there is $0$ itself, thus $f \equiv 0 \pmod{I_{bool}}$, in other words $f \in I_{bool}$. This completes the proof that $I_{bool}$ is radical, and thus the proof of the lemma.
\end{proof}

\subsection{The Weisfeiler--Leman algorithm}

We will not actually need details of the WL (family of) algorithm(s), but since it plays a central role in some of our results we felt we should include at least the basic definitions.

The algorithm relies on color refinement, which is easiest to think of in terms of a hash function. Let $f$ be an injective function from finite sets to real numbers. The WL algorithm iteratively assigns each node a real number value output by $f$, i.e. a color, based on the colors of its neighbors. Following the exposition from \cite{Huang-Villar}, the simplest version of the algorithm is stated as follows:
\begin{algorithm}
    \caption{Weisfeiler--Leman}
    \label{WL}
    \begin{algorithmic}[1]
        \Require{$G=(V,E)$}
        \State Let $f:S\hookrightarrow\mathbb{R}$ \Comment{a hash function of sets}
        \State $\forall v\in V:c_v^0=f(\deg(v))$ \Comment{$c_v^t$ is the color of vertex $v$ at time $t$}
        \Repeat
        \State $\forall v\in V:c_v^t\gets f(c_v^{t-1},\{\!\!\{c_w^{t-1}\mid (v,w)\in E\}\!\!\})$ \Comment{$f$ takes the color at the previous iteration and the multiset of all neighboring colors}
        \Until $\forall v\in V: c_v^t=c_v^{t-1}$ \\
        \Return $\{\!\!\{c_v^t\mid v\in V\}\!\!\}$
    \end{algorithmic}
\end{algorithm}

$G$ and $H$ are distinguished by $WL$ if and only if $WL(G)\ne WL(H)$ where $WL(G)$ is the output of algorithm \ref{WL}.

This version of the WL algorithm is referred to as 1-dimensional WL (also sometimes ``color-refinement''). Before we continue, we must define arbitrary dimensional WL, or $k$-WL, explored by Cai, F{\"u}rer, and Immerman. See \cite{CFI92} for the history.
\begin{algorithm}
    \caption{$k$-dimensional Weisfeiler--Leman ($k\ge2$)}
    \label{kWL}
    \begin{algorithmic}[1]
        \Require{$G=(V,E)$}
        \State Let $f:S\hookrightarrow\mathbb{R}$
        \State $\forall\vec{v}\in V^k:c_{\vec{v}}^0=f(\mathrm{type}(\vec{v}))$ \Comment{$\vec{v}$ is a $k$-tuple of vertices and $\mathrm{type}$ is the isomorphism type of the induced subgraph of $\vec{v}$}
        \Repeat
        \State $\forall\vec{v}\in V^k,w\in V:c_{\vec{v},w}^t\gets(c_{\vec{v}_{1/w}}^{t-1},\dots,c_{\vec{v}_{k/w}}^{t-1})$ \Comment{$\vec{v}_{i/w}$ is $\vec{v}$ with $w$ substituted at index $i$}
        \State $\forall\vec{v}\in V^k:c_{\vec{v}}^t\gets f(c_{\vec{v}}^{t-1},\{\!\!\{c_{\vec{v},w}^{t-1}\mid w\in V\}\!\!\})$
        \Until $\forall \vec{v}\in V: c_{\vec{v}}^t=c_{\vec{v}}^{t-1}$ \\
        \Return $\{\!\!\{c_{\vec{v}}^t\mid\vec{v}\in V\}\!\!\}$
    \end{algorithmic}
\end{algorithm}

We call the final value of $t$ in algorithm \ref{kWL}, the number of rounds that the algorithm took. See \cite[Figs.~2 and 3]{Huang-Villar} for a visual example of a pair of graphs which require dimension 2 to distinguish. $k$-WL is used in many places, most notably in Babai's quasipolynomial time graph isomorphism algorithm. 

\begin{theorem}[\cite{CFI92}]\label{logic}
    Let $\varphi$ be a first order logical statement (allowing counting) that distinguishes two graphs, i.e. is true for one and false for the other. Two graphs are distinguished by $\varphi$ using $k-1$ variables and quantifier depth $d$ if and only if $k$-WL distinguishes the graphs in $d$ rounds.
\end{theorem}\

\subsection{Simulation}
Throughout this paper, we will use the term ``simulate'' to compare distinguishing power of different techniques for graph isomorphism. To explain what we mean by this, let $\mathcal{A}_k$ be a family of algorithms for graph isomorphism with the following properties.
\begin{itemize}
    \item If $G\cong H$, then $\mathcal{A}_k(G,H)=1$ for all $k$
    \item If $G\ncong H$, then there exists a $k$ such that $\mathcal{A}_k(G,H)=0$
    \item If $\mathcal{A}_k(G,H)=0$, then $\mathcal{A}_{k+1}(G,H)=0$
\end{itemize}
Given two such families $\mathcal{A}_k,\mathcal{B}_k$, we say that $\mathcal{B}$ simulates $\mathcal{A}$ if there exists some non-negative constant $c$ such that for all $k$, if $\mathcal{A}_k(G,H)=0$ then $\mathcal{B}_{ck}(G,H)=0$.

For the families of algorithms $\mathcal{A}_k$ we consider, we may view the family $\mathcal{A}$ as a Cook--Reckhow proof system for \textsc{Graph Non-Isomorphism}. In these cases, if $\mathcal{A}_k(G,H)=0$, this generally yields a proof of non-isomorphism of size $n^{O(k)}$. In this setting, the above notion of simulation agrees with the notion of p-simulation from proof complexity \cite{CookReckhow}: one proof system $\mathcal{P}$ p-simulates another proof system $\mathcal{Q}$ if for all relevant instances, the shortest $\mathcal{P}$-proof for that instance is at most polynomially longer than the shortest $\mathcal{Q}$-proof.

\section{Degree of separating modules, invariants, and subgraph census} \label{sec:degree}
In this section, we show that degree $d$ separating modules can be found in $n^{O(d)}$ time, and are equivalent in power to degree-$\Theta(d)$ invariants. We introduce the notion of ``support-degree'' which counts how many vertices are ``involved'' in a monomial (rather than degree, which counts edges), and show that separating modules of support-degree $\Theta(d)$, invariants of support-degree $\Theta(d)$, and the census of subgraphs on $\Theta(d)$ vertices (the initial coloring of $\Theta(d)$-WL) are all equivalent. Since the algorithmic result helps motivate the exploration of degree in the first place, we start there.

\subsection{Algorithms for finding low-degree separating modules}
\begin{theorem} \label{thm:algorithm}
There is an algorithm that, given two $n$-vertex graphs $G,H$ and an integer $d$ between $0$ and $n^2$, finds a separating module of degree $\leq d$ for $G,H$ or reports that none exists, in time $n^{O(d)}$.
\end{theorem}

When we say ``computes a separating module'', we mean that the algorithm outputs a list of integer polynomials, each specified by the list of monomials and nonzero coefficients, such that those polynomials span a separating module.

We note here that even if one is working with polynomials over the complex numbers, the algorithm works in the usual Turing machine model: A consequence of the proof is that, if there is such a separating module, then there is one that has a basis consisting of polynomials with integer coefficients, where the sum of the bit-lengths of all coefficients is at most $n^{O(d)}$.

(The same result holds with the bound $2^{O(d^2)} n^{O(d)}$ when using the notion of ``support-degree'' that we introduce in Defintition~\ref{def:support-degree} below. The proof is the same, \emph{mutatis mutandis}.)

As a warm-up that involves essentially no representation theory, we start by proving the statement for invariants only:

\begin{proposition} \label{prop:compute-invariants}
There is an algorithm that, given two $n$-vertex graphs $G,H$ and an integer $d$ between $0$ and $n^2$, finds a separating \emph{invariant} of degree $\leq d$ for $G,H$ or reports that none exists, in time $n^{O(d)}$.
\end{proposition}

\begin{proof}
Let $R$ be the polynomial ring modulo the Boolean axioms, that is, $R = \F[x_{ij} : i,j \in [n]] / \langle x_{ij}^2 - x_{ij} : \forall i,j\rangle$, and let $R_{\leq d}$ be the vector subspace consisting of all (multilinear) polynomials of degree at most $d$. (In fact, our proof will work even if we replace $R$ with the un-quotiented polynomial ring $\F[x_{ij}]$, but we will not need that here.) Let $N = \binom{n}{2}$ be the total number of variables (the proof works for other variants of graphs \emph{mutatis mutandis} by changing $N$).

For each of the $\leq N^{O(d)} = n^{O(d)}$ monomials of degree $\leq d$, let $f$ be the invariant polynomial which is the sum over the $S_n$-orbit of that monomial. If $f(G) \neq f(H)$, then let $\hat{f}(X) := f(X) - f(G)$. We then have that $\hat{f}$ is an invariant of degree $\leq d$ such that $\hat{f}(G)=0$ and $\hat{f}(H) \neq 0$. 

We finish with the runtime analysis. Enumerating the monomials takes $n^{O(d)}$ time, and evaluating $f$ on $G$ and $H$ takes $n^{O(d)}$ time. The naive way of computing $f$ by summing over its $S_n$ orbit would take $n!$ time, so we employ a different approach: each degree-$d$ monomial $x_{i_1 j_1} \dotsb x_{i_d j_d}$ corresponds to a $d$-edge graph $M$ with vertex set $[n]$ and edge set $\{(i_1,j_1), \dotsc, (i_d,j_d)\}$. The sum over monomials can be computed by enumerating all labeled graphs isomorphic to $M$ and summing over the monomials corresponding to those graphs. Let $s = |\{i_1,\dotsc,i_d,j_1,\dotsc,j_d\}|$, which is in between $d$ and $2d$. Then $M$ has at least $n-s \geq n-2d$ isolated vertices. We enumerate the labeled graphs isomorphic to $M$ by first picking the $s$ non-isolated vertices---of which there are $\binom{n}{s} \leq n^{O(d)}$ many---and then, on those $s$ vertices, brute-force enumerating all $s!$ permutations applied to the vertices of $M$, maintaining the set of labeled graphs that occur. (This may seem like too much time, but if $s!$ is approaching $n!$, then $d$ must be approaching $n$, so the goal of $n^{O(d)}$ is a very generous bound.) The total runtime of this procedure is then $s! n^{O(d)} \leq 2^{s \log s} n^{O(d)} \leq 2^{2d \log(2d)} n^{O(d)} \leq n^{2d \frac{\log 2d}{\log n}} n^{O(d)} \leq n^{4d(1 + o(1))+ O(d)} \leq n^{O(d)}$, as claimed.
\end{proof}

We now come to the proof of the more general theorem for separating modules, which starts from the same basic idea but needs more (algorithmic) representation theory.

\begin{proof}[Proof of Theoerem~\ref{thm:algorithm}]
Let $R$ be the polynomial ring modulo the Boolean axioms, that is, $R = \F[x_{ij} : i,j \in [n]] / \langle x_{ij}^2 - x_{ij} : \forall i,j\rangle$, and let $R_{\leq d}$ be the vector subspace consisting of all (multilinear) polynomials of degree at most $d$. (In fact, our proof will work even if we replace $R$ with the un-quotiented polynomial ring $\F[x_{ij}]$, but we will not need that here.)

Let $\pi_1, \pi_2$ be two generators of the symmetric group $S_n$, say $\pi_1 = (12), \pi_2 = (12\dotsb n)$. We can compute the action of $\pi_i$ on (multilinear) of monomials of degree $\leq d$ to get a matrix $M_i$ where for a monomial $m$, we have $M_i(m) = \pi_i(m)$. 
That is, $M_i$ is a $D \times D$ matrix, where $D$ is the number of (multilinear) monomials of degree $\leq d$; indexing the rows and columns by monomials, we have $(M_i)_{\pi_i(m), m}=1$ for all monomials $m$, and all other entries of $M_i$ are $0$.
Since $S_n$ merely permutes monomials, the matrices $M_i$ are $\{0,1\}$-valued (in fact, permutation matrices), and they have dimension at most $\dim R_{\leq d} \leq \binom{n^2 + d}{d}$.\footnote{This upper bound is the number of all monomials of degree $d$, not only multilinear, so that one can see the proof still works in that setting; the number of multilinear monomials of degree $d$ is still $\binom{n^2}{d}$, which is not asymptotically so different when $d$ is not close to $n^2$.} So this can call be computed in time at most $\binom{n^2 + d}{d} \cdot n^c \leq n^{O(d)}$ (for some absolute constant $c$, to account for bookkeeping overhead; here we also use the assumption that $d \leq n^2$ to get that $\binom{n^2 + d}{d} \leq \binom{2n^2}{d} \leq (2n^2)^d \leq n^{d(2 + \log_n 2)} \leq n^{O(d)}$). 

Now we apply the algorithm of Chistov, Ivanyos, and Karpinski \cite[Theorem~6]{CIK} to the pair $M_1, M_2$, to break up $R_{\leq d}$ into irreps. That algorithm breaks a representation up into a direct sum of irreducible modules over the algebraic closure of the ground field; however, in the case of $S_n$, irreducible representations of $S_n$ over any field of characteristic zero are already irreducible over the algebraic closure, thus their algorithm indeed returns a decomposition into irreps over the ground field. 
That means we get a basis for $R_{\leq d}$ that is partitioned into bases for irreps. We can also use Chistov--Ivanyos--Karpinski \cite[Corollary~3]{CIK} or Brooksbank--Luks \cite[Theorem~3.5]{BL} to determine which irreps appearing are isomorphic to which other ones, as well as find an isomorphism between them. Thus, for each isomorphism type of irrep, at this point we have constructed a basis $v_1, \dotsc, v_k \in R_{\leq d}$ of one copy of that irrep, and linear isomorphisms between the copies of the irreps. So now focus on an isotypic component---the direct sum of all irreps that are equivalent to a given irrep---say it contains $m$ copies of a single irrep. Let $v_{i,j}$ be the image of $v_i$ in the $j$-th isomorphic copy, that is, for all $j,j'$, the map $v_{i,j'} \mapsto v_{i,j}$ for all $i=1,\dotsc,k$ is an equivalence of $S_n$-representations.

Now, for each of the $m$ copies in our isotypic component, $j=1,\dotsc,m$, consider the vector $u_j := (v_{1,j}(G), \dotsc, v_{k,j}(G))$ (here remember our vectors $v_{ij}$ are actually polynomials in $R$, so by $v_{i,j}(G)$ we mean apply that polynomial to the adjacency matrix of $G$). Arrange those vectors $u_j$ as the columns of a matrix $U$, of size $k \times m$, where $k$ is the dimension of the irrep we are focusing on and $m$ is its multiplicity in $R_{\leq d}$. Then the kernel of $U$, $\{x : Ux = 0\}$, consists of linear combinations of the columns that are zero on $G$. Each such linear combination corresponds to one copy of the irrep in that isotypic component that vanishes on $G$. 

By iterating over all the isotypic components that occur (of which there can be at most $\dim R_{\leq d} \leq n^{O(d)}$), we can then find all test modules that vanish on $G$. By doing similarly for $H$, we can then find all separating modules in $R_{\leq d}$. 

Finally, to turn the rational polynomials into integer ones we can clear denominators.
\end{proof}

\subsection{Relationship between separating modules and separating invariants}

With that motivating theorem in hand, we now pursue the equivalence between bounded-degree separating modules, bounded-degree separating invariants, and the counts of bounded-order subgraphs (the initial coloring of $O(1)$-WL).

Because our polynomials have variables for the edges, but WL works on vertices, in our work we have often run into a sort of ``quadratic mismatch'' between the two (a monomial of degree $d$ can involve up to $2d$ edges, and conversely a $d$-tuple of vertices can have up to $\Theta(d^2)$ edges, which would be encoded by a monomial of degree $\Theta(d^2)$). We introduce here a notion of ``support-degree'' for monomials that counts the number of vertices involved in a monomial:

\newcommand{\sdeg}{\mathrm{supp\text{-}deg}}

\begin{definition}[Support-degree] \label{def:support-degree}
The \emph{support-degree} of a monomial $x_{i_1 j_1} \dotsb x_{i_d j_d}$ is 
\[
\sdeg(x_{i_1 j_1} \dotsb x_{i_d j_d}) := |\{i_1,\dotsc,i_d,j_1,\dotsc,j_d\}|.
\]
The support-degree of a polynomial $f$, denoted $\sdeg(f)$, is the maximum of the support-degree of any of its monomials. If $T$ is a set of polynomials (in particular, a test module), we define its support-degree to be $\sdeg(T) := \max_{f \in T} \sdeg(f)$.
\end{definition}

We always have
\[
\sdeg(f) \leq 2\deg(f) \leq O(\sdeg(f)^2).
\]

\begin{observation}
For all polynomials $f,g$, 
\[
\sdeg(fg) \leq \sdeg(f) + \sdeg(g),
\]
and strict inequality is possible. That is, support-degree gives a filtration, but not a grading, on the polynomial ring in variables $x_{ij}$, as well as its quotient by the Boolean equations $x_{ij}^2 - x_{ij}$.
\end{observation}

\begin{proof}
For two monomials $m=x_{i_1 j_1} x_{i_2 j_2} \dotsb x_{i_d j_d}$ and $m'=x_{\ell_1 k_1} \dotsb x_{\ell_e k_e}$ supported on vertex sets $V = \{i_1,\dotsc,i_d,j_1,\dotsc,j_d\}$ and $V' = \{\ell_1,\dotsc,\ell_e,k_1,\dotsc,k_e\}$, respectively, we have
\begin{align*}
\sdeg(mm') & = |V \cup V'| = |V| + |V'| - |V \cap V'| \\
 & = \sdeg(m) + \sdeg(m') - |V \cap V'| \\
 & \leq \sdeg(m) + \sdeg(m')
\end{align*}
Since $V \cap V'$ need not be empty (e.\,g. $m = x_{12}, m' = x_{13}$), the final inequality can be strict.
\end{proof}

We next recall some key definitions from \cite{Dawar-Wilsenach20}.

\begin{definition}
    If $G$ is a permutation group acting on the variables of a circuit $C$, we say $C$ is $G$-symmetric if every $g \in G$ extends to at least one automorphism of $C$. That is, every permutation of inputs corresponds to at least one permutation of gates such that the resulting circuit computes the same output as the original.

    When $G \leq S_{n^2}$ is the action of $S_n$ on ordered pairs $\pi \cdot x_{i,j} = x_{\pi^{-1}(i), \pi^{-1}(j)}$ we call $C$ \emph{square-symmetric}. 
\end{definition}

\begin{definition}
    If $G$ is a permutation group acting on the variables of a circuit $C$, we say $C$ is {\it ($G$-)rigid} if any permutation of the inputs that comes from $G$ extends to at most one automorphism of $C$.
\end{definition}

Our next lemma, which has several uses, is to show that irreducible separating modules $V$ of degree $d$ can be computed by symmetric algebraic circuits of size $n^{O(d)}$, in the sense that there is such a circuit with multiple outputs, whose outputs form a spanning set for $V$. 

We note that, although it may seem ``obvious'' that given a separating module of degree $d$, it should have a symmetric circuit of size $n^{O(d)}$---since, after all, there are only $n^{O(d)}$ monomials total of that degree---there is an obstacle to proving this. 
Namely, even if we can compute $f$ by a (non-symmetric) depth-2 algebraic circuit of size $n^{O(d)}$, to get a \emph{symmetric} circuit, the natural approach is then to take the orbit of $f$. But if we did not choose $f$ carefully, its orbit could have size $n!$ (in fact, ``most'' $f$ will have orbits of that, maximum, size). When $d$ is small, the elements of that orbit cannot be linearly independent, but the naive way of getting a symmetric circuit from $f$ cannot avoid computing all the elements in the orbit of $f$. So we must also identify an element of the separating module whose orbit has size $n^{O(d)}$ (and not merely spans a set of dimension $n^{O(d)}$, which is always true in degree $d$). To overcome all of these obstacles simultaneously we use quite a bit of information from the representation theory of the symmetric group (albeit mostly in a black-box manner), in the proof of the following key lemma. 

\begin{lemma} \label{lem:spanning}
Suppose two graphs $G,H$ are separated by separating modules in degree $\leq d$. Then there is an irreducible separating module $V$ in degree $\leq d$ such that there is a rigid square-symmetric multi-output circuit $C$ of total size $n^{O(d)}$ and depth 2 whose outputs are a spanning set for $V$.

If we replace degree with support-degree, the same result holds with total size $2^{O(d^2)}n^{O(d)}$.\footnote{\label{fn:bitsize}Both results also hold if we count total bit-size of the circuit, namely including the bit-lengths of all constants involved, or equivalently, count the size of a constant-free circuit.}
\end{lemma}

\begin{proof}
The proof for support-degree is essentially the same as that for ordinary degree, \emph{mutatis mutandis}; the main difference is that there are $2^{O(d^2)}\binom{n}{d}$ monomials of support-degree $d$ rather than $n^{O(d)}$ monomials of degree $d$. 
Since the proof for support-degree is otherwise the same, we only give the details of the proof in the case of ordinary degree.

In the first half of the proof we find a polynomial $v_0$ whose orbit separates $G$ from $H$, has orbit size $n^{O(d)}$, and whose (integer) coefficients (in the monomial basis) have bit-length $n^{O(d)}$. In the second half we then use that polynomial $v_0$ to build a circuit satisfying the statement of the lemma, whose outputs span $V$.

\paragraph{Finding $v_0$.} Towards that end, suppose there is a separating module in degree $d$ that vanishes on $G$ but not on $H$. By the proof of Theorem~\ref{thm:algorithm}, we can compute bases for all the separating modules in $\F[X]_{\leq d}$, consisting of polynomials with integer coefficients of size (and even total bit-length) at most $n^{O(d)}$. (Note that this is not yet enough since it is possible that all of those polynomials have orbits as large as $n!$, which is too large for our use.)

Let $V$ be one such irreducible separating module, and let $\lambda$ be the partition indexing the isomorphism type of $V$. Let $S_{\lambda} := S_{\lambda_1} \times \dotsb \times S_{\lambda_n}$ be the corresponding Young subgroup. 
By, e.\,g., \cite[Sec.~7.2]{Fulton}, $V$ has a surjective $S_n$-equivariant linear map from $\F[S_n / S_{\lambda}]$, and thus by Obs.~\ref{obs:perm-to-linear} $V$ contains a non-zero vector that is stabilized by $S_{\lambda}$ (its stabilizer may be larger, but it is at least stabilized by $S_{\lambda}$). Let $\pi_{inv} = \sum_{\sigma \in S_{\lambda}} \sigma$ be the (scaled) projection operator onto the $S_{\lambda}$-invariants. Since $V$ contains $S_{\lambda}$-invariants, at least one of its basis vectors $v$ as computed above must have non-zero projection $\pi_{inv}(v)$; we may check all of them, since there are at most $n^{O(d)}$ of them. Let $v_0 = \pi_{inv}(v)$ for some such $v$. Since $\deg v_0 \leq d$, it can be computed as a depth-2 circuit (sum of monomials) of size at most $n^{O(d)}$.\footnote{To realize Footnote~\ref{fn:bitsize}, here note that $v_0$ is a sum of at most $|S_{\lambda}| \leq n!$ many permuted copies of $v$, so it has integer coefficients of magnitude at most $n! \cdot 2^{n^{O(d)}} \lesssim 2^{n^{O(d)} + n \log n} \leq 2^{n^{O(d)}}$, hence bit-size at most $n^{O(d)}$.}
We have thus found the $v_0$ we sought at the outset.

\paragraph{Building the circuit.}
The remainder of the proof is now to build our circuit based on the polynomial $v_0$.

Let $C_0$ be a depth-2 circuit computing $v_0$ as a sum of monomials with integer coefficients. For each $\tau \in S_n$, let $\tau \cdot C_0$ be the circuit $C_0$ applied to $\tau$ applied to the input variables $x_{ij}$. For $\tau \in S_{\lambda}$, since $v_0$ is fixed by $S_{\lambda}$, we have that $\tau$ must permute the gates of $C_0$ (which correspond to monomials of $v_0$) amongst themselves, giving back the same circuit $C_0$. The orbit of $C_0$ thus consists of at most (remember that the stabilizer of $v_0$ could be larger than $S_{\lambda}$) $|S_n| / |S_\lambda| = n! / \prod_{i=1}^{n} \lambda_i!$ distinct circuits; let $C$ be the union of these circuits. (Conceptually, we may start with $C$ as a disjoint union, but if two gates compute the same monomial we may merge them together into a single gate.) Then $C$ is a multi-output, square-symmetric circuit whose outputs are precisely the polynomials in the $S_n$-orbit of $v_0$. 

Since $V$ is irreducible, the orbit of any nonzero vector in $V$ is a spanning set for $V$, and thus the outputs of $C$ are a spanning set for $V$.

Next we show that $C$ is rigid: suppose $\alpha$ is an automorphism of $C$. If it moves one of the product gates at depth 1, then $\alpha$ moves the corresponding monomials. Since they were distinct monomials, it must also swap the variables that go into those monomials, and thus $\alpha$ cannot act trivially on the variables. Next, suppose instead that $\alpha$ moves some (output) gates at depth 2. Even if two output gates are supported on the same set of monomials, they cannot have the same coefficients, and thus $\alpha$ must also move some of the gates at depth 1 (moving the monomials) or depth 0 (which also has the effect of swapping which monomial each depth-1 gate computes). In either case, $\alpha$ does not fix the variables, and thus $C$ is rigid.

We will now make use of the fact that by Cor.~\ref{cor:Sn-reps-by-degree}, $\lambda$ has $\lambda_1 \geq n-2d$. Thus $|S_{\lambda}| \geq (n-2d)!$, so the number of outputs of $C$ (the orbit of $v_0$) is at most $n! / (n-2d)! < n^{2d}$. Since $v_0$ has degree $d$, the size of $C_0$ is at most $n^{O(d)}$, and thus the size of $C$ is at most $n^{2d} \cdot n^{O(d)} \leq n^{O(d)}$.\footnote{To realize the claim about bit-size in Footnote~\ref{fn:bitsize}, note here that the coefficients appearing in $C$ are the same as those appearing in $C_0$, so each coefficient has magnitude at most $2^{O(dn^d \log n)}$, and thus bit-size at most $O(d n^d \log n)$. Since there are most $n^{O(d)}$ such coefficients, the total bit-size of $C$ is at most $n^{O(d)}$.}
\end{proof}

\begin{proposition} \label{prop:module-to-invariant}
\begin{enumerate}
\item    If two graphs $G,H$ are separated by separating modules in degree $\le d$, then they are separated by an invariant of degree $\le 2d$. Moreover, this invariant can be computed by a rigid square symmetric algebraic circuit of size\footnote{All three parts of this result also hold for total bit-size; see Footnote~\ref{fn:bitsize}.} $n^{O(d)}$ and depth 2.

\item    The same holds with degree replaced everywhere with support-degree, with the resulting total size $2^{O(d^2)}n^{O(d)}$.

\item If $G,H$ are separated by a separating module with a spanning set computed by a (rigid) square-symmetric circuit of size $s$, then they are separated by an invariant computed by a (rigid) square-symmetric circuit of size $\leq 2s+1$.
\end{enumerate}
\end{proposition}

Since separating invariants are a special case of separating modules, this shows that whether measured by degree or by circuit-size, asymptotically there is little difference in the distinguishing power of separating modules and separating invariants for \textsc{Graph Isomorphism}.

\begin{proof}
1. If $G,H$ admit a separating module in degree $d$, the algorithm from Theorem~\ref{thm:algorithm} proves, in particular, that there exists an irreducible separating module $V$ in degree $d$ that has a basis consisting of polynomials with rational coefficients; by scaling, we may assume those coefficients are in fact integers. To make the separation of ideas clear, we first prove the existence of an invariant of low degree that also separates $G$ from $H$, then we show how to reduce its circuit complexity.

Take any nonzero polynomial $f \in V$ with integer coefficients; let $f_1,\dotsc,f_m$ be its orbit.  Let $\hat{f} = \sum_{i=1}^m f_i^2$. Since $S_n$ acts on these irreps by rational matrices, all the $f_i$ have rational coefficients. Thus their squares are polynomials that only take non-negative values on $\{0,1\}$-adjacency matrices (or indeed on any real-valued matrices). Thus the output of $\hat{f}$ has the property that it is $0$ on an adjacency matrix $A$ if and only if every polynomial $f_i$ vanishes on $A$, since when any term is non-zero, it must be positive, so the terms can never cancel one another. Since $V$ is irreducible, any nonzero orbit spans $V$, and thus $\hat{f}(A)=0$ iff $h(A)=0$ for all $h \in V$. Thus, if $V$ vanishes on $G$ but not on $H$, the same is true of $\hat{f}$.

However, the size of $f$'s orbit could be as large as $n!$, far too large for our circuit size bound. We know show how to choose $f$ so that we get a small circuit. 

Let $C$ be the rigid square-symmetric circuit of size $n^{O(d)}$ and depth 2 computing a spanning set for $V$ from Lemma~\ref{lem:spanning}. Let $C'$ be the circuit that starts the same as $C$, but then takes each output of $C$, squares it, and adds them all together. Then $size(C') \leq 2size(C)+1$, and the degree of the polynomial computed by $C'$ is $2 \deg(V)$. Furthermore, since $C$ was (square-)symmetric, summing the squares of its outputs produces an invariant polynomial. The same argument as above shows that $C'$ vanishes on an adjacency matrix $A$ if and only if every polynomial in $V$ does. 

Finally, to get $C'$ down to depth 2, we note that we can replace each output gate $g$ of $C$ with a depth-2 circuit computing $g^2$ as a sum of monomials, and of size at most the square of the size of the subcircuit rooted at $g$. Once this is done, the final addition gate of $C'$ can be merged into the addition gates that are the outputs of $C$, resulting overall in a circuit whose size is still $n^{O(d)}$ but whose depth is now only 2, that computes a separating \emph{invariant} for $G,H$ of degree at most $2d$, as claimed.

2. The same proof works for support-degree, \emph{mutatis mutandis}.

3. Use the same ``sum the squares of the outputs'' trick from the first paragraph above. 
\end{proof}

\subsection{Equivalence between support-degree-$d$ separating invariants and $d$-vertex subgraph counts} \label{sec:invariants1}
In order to characterize the power of invariants of degree $d$ (and hence, by Proposition~\ref{prop:module-to-invariant}, also separating modules of degree $O(d)$), we give an explicit basis for such invariants. If $f$ is an invariant polynomial on the adjacency matrices of $n$-vertex graphs, consider one of its monomials $\alpha m$, where $\alpha \in \mathbb{F}$ and $m$ is a product of variables. Let $H$ be the $n$-vertex graph whose edge set is precisely $E(H) := \{(i,j) : x_{ij} \text{ divides } m\}$. Then we see that for any graph $G$, $m(G)=1$ if and only if $H$ occurs in $G$ as a (not necessarily induced) subgraph in the position specified by $H$. Furthermore, since $f$ is invariant, every monomial of the form $\sigma(m)$ ($\sigma \in S_n$) must also occur in $f$ with the same coefficient $\alpha$. This leads us to define:

\[
S_H :=\frac{1}{|\Aut(H)|}\sum_{\sigma\in S_n}\sigma\left(\prod_{(i,j)\in E(H)}x_{ij}\right) = \sum\left\{\sigma\left(\prod_{(i,j)\in E(H)}x_{ij}\right) : \sigma \in S_n\right\}
\]
The normalizing factor here is set so that $S_H$ is precisely the sum over the monomials $\sigma(m)$, with each one occuring with coefficient $1$. Put another way: since we are summing over the whole symmetric group, if the graph $H$ associated to $m$ has some automorphisms, we would get each such monomial with coefficient $|\Aut(H)|$. The degree of $S_H$ is $|E(H)|$ and the support-degree of $S_H$ is the number of non-isolated vertices of $H$.

The preceding reasoning leads quickly to:

\begin{observation}\label{obs:invariants}
    Every (degree $d$ homogeneous, resp., support-degree $d$ homogeneous) invariant polynomial is a linear combination of the $S_H$ polynomials (where $|E(H)|=d$, resp., $|V(H)|=d$).
\end{observation}
\begin{proof}
    If $f$ is an invariant polynomial, then it is a linear combination of orbits of monomials. To each monomial $m = x_{i_1 j_1} x_{i_2 j_2} \dotsb x_{i_k j_k}$, we associate the graph $H_m$ on vertices $\{i_1,i_2,\dotsc,i_k, j_1, \dotsc, j_k\}$ with edge set $\{(i_1, j_1), \dotsc, (i_k, j_k)\}$. Then the sum over the orbit of $m$ is a scalar multiple of $S_H$. Moreover, if $m$ has degree $d$, then $H_m$ has exactly $d$ edges, and if $m$ has support-degree $d$, then $H_m$ has exactly $d$ vertices.
\end{proof}

We now have all the ingredients in place for the main complexity theorem of this section. To state it more succinctly, we introduce the notation:
\newcommand{\WLDim}{\textsf{WLDim}}
\newcommand{\SepDeg}{\textsf{SepDeg}}
\newcommand{\InvDeg}{\textsf{InvDeg}}
\newcommand{\SepSuppDeg}{\textsf{SepSuppDeg}}
\newcommand{\InvSuppDeg}{\textsf{InvSuppDeg}}
\begin{itemize}
\item $\WLDim_0(G,H)$ is the minimum $k$ such that $k$-vertex subgraph census (the 0th round of $k$-WL) distinguishes $G$ and $H$

\item $\SepDeg(G,H)$ is the minimum degree of any separating module for $G,H$; $\SepSuppDeg(G,H)$ is the minimum support-degree of any separating module for $G,H$.

\item $\InvDeg(G,H)$ is the minimum degree of any separating invariant for $G,H$; $\InvSuppDeg(G,H)$ is the minimum support-degree of any separating invariant for $G,H$.
\end{itemize}

\begin{theorem} \label{thm:degree-new}
For all graphs $G,H$ we have
\[
\WLDim_0(G,H) \leq \InvSuppDeg(G,H) \leq 2\SepSuppDeg(G,H) \leq 2\WLDim_0(G,H)
\]
and
\[
\WLDim_0(G,H) \leq 2\InvDeg(G,H) \leq 4\SepDeg(G,H) \leq O(\WLDim_0(G,H))^2).
\]
Consequently, $\WLDim_0 = \Theta(\InvSuppDeg) = \Theta(\SepSuppDeg)$, and $\WLDim_0$ is $O(1)$ if and only if $\InvDeg$ and $\SepDeg$ are.
\end{theorem}

\begin{proof}
By Proposition~\ref{prop:module-to-invariant}, if $G,H$ admit a separating module of support-degree (resp. degree) $d$, then they admit a separating invariant of support-degree (resp. degree) $2d$. This gives the middle inequality of both sequences of inequalities.

By Observation~\ref{obs:invariants}, invariants of support-degree (resp., degree) $d$ are linear combinations of the polynomials $S_K$ for graphs $K$ on $\leq d$ vertices (resp., edges). The value of those invariants on a graph $G$ are completely determined by the counts of induced subgraphs on $\leq d$ (resp., $2d$) vertices, giving the first inequality of both sequences.

Conversely, the initial coloring of $d$-WL is determined by the invariants $S_K$ for the graphs $K$ on $d$ vertices, hence is determined by invariants of support-degree $\leq d$ (resp., degree $\leq O(d^2)$), giving the final inequality of both sequences.
\end{proof}

\section{Support size and complexity measures on the symmetric group} \label{sec:support}
A key concept that will be useful in our proofs is that of ``support'' of an element of a set that is acted on by $S_n$. In this section we recall the definition of support and some of its basic properties, and we relate it to other complexity measures on the symmetric groups. The notion of support will be used crucially in the next section as well.

\begin{definition}[{Support, cf. \cite{BGS99, Dawar-Wilsenach22}\footnote{Dawar and Wilsenach \cite{Dawar-Wilsenach22} define the support of a subgroup $G \leq S_n$ rather than of an element $\omega \in \Omega$, but these are equivalent concepts: For $\omega \in \Omega$, if we define $G := \stab_{S_n}(\omega)$ then supports of $G$ in their sense are the same as supports of $\omega$ in our sense. And conversely, given $G \leq S_n$, we may take $\Omega = S_n / G$ to be the coset space, and $\omega = G$ to be the trivial coset, and then supports of $\omega$ in our sense are the same as supports of $G$ in their sense.}}]\label{def:support}
    Suppose $S_n$ acts on a set $\Omega$. A \emph{support} of $\omega \in \Omega$ is a subset $X \subseteq [n]$ (sic!) such that every element in the pointwise stabilizer of $X$  also stabilizes $\omega$, that is
    \[
    \bigcap_{x \in X} \stab_{S_n}(x) \subseteq \stab_{S_n}(\omega). 
    \]
\end{definition}

For concreteness, it may help the reader to note that if $X = \{1,\dotsc,k\}$, then the pointwise stabilizer on the left-hand side of the above definition is merely a copy of $S_{n-k}$, and the defining condition becomes $S_{n-k} \subseteq \stab_{S_n}(\omega)$ (though, despite the latter notation, it often matters \emph{which} copy of $S_{n-k}$ is being referred to).

Note that we have (at least) two actions of $S_n$ here: the action on $\Omega$, and the defining action on $[n]$. And supports are defined to always be subsets of $[n]$. A typical application of this in our paper will have $\Omega = [n] \times[n]$ the possible edge set of directed graphs on vertex set $[n]$ (or $\Omega = \binom{[n]}{2}$ for undirected graphs), and in that case supports are (by definition) sets of vertices. 

Although we don't technically need the following lemma, we find it useful for gaining some intuition about supports, which may aid the reader in following our subsequent proofs.

\begin{lemma}[\cite{BGS99, Dawar-Wilsenach22}] \label{lem:min-support-unique}
    Suppose $\omega\in\Omega$ has a support of size less than $n/2$. Then there exists a unique minimum-size support of $\omega$. Moreover, this support is a subset of each support of $\omega$ of size $n/2$.
\end{lemma}

\begin{definition}[Support size] \label{def:support-size}
Suppose $S_n$ acts on $\Omega$; the \emph{support size} of $\omega \in \Omega$ is the minimum size of any support of $\omega$. The \emph{support size} of the $S_n$-set $\Omega$ is $\max_{\omega \in \Omega} \text{support-size}(\omega)$. 
\end{definition}

We will frequently apply this to the case where $\Omega$ is the set of gates of a circuit that is acted on by $S_n$. 

Note that if $m$ is a monomial then $\sdeg(m)$ as in the previous section agrees with the size of a minimal support of $m$ in the sense of Definition~\ref{def:support}; however, for a more general polynomial $f$, $\sdeg(f)$ may not agree with the minimum size of a support for $f$ (it is, instead, the maximum over monomials of the minimum support size of those monomials).

Using ideas from the proof of Lemma~\ref{lem:spanning}, we now show that the support size of a circuit computing a spanning set for a representation inside a polynomial ring $R$ on which $S_n$ acts is closely related to other complexity measures on the symmetric group. In particular, \cite{DFLLV21} show that the number $n - \lambda_1$ associated to a partition $\lambda$ indexing an irrep is closely related to a number of other measures of complexity on the symmetric group, including sensitivity, block-sensitivity, and decision tree complexity. In their setting, they consider irreps inside the coordinate ring (or group algebra) of $S_n$ directly; Proposition~\ref{prop:alg-circuit} generalizes the connection to the action of $S_n$ on other polynomial rings.

\begin{proposition}\label{prop:alg-circuit}
Suppose $S_n$ acts on a set $\Omega$, and let $R$ be the polynomial ring on variable set $\Omega$, and let $s_\Omega$ denote the support size of $\Omega$ (Definition~\ref{def:support-size}).

If $V \subseteq R$ is an irreducible representation of shape $\lambda$ with $k \geq n - \lambda_1$, occurring in $R_{\leq d}$ then there is a rigid square-symmetric algebraic circuit $C$ of support size at most $\max\{k,d \cdot s_\Omega\}$ computing a spanning set for $V$.

Conversely, if $C$ is a rigid, symmetric algebraic circuit of support size $k$ computing polynomials in $R$, then the outputs of $C$ span a representation of $S_n$, and that representation is a direct sum of irreducible representations of shapes $\lambda$ with $k \geq n - \lambda_1$.
\end{proposition}

\begin{proof}
($\Rightarrow$) As in Lemma~\ref{lem:spanning}, $V$ contains a vector $v_0$ that is stabilized by the Young subgroup $S_{\lambda}$. As there, let $C_0$ be a depth-2 circuit computing $v_0$ as a sum of monomials, and let $C$ be the union of the circuits $\tau C_0$ for all $\tau \in S_n$ (omitting repeats). Now, while the number of copies of $C_0$ is still at most $|S_n| / |S_{\lambda}| \leq n! / (n-k)! < n^{k}$, unlike the above proof we no longer know that $C_0$ itself has size $n^{O(k)}$, so we instead must argue more directly to bound the support size of $C$.

Now, $S_{\lambda}$ contains $S_{\lambda_1}$, and since $S_{\lambda}$ fixes the output gate corresponding to $v_0$, $S_{\lambda_1}$ does as well. Thus $v_0$ has a support of size at most $n - \lambda_1 \leq k$. Since the output gates are all in the orbit of $v_0$, the same holds for them as well. 

Now consider a monomial gate in $C$. Since it has degree $d$, it is a product of at most $d$ many variables $x_{\omega}$, say those corresponding to $\omega_1,\dotsc,\omega_d \in \Omega$. Since each $\omega_i$ has a support of size $\leq s_\Omega$, if a permutation fixes a support for each $\omega_i$, then it fixes the monomial gate. Thus there is a set of size at most $d s_{\Omega}$ such that any permutation fixing that set pointwise fixes the monomial gate. Thus it has a support of size at most $ds_{\Omega}$.

($\Leftarrow$) Conversely, suppose $C$ is a rigid symmetric algebraic circuit of support size $s$ computing a polynomial in $R$. Let its outputs be $f_1,\dotsc,f_m$. Then since $C$ is a rigid symmetric circuit, the action of $S_n$ on the variables permutes the set of output gates computing $\{f_1,\dotsc,f_m\}$ in the same way that the action of $S_n$ on the polynomial ring permutes $\{f_1,\dotsc,f_m\}$ (by induction on the structure of the circuit). Thus the span of $f_1,\dotsc,f_m$ is an $S_n$-representation $V$.

Finally, we show that $V$ is a direct sum of irreducible representations of shapes $\lambda$ with $k \geq n - \lambda_1$. Since $C$ has support size $k$, the gate computing $f_1$ is fixed by a subgroup isomorphic to $S_{n-k}$; thus the stabilizer $H$ of $f_1$ contains $S_{n-k}$. Since $\{f_1,\dotsc,f_m\}$ spans $V$, by Obs.~\ref{obs:perm-to-linear} there is an $S_n$-equivariant linear map $\varphi\colon \F[S_n / H] \to V$, such that the image of the trivial coset is $H$ is $f_1$. By equivariance, the $S_n$-orbit of $f_1$ is contained in the image of $\varphi$, and since the orbit of $f_1$ spans $V$, $\varphi$ is surjective. Since $S_{n-k} \leq H$, there is an $S_n$-equivariant surjection $S_n / S_{n-k} \to S_n / H$, and which extends to a surjective $S_n$-equivariant linear map $\F[S_n / S_{n-k}] \to \F[S_n / H]$. Composing these, we get an $S_n$-equivariant linear surjection $\F[S_n / S_{n-k}] \to V$. By Corollary~\ref{cor:branching}, $\F[S_n / S_{n-k}]$ only contains irreducible representations of shape $\lambda$ with $k \geq n-\lambda_1$. Thus the same is true for any quotient of $\F[S_n / S_{n-k}]$ as an $S_n$-representation, and in particular is true for $V$. This completes the proof.
\end{proof}

\section{Symmetric circuit size of separating modules and Weisfeiler--Leman} \label{sec:sym-ckt-size}

In this section we show that symmetric algebraic circuits of size $n^{\Theta(k)}$ and depth $r$ are essentially equivalent to $\Theta(r)$ rounds of the $\Theta(k)$-WL algorithm. 
We will also prove an analogous statement for Boolean threshold circuits---indeed, the statement for Boolean threshold circuits will be used in the course of the proof for algebraic circuits. In both the Boolean setting and the algebraic setting, one direction of the equivalence is already known:

\begin{theorem}[{Dawar--Wilsenach, cf. \cite[Thm.~4.8]{DW24}}] \label{thm:ADW-WL}
Let $G,H$ be two graphs. If there is a rigid, square-symmetric circuit (algebraic or Boolean, as long as all individual gates are symmetric functions of their inputs) $C$ of support size at most $k \leq n/4$, and depth $\Delta$ such that $C(G) \neq C(H)$, then $G$ and $H$ are separated by $O(\Delta k)$ rounds of $O(k)$-WL.
\end{theorem}

They phrase this in terms of a winning strategy for Spoiler in Hella's bijective pebble game with $2k$ pebbles, but this is equivalent to $2k$-WL \cite{hella}. We now record how this extends to separating modules:

\begin{theorem} \label{thm:SeptoWL}
    For each $1/\log_2 n<\varepsilon<1$ and $k \leq \varepsilon n^{1-\varepsilon}$, if two $n$-vertex graphs $G,H$ admit a separating module with symmetric algebraic circuits of size $n^{O(k)}$ and depth $\Delta$, then they can be distinguished by $O(\Delta k /\varepsilon)$ rounds of $O(k/\varepsilon)$-WL.
\end{theorem}

\begin{proof}
Under the assumption, by Proposition~\ref{prop:module-to-invariant} there is a separating \emph{invariant} computed by a (single-output) symmetric algebraic circuit of depth $\Delta+2$ and size $n^{O(k)}$. By Lemma~\ref{lem:binomial}, this circuit has support size $O(k/\varepsilon)$. The result now follows directly from Thm.~\ref{thm:ADW-WL}.
\end{proof}

In the remainder of this section we will prove the following theorem by taking a Boolean circuit which simulates WL and converting it to a rigid square-symmetric algebraic circuit that also simulates WL.

\begin{theorem} \label{thm:WLtoSep}
    The $r$-round $k$-dim WL algorithm can be simulated by separating modules (in fact, invariant polynomials) computed by rigid square-symmetric algebraic circuits of size $rn^{O(k)}$ and depth $O(r)$. 
\end{theorem}

\begin{remark} \label{rmk:hom-count}
An anonymous reviewer proposed the intriguing suggestion that this theorem might also be provable using homomorphism-counting polynomials. Indeed, the homomorphism counts from graphs with $k$-pebble forest covers of depth $r$ (a proper subset of the collection of graphs of treewidth $k$ and treedepth $r$) capture $r$ rounds of $k$-WL \cite{AFSS26}. However, as far as we are aware the symmetric algebraic circuit size of homomorphism-counting polynomials as separating modules has only been studied as a function of the treewidth of the graphs \cite{DPS26}; this gives the same dependence on $(n,k)$ as above $n^{O(k)}$, but does not immediately give us the dependence on $r$ as in the above theorem.
\end{remark}

In the setting of Boolean threshold circuits, we will also use the following key result from Dawar \& Wilsenach. The statement involves one new bit of terminology: when a circuit $C$ is acted on by a group, we say the \emph{orbit size} of $C$ is the maximum size of any orbit of gates in $C$.

\begin{theorem}[{Dawar \& Wilsenach \cite[Thm.~4.10]{Dawar-Wilsenach22}}] \label{thm:DW-support}
Let $C$ be a rigid square-symmetric (threshold) circuit of order $n>8$. For every $k\le n/4$, if the orbit size of $C$ is less than $\binom{n}{k}$, then the support size of $C$ is less than $k$.
\end{theorem}

First, we look at a seemingly-stronger definition of rigidity from \cite{Anderson-Dawar} (see their Definition 5 and Proposition 1) that is easier to apply and turns out to be equivalent to the definition of rigidity from  \cite{Dawar-Wilsenach20}. This definition uses the following notion of twins: \begin{definition}[Twins]
    Two gates $g,g'\in C$ are {\it in-twins} if they are the same type of gate and have the exact same set of inputs: for all $g'' \in C$, we have $(g'',g) \in C \Longleftrightarrow (g'',g') \in C$. They are \emph{twins} if they are in-twins and also out-twins: for all $g'' \in C$, we have $(g,g'') \in C \Longleftrightarrow (g',g'') \in C$.
\end{definition}

\begin{proposition}[{cf. \cite[Proposition~1]{Anderson-Dawar}}] \label{prop:rigid}
Let $G \leq S_n$ be a permutation group and $C$ a Boolean or algebraic circuit on $n$ variables that is $G$-symmetric (or even just a directed acyclic graph with $n$ source nodes). If $C$ has twins it is not $G$-rigid, and if $C$ is not $G$-rigid then it has in-twins.
\end{proposition}

\begin{proof}
First, suppose $C$ has twin gates $g,g'$. Then there is an automorphism $\gamma$ of $C$ that swaps the two gates $g,g'$ and fixes all the other gates. In particular, this means that both $\gamma$ and the trivial automorphism are extensions of the trivial permutation of the inputs to the whole circuit, which is the identity element of $G$, so $C$ is not $G$-rigid.

Conversely, suppose $C$ is not $G$-rigid. Then there is some $\pi \in G$ and two distinct automorphisms $\alpha,\beta$ of $C$ that both agree with $\pi$ in how they act on the inputs. Then $\gamma := \alpha\beta^{-1}$ is a nontrivial automorphism of $C$ that fixes all the inputs. Since $\gamma$ is nontrivial, it must move at least one gate. Let $g$ be a gate that is moved by $\gamma$ such that none of the inputs to $g$ are moved by $\gamma$; such a $g$ must exist by structural induction, because none of the inputs to $C$ are moved by $\gamma$, yet $\gamma$ moves at least one gate. Since $\gamma$ is an automorphism of $C$ that does not move the inputs to $g$, we have that $\gamma(g)$ and $g$ are in-twins.
\end{proof}

The Boolean circuit we start with is constructed by Grohe and Verbitsky. They construct an $r+2$ layer Boolean threshold circuit in which the first layer simulates the initial coloring step of WL, the next $r$ layers simulate each round of WL refinement, and the last layer computes the output. We want to prove that any permutation of inputs induces exactly one automorphism of the circuit. This will not quite be true of their circuit as written, but we will show how to modify their circuits to ensure this property.

Below we use the standard definition for $\mathsf{TC}^0$, i.e. the class of problems computable by constant depth, polynomial size Boolean circuits with AND, OR, NOT, and MAJ gates. See \cite{Viola} for more details.

\begin{theorem}[{Grohe--Verbitsky \cite[Theorem~3.2]{Grohe-Verbitsky}}]\label{thm:GVcircuit}
    The $r$-round $k$-dim WL algorithm can be implemented by a logspace uniform family of Boolean threshold circuits of depth $O(r)$ and size $O(rn^{3k})$ where $k\ge2$ and $r=r(n)$.
\end{theorem}

\begin{proof}[Abridged proof]
    First we define layer 0, which is a $\mathsf{TC}^0$ circuit with $2k^2\cdot 2n^k$ output gates that computes the marked isomorphism type of every $k$-tuple of nodes in both $G$ and $H$. Note that each marked isomorphism type can be represented with a bit-string of length $2k^2$ because WL does not require the $k$-tuple of nodes to be $k$ distinct nodes,\footnote{\label{fn:GVisotype}One adjacency matrix of size $\leq k \times k$ records the marked isomorphism type of the graph; another adjacency matrix of size exactly $k \times k$ encodes which indices in the tuple are equal to which other indices. This latter part can be done much more efficiently (down to $k\log k$ bits), but doing so is more complicated.} and there are $2n^k$ tuples in both graphs. First, this layer computes the formula $$Y_0(a,b)=\bigwedge_{i\in[2k^2]}(Z_0(a,i)\leftrightarrow Z_0(b,i))$$ where $a,b$ are length-$2k^2$ labels for the marked isomorphism type of $k$-tuples of nodes in $G$ or $H$ and $Z_0(a,j)$ computes the $j$-th bit of the bitstring which encodes the marked isomorphism type of $a$. $Y_0(a,b)=1$ if and only if the marked isomorphism types of $a$ and $b$ are the same. These outputs are combined to compute $$X_0(a,c)=Y_0(a,c)\land\bigwedge_{d=1}^{c-1}\neg Y_0(a,d)\footnote{The negation is omitted in the original paper; we believe that was just a typo.}$$ where $c$ is the index of the first tuple which has the same marked isomorphism type as $a$. $X_0$ requires constant depth and threshold gates, so the first layer is a $\mathsf{TC}^0$ circuit.
    
    \begin{remark}
        Because we just assign the lexicographically first matching index as the name of the color of the isomorphism type, the circuit is not square-symmetric (the lex leader need not move equivariantly under the $S_n$ action).\footnote{In an earlier version of this paper we had said that these circuits were not rigid; in fact, using the criterion of no in-twins from Proposition~\ref{prop:rigid} we can see that it is rigid, but as we have argued here it is not square-symmetric.} 
    \end{remark}

    Layers 1 through $r$ are each copies of a $\mathsf{TC}^0$ circuit that does the WL iteration and are defined inductively. It does this by checking two conditions. The first is that any pair of tuples with different colors, either in the same graph or one in either graph, must have different colors after passing through the layer. This is computed by the formula $$\bigwedge_{c\in[2n^k]}(X_{\ell-1}(a,c)\leftrightarrow X_{\ell-1}(b,c)).$$ The second condition is that for any pair of tuples with one tuple in either graph, if the multiset of colors of $i$-neighboring tuples is distinct, then they get distinct colors after passing through the layer. The following formula checks that the colors of $i$-neighbors of $a$ when substituting $v,v'\in V(G)$ are distinct. $$\phi(v,v')=\bigwedge_{j\in [k]}\bigwedge_{c\in[2n^k]}(X_{\ell-1}(a^{jv},c)\leftrightarrow X_{\ell-1}(a^{jv'},c))\footnote{$\phi$ implicitly depends on $a$. We make this explicit in our proof of Lem.~\ref{lem:mGVcircuit}}$$ where $a^{jv}$ is the bitstring representing the marked isomorphism type of the tuple of nodes $a$ where the $j$-th element has been swapped with $v$. We can then define $\psi(v,w)$ the same way for $w\in V(H)$. The full condiction check is then computed by the formula $$\bigwedge_{v\in V(G)}\left(\left(\sum_{v'\in V(G)}\phi(v,v')\right)=\left(\sum_{w\in V(H)}\psi(v,w)\right)\right).$$ The conjunction of these two conditions gives the definition of $Y_\ell(a,b)$. $X_{\ell}(a,c)$ is defined in the same way as $X_0$.

    The final layer compares the counts of each color to see that they are equal in the two graphs, using threshold gates to check numerical equality and an AND gate over the different types of colors. The formula is $$\bigwedge_{c\in[2k^2]}\left(\sum_{a=1}^{n^k}X_r(a,c) = \sum_{b=n^k+1}^{2n^k}X_r(b,c)\right).$$ Thus, the whole circuit is a threshold circuit of depth $\Theta(r)$.
\end{proof}

As remarked above, the construction is not square-symmetric, so we will need to modify it to be square-symmetric. In order to get separating modules out of the construction, we will also need to modify it to produce circuits that look at the two graphs separately from one another, rather than together in a single circuit as in the Grohe--Verbitsky construction. This is achieved in the following lemma:

\begin{lemma}\label{lem:mGVcircuit}
For any $n$-vertex graph $G$ and any $k,r \in \mathbb{N}$, there is a rigid square-symmetric Boolean threshold circuit $C_{G,k,r}$ of depth $\Theta(r)$ and size $O(rn^{3k})$ such that for all $n$-vertex graphs $H$, $C_{G,k,r}(H)=1$ if and only if $G$ and $H$ are not distinguished by $r$ rounds of $k$-WL.
\end{lemma}

\begin{proof}[Proof idea]
Given the graphs $G,H$ first look at which $k$-marked isomorphism types actually occur in $G$, and number the initial colors of $G$ lexicographically (the choice of numbering doesn't matter, but it is something that we get to fix into the circuit). Then build a circuit that is similar to the Grohe--Verbitsky circuit, except ``tailored'' to $G$, in that the initial colors aren't determined the same way as the original construction, but rather we decide what the indexing of initial colors will be based on the $k$-marked isomorphism types we see in $G$. Call that circuit $C_{G,k,0}$; more specifically, this circuit will have the property that $C_{G,k,0}(H)=1$ iff $H$ and $G$ have the same counts of marked isomorphism types, and otherwise $C_{G,k,0}(H)=0$.

If those counts are the same, we do the same procedure again, but instead simulating round 1 of $k$-WL. We use the multi-output circuit $C_{G,k,0}$ as input to round 1, then look at the $\leq n^k$ distinct colors that occur in $G$ and build a new circuit $C_{G,k,1}$ that has those particular colors ``baked in'' in lexicographic order (again, the choice of ordering is not relevant). 

    Importantly, the construction does not have any duplicate nodes in any layer. Every pair of nodes in one layer compute different things, and so their inputs are different. Thus, this circuit has no in-twins and is therefore rigid (Proposition~\ref{prop:rigid}).

    Furthermore, since each circuit $C_{G,k,r}(H)$ only takes $H$ as input, it is easier for us to ensure that it is square-symmetric (as a function of $H$). 
\end{proof}

\begin{proof}[Proof of Lemma~\ref{lem:mGVcircuit}]
By induction on $r$. 

For $r=0$, we (the circuit builder) look at the at most $n^k$ marked isomorphism types of induced subgraphs of $G$ on $k$ vertices; we list these in some order arbitrarily, say $t_{0,1}, \dotsc, t_{0,m_0}$ where $m_0 \leq n^k$, and each $t_{0,i}$ is the marked isomorphism type of a $k$-tuple (see Footnote~\ref{fn:GVisotype}), hence has bit-length $\leq 2k^2$. 
 
 Analogous to the $Y_0(a,b)$ of the Grohe--Verbitsky circuit, for a $k$-tuple $a$ of vertices in $H$, and a marked isomorphism type $t_{0,j}$, we define
\[
X'_{0,j}(a)=\bigwedge_{i\in[2k^2]}(Z_0(a,i)\leftrightarrow t_{0,j}[i]),
\]
where $t_{0,j}[i]$ is the $i$-th bit of $t_{0,j}$, and as before $Z_0(a,i)$ is the $i$-th bit of the bitstring encoding the marked isomorphism type of $a$ in $H$. We highlight the differences with $Y_0(a,b)$ from the Grohe--Verbitsky circuit: (1) we have put $j$ in the subscript of $X'_0$ to indicate that this is now baked into the circuit, whereas $a$ remains dependent on the inputs to the circuit; (2) note that because $t_{0,j}$ is a constant bit-string, $(Z_0(a,i)\leftrightarrow t_{0,j}[i])$ is just $Z_0(a,i)$ if $t_{0,j}[i]=1$ or $\neg Z_0(a,i)$ if $t_{0,j}[i]=0$. (We no longer need a separate $Y_0$ and then $X_0$ because of the way we label $k$-tuples; we use $X'_0$ here because it fits better into the inductive scheme, even though the formula is more similar to their $Y_0$.)

Now, for $i=1,\dotsc,m$, let $c_{0,i}$ be the count of the marked isomorphism type $t_{0,i}$ in $G$. (Note again that we are ``finding'' these counts as the circuit designers, the circuit itself does not need to calculate them.) The output gate of the circuit $C_{G,k,0}$ is then
\[
C_{G,k,0}(H) := \bigwedge_{i=1}^{m_0} \left(c_{0,i} = \sum_{a \in V(H)^k} X'_{0,i}(a) \right).
\]
Recall that here $H$ is the only input to the circuit, while $G$ is being used in its construction, so $Z_0$ and $X'_{0,i}$ depend on the input bits that are the adjacency matrix of $H$.

Now let $r \geq 1$. Let $t_{r,1},\dotsc,t_{r,m_r}$ denote the distinct colors of $k$-tuples of vertices of $G$ at round $r$ (in particular, $m_r \leq n^k$). In order to give names to the various parts of the structure corresponding to these colors, we introduce some additional notation. For any $k$-tuple $b \in V(G)^k$, let $\chi_r(G;b)$ be the color of $b$ after the $r$-th round of $k$-WL. Now let $b$ be any $k$-tuple in $V(G)^k$ whose color is $t_{r,i} = \chi_r(G;b)$. Define a function $f$ such that $f(i)$ is the unique index of the color of $b$ at the previous round, i.\,e. $\chi_{r-1}(G;b) = t_{r-1,f(i)}$. Then by definition, $t_{r,i} = (t_{r-1,f(i)}, M_i)$ where $M_i$ is the following multiset of $k$-tuples of colors from round $r-1$:
\[
M_i := \{\!\{(\chi_{r-1}(b^{1v}), \chi_{r-1}(b^{2v}), \dotsc, \chi_{r-1}(b^{kv})) : v \in V(G) \}\!\}.
\]
We now use this notation to help us define the next layers of $C_{G,k,r}$. First, for each tuple $\tau$ that occurs in any $M_i$ (note: $\tau$ is a $k$-tuple of colors from round $r-1$), we define
\[
\phi'_{\tau}(a,v) := \bigwedge_{j \in [k]} \bigwedge_{\ell \in [m_{r-1}]} \left( X'_{r-1,\ell}(a^{jv}) \leftrightarrow [\tau_j = t_{r-1,\ell}] \right)
\]
Note that the indicator function $[\tau_j = t_{r-1,\ell}]$ is a constant independent of the input $H$ (though it does depend on $r,j,\ell$), so, as before, the entire expression in parentheses is either $X'_{r-1,\ell}(a^{jv})$ or $\neg X'_{r-1,\ell}(a^{jv})$ depending on the value of that constant. 

Next let us use the notation $M_i(\tau) \in \mathbb{Z}_{\geq 0}$ to denote the multiplicity of the $k$-tuple $\tau$ in $M_i$. We write $\tau \in M_i$ iff $M_i(\tau) > 0$. 
\[
X'_{r,i}(a) = X'_{r-1,f(i)}(a) \wedge \bigwedge_{\tau \in M_i} \left( M_i(\tau) = \sum_{v \in V(H)} \phi'_{\tau}(a,v) \right).
\]
As before, let $c_{r,i}$ denote the number of $k$-tuples in $G$ that have color $t_{r,i}$ at round $r$. Then the output gate of $C_{G,k,r}$ is
\[
C_{G,k,r}(H) := \bigwedge_{i=1}^{m_r} \left(c_{r,i} = \sum_{a \in V(H)^k} X'_{r,i}(a) \right)
\]

    The circuit is square-symmetric because any permutation on inputs is in essence a permutation of labelings of vertices or vertices within a $k$-tuple of $H$, and this visibly induces an automorphism of the circuit (it acts in the obvious way on each of the above formulas). That is, we have fixed $G$, and then the circuit computes the multiplicities in $H$ of the $k$-WL colors that occurred in $G$.    
\end{proof}

\begin{lemma}\label{lem:bool-to-alg}
    Any Boolean threshold circuit $C$ which is rigid square-symmetric under the action of any group $G$ can be converted to an algebraic circuit which computes the same output of $C$ and is rigid square-symmetric under the action of $G$.
\end{lemma}
\begin{proof}
    The construction uses the standard folklore conversion of replacing AND gates with multiplication gates, NOT gates with $1-x$ gates, and threshold gates with appropriate interpolation gadgets (take the sum of the inputs, then apply Lagrange interpolation to get an algebraic circuit that outputs 1 if that sum is greater than the threshold and 0 if it is less). Replacing gates by gadgets in this way preserves rigidity and square-symmetry between structures in the algebraic circuit. Within the gadgets there are no twins and permuting inputs will still correspond to at least one automorphism since the structure is otherwise preserved. 
\end{proof}

\begin{proof}[Proof of Thm.~\ref{thm:WLtoSep}]

The modified Grohe--Verbitsky circuit still computes the same output as the original Grohe--Verbitsky circuit, so by Lem.~\ref{lem:mGVcircuit} the modified Grohe--Verbitsky circuit is a rigid square-symmetric circuit, and furthermore by Thm.~\ref{thm:GVcircuit} it also simulates WL. We can convert the modified Grohe--Verbitsky circuit into a rigid square-symmetric algebraic circuit $C^*$ which computes the same output by Lem.~\ref{lem:bool-to-alg}. $C^*$ is a single-output rigid square-symmetric circuit, so in particular it computes an invariant polynomial. This completes the proof.
\end{proof}

Combining Thm.~\ref{thm:SeptoWL} and Thm.~\ref{thm:WLtoSep}, we get the equivalence:

\begin{corollary}\label{col:equivalence}
    For any $1/\log_2 n<\varepsilon<1$ and $k \leq \varepsilon n^{1-\varepsilon}$, $r$ rounds of $k$-WL is simulated by separating modules of symmetric algebraic circuits of size $n^{O(k)}$ and depth $O(r)$, which in turn are simulated by $O(rk/\varepsilon)$ rounds of $O(k/\varepsilon)$-WL. In particular, when $\varepsilon$ is constant, then $\Theta(r)$ rounds of $\Theta(k)$-WL is equivalent in distinguishing power to separating modules of symmetric algebraic circuits of size $n^{\Theta(k)}$ and depth $\Theta(r)$.

    The same results hold if everywhere we replace ``separating module'' by ``separating invariant.''
\end{corollary}

\section{Multiplicity obstructions for graph non-isomorphism} \label{sec:multiplicity}
\subsection{Additional preliminaries for this section}
With one exception towards the end of this background section, that we will clearly mark, all of the material here is completely standard; see, for example, any standard textbook on group or representation theory, e.\,g., \cite{Fulton-Harris, Dummit-Foote}. 

\paragraph{Representations of direct products.}
Given groups $\Gamma, \Delta$ and representations $\rho \colon \Gamma \to \GL(V)$ and $\sigma \colon \Delta \to \GL(U)$, we get a representation $(\rho \boxtimes \sigma) \colon \Gamma \times \Delta \to \GL(V \otimes U)$ defined by
\[
(\rho \boxtimes \sigma)((g,h)) = \rho(g) \otimes \sigma(h).
\]
We use $\Irr(\Gamma)$ to denote the set of equivalence classes of irreps. If the characteristic of $\F$ is coprime to $|\Gamma \times \Delta|$, then 
\[
\Irr(\Gamma \times \Delta) = \{\rho \boxtimes \sigma : \rho \in \Irr(\Gamma), \sigma \in \Irr(\Delta)\}.
\]

\paragraph{Representations of subgroups; induction and restriction.}
Lastly, given a group $\Gamma$ and a subgroup $\Sigma \leq \Gamma$, if $(\rho,V)$ is a representation of $\Gamma$, $\Res^\Gamma_{\Sigma} V$ denotes the same vector space $V$, considered as a $\Sigma$-module ``by restriction,'' that is, the $\Sigma$-representation associated to $\Res^\Gamma_{\Sigma} V$ is precisely $\rho|_{\Sigma} \colon \Sigma \to \GL(V)$. A straightforward exercise is that if $K \leq \Delta$, $V$ is a $\Gamma$-module and $U$ is an $\Delta$-module, then $\Res^{\Gamma \times \Delta}_{\Sigma \times \Delta} (V \boxtimes U) = (\Res^\Gamma_{\Sigma} V) \boxtimes (\Res^\Delta_{K} U)$.

In the opposite direction, if $(\rho, V)$ is a representation of the subgroup $\Sigma \leq \Gamma$, the \emph{induced representation} $\Ind_\Sigma^\Gamma V$ is a representation of $\Gamma$ on the vector space $\F[\Gamma] \otimes_{\F[\Sigma]} V$, that is, we start with the vector space $\F[\Gamma] \otimes_\F V$, on which $\Gamma$ acts by $\gamma \cdot (\gamma' \otimes v) := (\gamma \gamma') \otimes v$, and then we mod out by the relations $(\gamma \sigma) \otimes v = \gamma \otimes (\sigma \cdot v)$ for all $\gamma \in \Gamma, \sigma \in \Sigma$. In particular, $\dim \Ind_\Sigma^\Gamma V = |\Gamma/\Sigma| \cdot \dim V$. 

\paragraph{Character table.} It is a standard theorem that the number of irreps up to isomorphism is equal to the number of conjugacy classes. This lets us arrange all the characters of all irreps into a square table with side length equal to the number of conjugacy classes of $\Gamma$, called the \emph{character table}. We index its columns by conjugacy classes and index its rows by the irreps of $\Gamma$. For a column corresponding to the conjugacy class of $g \in \Gamma$ and a row corresponding to the irrep $V$, we put the character value $\chi_V(g)$ in the $(V,[g])$ entry of the table (recall that this value is the same for any choice of element in the conjugacy class $[g]$).  
\begin{theorem}[{See, e.\,g., \cite[Theorem~2.12]{Fulton-Harris}}] \label{thm:char-table-invertible}
    The character table is an invertible matrix.
\end{theorem} 

\paragraph{Inner products of characters.} The characters for any representation have a natural notion of an inner product $$\langle\chi,\chi'\rangle_\Gamma=\frac{1}{|\Gamma|}\sum_{g\in \Gamma}\chi(g)\overline{\chi'(g)}$$ where the bar denotes complex conjugation. Recall that the characters for $S_n$ are all integer-valued (see Section~\ref{sec:prelim:Sn}), so complex conjugation does not affect the inner product in our setting. If $\chi$ is the character of an irrep, then the inner product $\langle \chi, \chi' \rangle$ is equal to the multiplicity of $\chi$ in $\chi'$ (e.\,g., \cite[Corollary~2.16]{Fulton-Harris}).

We will use the following theorem

\begin{theorem}[{Frobenius Reciprocity, see, e.\,g., \cite[Corollary~3.20]{Fulton-Harris}}] \label{thm:frobenius}
\[
    \langle\Ind_\Sigma^{\Gamma}\chi_{\rho'}, \chi_\rho \rangle_\Gamma = \langle\chi_{\rho'},\Res^{\Gamma}_\Sigma\chi_\rho\rangle_\Sigma.
    \]
\end{theorem}

The character of most interest to us in this section is the character of the trivial representation, which we will denote $\chi_{\bf 1}$. The trivial representation ${\bf 1}:\Gamma\to \GL_1(\mathbb{C})$ sends all group elements to the 1-dimensional identity matrix. Thus, $\chi_{\bf 1}(g)=1$ for all $g\in \Gamma$. The trivial representation (as with all 1-dimensional representations) is an irrep.

The one straightforward fact that we will use that might not be covered in standard textbooks (but is not beyond the difficulty of an exercise) is the following. The coset representation $\F[S_n / \Gamma]$ is equal to the induced representation in $S_n$ of the trivial representation in $\Gamma$, i.\,e. 
\[
\F[S_n / \Gamma] \cong\Ind^{S_n}_\Gamma{\bf 1}.
\]

\subsection{Computing multiplicities}
We begin this section by showing that the representation-theoretic multiplicities supported on a graph's orbit, up to degree $d$, can be computed in $n^{O(d)}$ time. (Or up to support-degree $d$ in $2^{O(d^2)} n^{O(d)}$ time.)

\begin{theorem} \label{thm:algorithm-multiplicities}
Let $\F$ be a field of characteristic $0$ or $> n$. There is an algorithm which, given a graph $G$ and a positive integer $d$, computes the multiplicities of every $S_n$-irrep that occurs in $\F[S_n \cdot  G]$ up to degree $d$ in time $n^{O(d)}$.

In particular, there is an algorithm that, given two $n$-vertex graphs $G,H$ and an integer $d$ as input, will find a multiplicity obstruction in degree $\leq d$ if one exists, or report that none does, in time $n^{O(d)}$.
\end{theorem}

As with Lemma~\ref{lem:spanning}, an analogous result holds with degree replaced by support-degree, but with runtime $2^{O(d^2)} n^{O(d)}$ (the number of monomials of support-degree at most $d$). 

\begin{proof}
The third paragraph of the proof of Theorem~\ref{thm:algorithm} already says directly that we can compute all irreps in $\F[X]$ up to degree $d$ that vanish on $G$ (resp., $H$), as well as which ones are equivalent to which other ones as abstract representations. We can then apply the algorithm of Brooksbank \& Luks \cite{BL} again to tell which irreps that vanish on $G$ are equivalent to which ones that vanish on $H$. At that point, the algorithm can directly check if any of the multiplicities differ.
\end{proof}

\subsection{Characterizing multiplicities}
\begin{definition}[Cycle index]
    The \emph{cycle index} of a permutation group $G \leq S_n$ is the multivariate generating function for the multiset of its cycle types. That is, if $c_i(\pi)$ denotes the number of $i$-cycles in $\pi \in S_n$, the cycle index of $G$ is
    \[
    cyc_G(x_1,\dotsc,x_n) := \sum_{\pi \in G} x_1^{c_1(\pi)} x_2^{c_2(\pi)} \dotsb x_n^{c_n(\pi)}.
    \]
\end{definition}

The main result of this section is:

\begin{theorem} \label{thm:cycle-index}
For all $G \leq S_n$, the collection of all multiplicities of $S_n$-representations in the permutation representation $S_n / G$ determine and are determined by the cycle index of $G$.
\end{theorem}

We first prove an immediate corollary about the power of multiplicities to distinguish non-isomorphic graphs:

\begin{corollary} \label{cor:multiplicity}
Two $n$-vertex graphs $G$ and $H$ are separated by a multiplicity obstruction if and only if the cycle indices of $\Aut(G) \leq S_n$ and $\Aut(H) \leq S_n$ are distinct.
\end{corollary}

\begin{proof}
Let $R = \F[x_{ij}] / \langle x_{ij}^2 - x_{ij} \rangle$ and let $I_G \subseteq R$ be the vanishing ideal of the orbit $S_n \cdot G$ (and similarly for $I_H$). Then for any irrep $\chi$, we have $\mult_\chi(\F[S_n \cdot G]) = \mult_\chi(R) - \mult_\chi(I_G)$, so the multiplicities of irreps that vanish on $S_n \cdot G$ give equivalent information to the multiplicities that occur in the coordinate ring of the orbit, $\F[S_n \cdot G]$. Now, because the orbit is a finite set, and the Orbit--Stabilizer Theorem gives an $S_n$-equivariant isomorphism $S_n \cdot G \to S_n / \Aut(G)$ (the latter being the coset space), as $S_n$ representations we have $\F[S_n \cdot G]$ and $\F[S_n / \Aut(G)]$ are equivalent: $\mult_\chi \F[S_n \cdot G] = \mult_\chi \F[S_n / \Aut(G)]$. By Theorem~\ref{thm:cycle-index}, the latter multiplicities contain the same information as the cycle index of $\Aut(G)$. Thus $G$ and $H$ are distinguished by $S_n$-multiplicities in their vanishing ideals $I_G,I_H$ if and only if $\Aut(G)$ and $\Aut(H)$ have distinct cycle indices.
\end{proof}

\begin{proof}[Proof of Theorem~\ref{thm:cycle-index}]
    Let $G,H \leq S_n$ be two subgroups of $S_n$. Let $u$ be the column vector indexed by cycle types, where for each cycle type $c$, $u_c$ is the number of elements of $G$ of cycle type $c$. We will refer to $u$ as the ``cycle index vector'' of $G$. Similarly, let $v$ be the cycle index vector of $H$. By construction, the cycle index vector and the cycle index polynomial ``contain the same data'' (just represented by different types of objects).

    We first show that if either the multiplicities or the cycle indices are the same, then $|G|=|H|$. Note that $\sum_c u_c = |G|$, so if the cycle indices are the same, then $|G| = \sum_c u_c = \sum_c v_c = |H|$. On the other hand, if the multiplicities of $S_n$-irreps in $\F[S_n / G]$ are the same as those in $\F[S_n / H]$, then we have $|S_n / G| = \dim \F[S_n / G] = \dim \F[S_n /H] = |S_n / H|$, and therefore $|G|=|H|$.

    Now we show how to recover the multiplicities from the cycle index vectors (and then at the end, because this uses the character table and the character table is invertible, we will get the other direction almost for free). 
    Let $\chi$ be an irreducible character of $S_n$ for which we wish to know its multiplicity in $\F[S_n / G]$. As a representation, the latter is the induced representation $\mathrm{Ind}_G^{S_n} \chi_1$, where $\chi_1$ is the trivial character of $G$. So we seek the inner product $\langle \chi, \mathrm{Ind}_G^{S_n}\chi_1 \rangle_{S_n}$, where we have used the subscript $S_n$ to remind the reader that here this is the inner product of $S_n$-characters. By Frobenius Reciprocity (Theorem~\ref{thm:frobenius}), we have
    \[
    mult_\chi(\F[S_n / G]) = \langle \chi, \mathrm{Ind}_G^{S_n}\chi_1 \rangle_{S_n} = \langle \mathrm{Res}^{S_n}_G \chi, \chi_1 \rangle_G.
    \]
    To calculate the latter, note that $\chi_1$ on $G$ is just the constant-$1$ function, and $\mathrm{Res}^{S_n}_G(\chi)(g) = \chi(g)$, which is completely determined by the cycle type of $G$. Thus we seek to compute
    \[
    \langle \mathrm{Res}^{S_n}_G \chi, \chi_1 \rangle_G = \frac{1}{|G|} \sum_{g \in G} \chi(g).
    \]
    Since $\chi(g)$ only depends on the cycle type of $g$, the latter sum is a scaled version of summing the $\chi$-th row of the character table $M$ of $S_n$, with each entry weighted by the number of elements of $G$ with that cycle type. Thus we have that the $\chi$-th entry of $Mu$ (where $u$ is the cycle index vector of $G$) is the multiplicity of $\chi$ in $\F[S_n / G]$:
    \[
    mult_\chi(\F[S_n / G]) = \frac{1}{|G|}\left(Mu\right)_\chi
    \]

    Putting these together, we get that the vector $\frac{1}{|G|}Mu$ is the vector of all the multiplicities of $S_n$ irreps in $\F[S_n / G]$. Thus, if $G$ and $H$ have the same cycle index, then $u=v$ and (by the first paragraph above) $|H|=|G|$, so the vector of multiplicities $(1/|G|) Mu$ for $G$ is equal to $(1/|H|)Mv$ for $H$.

    Conversely, if the multiplicity vectors are equal, then we have $(1/|G|) Mu = (1/|H|) Mv$. As $|G|=|H|$ (again by the first paragraph above), and $M$ is invertible (Theorem~\ref{thm:char-table-invertible}), we get that $u=v$, and thus $G$ and $H$ have the same cycle index. \qedhere

\end{proof}

\subsection{Limitations of multiplicity obstructions}
\begin{observation} \label{obs:multiplicity-nogo}
Almost all pairs of graphs are not distinguishable by multiplicities but are distinguishable by separating modules. The same is true for almost all pairs of random regular graphs of degree $3 \leq d \leq n-4$.
\end{observation}

\begin{proof}
Since all orbits of $S_n$ are finite, any two distinct orbits (=non-isomorphic graphs) are in fact separated by invariants (e.\,g., \cite[Thm.~2.3.6]{DerksenKemper}), which is a particular case of separating modules. On the other hand, almost all graphs have trivial automorphism group \cite{ErdosRenyiAsymmetric}, so by Cor.~\ref{cor:multiplicity} any two such graphs are not distinguished from one another by multiplicity information alone. For random regular graphs, we use the analogous result on automorphism groups in that setting \cite{KSVAsymmetric}.
\end{proof}

After the above rather negative answer, a natural question arises: is having the same (or conjugate) automorphism groups the only way that multiplicity obstructions can fail to exist?  That is, is it the case that every two $n$-vertex graphs whose automorphism groups are not conjugate in $S_n$ can be distinguished by multiplicity obstructions? Here we give a negative answer to this question:

\begin{proposition}\label{prop:nonisocycleindex}
For infinitely many $n$, there exist pairs of subgroups $\Gamma_n, \Gamma'_n \leq S_n$ that are not isomorphic (hence, \emph{a fortiori}, not conjugate), yet have the same cycle indices and both arise as automorphism groups of $n$-vertex graphs.
\end{proposition}

In fact, the permutation groups we exhibit in the proof are \emph{regular}, in the sense of being transitive and free (the stabilizer of any element of $[n]$ is trivial), or equivalently, they are transitive subgroups of $S_n$, and have order exactly $n$.

\begin{proof}
For a group $G$ of order $n$, let $L_G \leq S_n$ be its left regular representation. If $g \in G$ has order $k$, then its action in the left regular representation consists precisely of $n/k$ many $k$-cycles (which are the $n/k$ cosets of the subgroup $\langle g \rangle$). 
In particular, if $G$ is a $p$-group of exponent $p$, then all non-identity elements have order $p$, and so they all have the same cycle type. Thus, the left regular representations of any two $p$-groups of order $n$ and exponent $p$ will have the same cycle indices. 

To ensure that we find $G$ such that $L_G$ is the automorphism group of a graph, we use the theorem classifying which groups have so-called graphical regular representations (GRR) \cite{Godsil}, i.\,e., where $G$ acts regularly on a graph and is the complete automorphism group of that graph. In our particular construction, since we use $p$-groups, they are, in particular, solvable, so although a full classification of GRR groups is available \cite{Godsil}, we can in fact rely on earlier results that focused on solvable groups. By \cite{Watkins,Hetzel}, any solvable group admits a GRR unless it is one of: (a) an Abelian group of exponent $> 2$, (b) a generalized dicyclic group (that is, has an abelian normal subgroup $N \unlhd G$ of index 2 and an element $x\notin N$ of order $4$ such $xnx^{-1} = n^{-1}$ for all $n \in N$), or (c) one of 13 explicit exceptions of order at most 32 (see \cite{Godsil} for an overview of the history and list of exceptions with further citations).

To avoid (a), we can simply choose $p$-groups that are non-Abelian. To avoid (b), if we choose $p$ odd, then by Cayley's Theorem $G$ cannot have an element of order 4. Avoiding (c) only involves avoiding finitely many groups. By constructions of Higman \cite{Higman} and Sims \cite{Sims}, there are $p^{\Omega(k^3)}$ (non-Abelian) $p$-groups of order $p^k$ and exponent $p$, furnishing us with many examples at all sufficiently large prime powers $n$, and in particular, for infinitely many $n$.
\end{proof}

We also explored by exhaustive computer search for small(er) groups for which the above construction using the left regular representation yielded examples. The smallest we found were of order 32:

\begin{example}\label{ex:nonisocycleindex}
    The smallest GRR groups with the property in Proposition~\ref{prop:nonisocycleindex} occur at order 32. They are listed below in 9 sets. Each set is grouped by cycle index, i.\,e., every group within any set all have the same cycle index. They are listed using their structure description with a link to the GroupNames.org database \cite{groupnames}, as well as their ID in the Small Groups Library (with order omitted since they are all of order 32).

    \begin{itemize}
        \item \href{https://people.maths.bris.ac.uk/~matyd/GroupNames/1/C2.C4%5E2.html}{$C_2.C_4^2$} (2), \href{https://people.maths.bris.ac.uk/~matyd/GroupNames/1/C4%5E2sC2.html}{$C_4^2\rtimes C_2$} (24), \href{https://people.maths.bris.ac.uk/~matyd/GroupNames/1/C2%5E2sQ8.html}{$C_2^2\rtimes Q_8$} (29), \href{https://people.maths.bris.ac.uk/~matyd/GroupNames/1/C4%5E2s2C2.html}{$C_4^2\rtimes_2C_2$} (33)
        \item \href{https://people.maths.bris.ac.uk/~matyd/GroupNames/1/C8sC4.html}{$C_8\rtimes C_4$} (4), \href{https://people.maths.bris.ac.uk/~matyd/GroupNames/1/C4.10D4.html}{$C_4._{10}D_4$} (8), \href{https://people.maths.bris.ac.uk/~matyd/GroupNames/1/C4sC8.html}{$C_4\rtimes C_8$} (12)
        \item \href{https://people.maths.bris.ac.uk/~matyd/GroupNames/1/C2%5E2sC8.html}{$C_2^2\rtimes C_8$} (5), \href{https://people.maths.bris.ac.uk/~matyd/GroupNames/1/C2xM4(2).html}{$C_2\times M_4(2)$} (37), \href{https://people.maths.bris.ac.uk/~matyd/GroupNames/1/C8oD4.html}{$C_4\circ D_4$} (38)
        \item \href{https://people.maths.bris.ac.uk/~matyd/GroupNames/1/C2%5E3sC4.html}{$C_2^3\rtimes C_4$} (6), \href{https://people.maths.bris.ac.uk/~matyd/GroupNames/1/C4xD4.html}{$C_4\times D_4$} (25), \href{https://people.maths.bris.ac.uk/~matyd/GroupNames/1/C2%5E2.D4.html}{$C_2^2.D_4$} (30), \href{https://people.maths.bris.ac.uk/~matyd/GroupNames/1/C4.4D4.html}{$C_4._4D_4$} (31), \href{https://people.maths.bris.ac.uk/~matyd/GroupNames/1/ES-(2,2).html}{$2_-^{1+4}$} (50)
        \item \href{https://people.maths.bris.ac.uk/~matyd/GroupNames/1/D4sC4.html}{$D_4\rtimes C_4$} (9), \href{https://people.maths.bris.ac.uk/~matyd/GroupNames/1/C2xSD16.html}{$C_2\times SD_{16}$} (40), \href{https://people.maths.bris.ac.uk/~matyd/GroupNames/1/C4oD8.html}{$C_4\circ D_8$} (42)
        \item \href{https://people.maths.bris.ac.uk/~matyd/GroupNames/1/Q8sC4.html}{$Q_8\rtimes C_4$} (10), \href{https://people.maths.bris.ac.uk/~matyd/GroupNames/1/C4.Q8.html}{$C_4.Q_8$} (13)
        \item \href{https://people.maths.bris.ac.uk/~matyd/GroupNames/1/C4wrC2.html}{$C_4\wr C_2$} (11), \href{https://people.maths.bris.ac.uk/~matyd/GroupNames/1/C8.C2%5E2.html}{$C_8.C_2^2$} (44)
        \item \href{https://people.maths.bris.ac.uk/~matyd/GroupNames/1/C2xC2%5E2sC4.html}{$C_2\times C_2^2\rtimes C_4$} (22), \href{https://people.maths.bris.ac.uk/~matyd/GroupNames/1/C4sD4.html}{$C_4\rtimes D_4$} (28), \href{https://people.maths.bris.ac.uk/~matyd/GroupNames/1/C2xC4oD4.html}{$C_2\times C_4\circ D_4$} (48)
        \item \href{https://people.maths.bris.ac.uk/~matyd/GroupNames/1/C2%5E2wrC2.html}{$C_2^2\wr C_2$} (27), \href{https://people.maths.bris.ac.uk/~matyd/GroupNames/1/C4s1D4.html}{$C_4\rtimes_1D_4$} (34), \href{https://people.maths.bris.ac.uk/~matyd/GroupNames/1/ES+(2,2).html}{$2_+^{1+4}$} (49)
    \end{itemize}
\end{example}

If we do not limit ourselves to the left regular representation, we find that the smallest $n$ such that $S_n$ contains a pair of non-conjugate subgroups of the same cycle index is $n=8$,\footnote{$\langle (1,8,5,4)(2,3,7,6), (2,7)(3,6), (1,6)(2,8)(3,5)(4,7)\rangle \cong C_4 \circ D_4$, where $\circ$ denotes a central product, and $\langle (2,7)(3,6), (1,7,4,6)(2,8,3,5)\rangle \cong C_2^2 \rtimes C_4$ are not even abstractly isomorphic---hence, \emph{a fortiori}, not conjugate---but have the same cycle indices.} though we do not know if there exist such small examples that are also the automorphism groups of simple, uncolored graphs. At $n=16$ we found an example that are automorphism groups of directed graphs, but (according to the GRR Theorem) not undirected graphs.

Next, as pointed out in the introduction, GCT motivates us to ask whether all graphs that are characterized by their symmetries are distinguishable by multiplicities. Here we show that in constant degree, only finitely many graphs can be distinguished by multiplicities:

\begin{proposition} \label{prop:FI}
For any fixed degree $d$, as $n \to \infty$, the multiplicities of irreps in degree $\leq d$ can distinguish at most finitely many pairs of non-isomorphic graphs. In particular, there is no constant $d$ such that multiplicity information in degree $\leq d$ distinguishes all pairs of non-isomorphic graphs that are characterized by their symmetries.
\end{proposition}

\begin{proof}
The polynomials of degree $d$ in the polynomial ring $\F[x_{ij} : i,j \in [n]]$ are, as a representation of $S_n$, $\mathrm{Sym}^d(V_n \otimes V_n)$, where $V$ is the defining $n$-dimensional representation of $S_n$. By \cite[Prop~3.4.1]{CEF}, the sequence $V_n$ is the sequence of representations of a finitely generated FI-module. Then by \cite[Thm.~3.4.2]{CEF}, the same is true of $V_n \otimes V_n$, and then by \cite[Prop.~3.4.3]{CEF} the same is true of $\mathrm{Sym}^d(V_n \otimes V_n)$ for any fixed value of $d$. By \cite[Thm.~1.13]{CEF}, this sequence of representations is thus ``uniformly representation stable'' in the sense of Church \& Farb \cite{ChurchFarb}. Part of the definition of the latter is that there is a finite list of partitions $\lambda_1,\dotsc,\lambda_k$ and multiplicities $m_1,\dotsc,m_k \in \mathbb{N}$---all of which depend only on $d$, not on $n$---such that  for all sufficiently large $n$, $\mathrm{Sym}^d(V_n \otimes V_n) = \bigoplus_{i=1}^k S_{\lambda_i[n]}^{\oplus m_i}$, where $\lambda[n]$ is the extension of $\lambda$ to a partition of $n$ by increasing the first part of the partition so the whole partition adds up to $n$. In particular, the total multiplicity of all irreps in $\mathrm{Sym}^d(V_n \otimes V_n)$ is $\sum_{i=1}^k m_i$, which is a constant that depends only on $d$, not on $n$. Thus the total number of possibilities for the multiplicities of each irrep in degree $d$ is $\prod_{i=1}^k (m_i+1)$, which is again a constant that depends on $d$ but not on $n$, so the multiplicities in degree $d$ (and hence, \emph{a fortiori}, in degree $\leq d$) can partition all isomorphism classes of graphs into only $f(d)$ sub-classes, where $f$ is a function that depends only on $d$, not on $n$. 

Those multiplicities are an upper bound on the multiplicities in the quotient ring $\F[x_{ij}] / \langle x_{ij}^2 - x_{ij} \rangle$, so we similarly get a bound there that depends only on $d$ and not on $n$. (We remark that the family of ideals $\langle x_{ij}^2 - x_{ij} : i,j \in [n] \rangle$ is a finitely generated FI-ideal, so the above argument also applies directly to the quotiented ring as well.)

For the ``in particular,'' we recall that a disjoint union of $k$ many $\ell$-cliques is characterized by its symmetries (for any $k,\ell \geq 1$). For any factorization $n=k\ell$ we get such a graph, and graphs with different pairs $(k,\ell)$ are non-isomorphic (and not isomorphic to one another's complements either). For example, for $n$ a product of $\mu$ distinct primes, there are $2^{\mu}$ different pairs $(k,\ell)$ with $k\ell = n$, hence $2^\mu$ mutually non-isomorphic $n$-vertex graphs that are characterized by their symmetries. Choosing $\mu > \log_2 f(d)$ we find that these graphs that are characterized by their symmetries cannot all be distinguished from one another by multiplicities up to degree $d$.
\end{proof}

In the notation of the above proof, we see that we could get a more explicit lower bound on the degree than merely $\omega(1)$ if we could estimate the growth rate of $f(d)$ as a function of $d$.

\subsection{Multiplicity obstructions are stronger than occurrence obstructions}
By analogy with occurrence obstructions in GCT, given two $n$-vertex graphs $G,H$, we say that an irrep $V$ of $S_n$ is an \emph{occurrence obstruction} if \[
\mult_V(\F[S_n / \Aut(G)]) = 0 \neq \mult_V(\F[S_n / \Aut(H)])\]
(or vice versa). (The reason for the name is that the occurrence of $V$ in the coordinate ring of the orbit of $H$ obstructs an isomorphism with $G$.)
By the ``occurrence information'' for a group $\Gamma \leq S_n$ we mean the set (not multiset) of which equivalence classes of irreps occur in $\F[S_n / \Gamma]$ with non-zero multiplicity.

\begin{lemma} \label{lem:occurrence-inheritance}
Suppose $\Gamma_1, \Gamma_2 \leq S_k$; for $i=1,2$, let $\Gamma_i^{[n]}$ be the group $S_{n-k} \times \Gamma_i \leq S_n$. Then 
\begin{enumerate}
\item $\Gamma_1,\Gamma_2 \leq S_k$ have the same multiplicity information if and only if $\Gamma_1^{[n]},\Gamma_2^{[n]} \leq S_n$ have the same multiplicity information; and

\item If $\Gamma_1,\Gamma_2 \leq S_k$ have the same occurrence information, then $\Gamma_1^{[n]},\Gamma_2^{[n]} \leq S_n$ have the same occurrence information.
\end{enumerate}
\end{lemma}

\begin{proof}
\begin{enumerate}
\item We use the characterization in terms of cycle indices (Theorem~\ref{thm:cycle-index}). The cycle indices of $\Gamma_i$ determine and are determined by those of $\Gamma_i^{[n]}$. More specifically, when viewed as multivariate generating functions, the cycle index of $G \times H$ is the product of the cycle indices of $G$ and $H$. Writing $c(G)$ for the cycle index polynomial, we then we have $c(\Gamma_i^{[n]}) = c(\Gamma_i) c(S_{n-k})$, from which it follows immediately that $c(\Gamma_1) = c(\Gamma_2) \Longleftrightarrow c(\Gamma_1^{[n]}) = c(\Gamma_2^{[n]})$.

\item By Frobenius reciprocity, an irrep of $S_n$ contains a $\Gamma$-invariant if and only if it occurs with non-zero multiplicity in $\F[S_n / \Gamma]$. We will show that each irrep of $S_n$ contains a $\Gamma_1^{[n]}$ invariant if and only if it contains a $\Gamma_2^{[n]}$ invariant. Toward that end, let $V$ be an irrep of $S_n$. Suppose the restriction of $V$ to $\Gamma_1^{[n]}$ contains an invariant. Since $\Gamma_i^{[n]} = S_{n-k} \times \Gamma_i \leq S_{n-k} \times S_k$, we will first examine the restriction of $V$ to $S_{n-k} \times S_k$. Suppose $\mathrm{Res}^{S_n}_{S_{n-k} \times S_k} V = \bigoplus_i U_i \boxtimes V_i$, where each $U_i$ is an irrep of $S_{n-k}$ and each $V_i$ is an irrep of $S_k$ (recall that irreps of a direct product are tensor products of irreps of the factors, see Section~\ref{sec:prelim:rep}). 
Then 
\begin{align*}
\mathrm{Res}^{S_n}_{\Gamma_{i}^{[n]}} V & =  \mathrm{Res}^{S_{n-k} \times S_k}_{\Gamma_i^{[n]}}\left(\mathrm{Res}^{S_n}_{S_{n-k} \times S_k} V\right) \\
 & = \mathrm{Res}^{S_{n-k} \times S_k}_{\Gamma_i^{[n]}}\left(\bigoplus_j U_j \boxtimes V_j\right) \\
 & = \bigoplus_j \mathrm{Res}^{S_{n-k} \times S_k}_{\Gamma_i^{[n]}}\left(U_j \boxtimes V_j\right) \\
 & = \bigoplus_j \mathrm{Res}^{S_{n-k}}_{S_{n-k}}U_j \boxtimes \mathrm{Res}^{S_k}_{\Gamma_i} V_j \\
 & = \bigoplus_j U_j \boxtimes \mathrm{Res}^{S_k}_{\Gamma_i} V_j,
 \end{align*}
where the second to last line follows from the basic fact about restrictions of direct products (see Section~\ref{sec:prelim:rep}).

Thus, the trivial $\Gamma_i^{[n]}$ representation occurs in $\mathrm{Res}^{S_n}_{\Gamma_i^{[n]}} V$ if and only if there is some $j$ such that (1)  $U_j$ is the trivial rep of $S_{n-k}$ and (2) $\mathrm{Res}^{S_k}_{\Gamma_i} V_j$ contains the trivial $\Gamma_i$ representation. 

But now, by our assumption that the occurrence information for $\Gamma_1,\Gamma_2 \leq S_k$ are the same, we have that $\mathrm{Res}^{S_k}_{\Gamma_1} V_j$ contains a trivial representation if and only if $\mathrm{Res}^{S_k}_{\Gamma_2} V_j$ does. Combining with the preceding paragraph, that gives us that $V$ contains a $\Gamma_1^{[n]}$ invariant if and only if it contains a $\Gamma_2^{[n]}$ invariant. \qedhere
\end{enumerate}
\end{proof}

Although the converse of Lemma~\ref{lem:occurrence-inheritance}(2) may well be true, we do not see how to prove it along the above lines. (The issue is that for each $S_n$-irrep $V$ with a $\Gamma_i^{[n]}$-invariant, all we conclude from assuming that $\Gamma_1^{[n]}$ and $\Gamma_2^{[n]}$ have the same occurrence information is that there exists a $j$ such that $V_j$ contains a $\Gamma_i$-invariant, but it is possible that this occurs in different $V_j$ for for $i=1,2$.)

\begin{corollary}
If $(G,H)$ is a pair of $k$-vertex graphs that are separated by multiplicities but not occurrence obstructions, then $G \sqcup K_{n-k}$ and $H \sqcup K_{n-k}$ are $n$-vertex graphs that are separated by multiplicities but not occurrence obstructions, for all $n > 2k$.  
\end{corollary}

\begin{proof}
If $n > 2k$, then $\Aut(G \sqcup K_{n-k}) = \Aut(G) \times S_{n-k}$, since $G$ is too small to contain a clique of size $n-k$. The result then follows immediately from Lemma~\ref{lem:occurrence-inheritance}.
\end{proof}

Thus finding infinitely many such examples reduces to finding a single example. For the sake of interest, we give a different (albeit related) construction of an infinite family of pairs of graphs that are separated by multiplicities but not occurrence obstructions. Our construction in the next proposition also has the advantage that it does not have the ``gap'' between $k$ and $2k$ that would arise by applying the preceding corollary.

\begin{proposition} \label{prop:mult-vs-occurrence}
For all $n \geq 6$, there are pairs of simple undirected graphs $(G_n, H_n)$ on $n$ vertices such that $G_n$ and $H_n$ are separated by multiplicity obstructions but not by occurrence obstructions.
\end{proposition}

For simple undirected graphs, the value of $6$ is optimal: by explicit calculation, we found that there are no pairs of simple undirected graphs on $n \leq 5$ vertices that are separated by multiplicities but not separated by occurrences. 

\begin{proof}
We exhibit one infinite family of pairs of graphs. Let $G_n = P_4 + [n-4]$ be the disjoint union of a 4-vertex path (length 3) and $n-4$ isolated vertices. We define $H_n$ to be a large ``dandelion'': a path on 4 vertices, with one endpoint connected to $n-4$ vertices of degree 1 (which we refer to as the ``seeds'' of the dandelion). One may think of $G_n$ as being $H_n$ after the seeds of the dandelion were blown off.

Then we have $\Aut(G_n) = C_2 \times S_{n-4}$, where $C_2$ acts on the 4-vertex path by reversing it, which has cycle type $(\bullet \bullet)(\bullet \bullet)$. For concreteness, we suppose the path to be on vertices $n-3,n-2,n-1,n$, in order. 

Similarly, for $n \geq 6$, we have $\Aut(H_n) = S_{n-4}$,  permuting the seeds of the dandelion arbitrarily. (At $n=4$ we have $G_4 = H_4 = P_4$ and at $n=5$, $H_n = P_5$.) For ease of comparison with $\Aut(G_n)$, let us suppose the 4 non-seed vertices are $n-3,n-2,n-1,n$, so $\Aut(H_n)$ is the copy of $S_{n-4}$ acting on $[n-4]$.

Since the cycle indices of these groups are distinct (which follows from the groups having different sizes), by Theorem~\ref{thm:cycle-index} we have that $G_n$ and $H_n$ are distinguished by a multiplicity obstruction.

It remains to show that they are not distinguished by occurrence obstructions. By standard representation theory (which we have recalled in the proof of Theorem~\ref{thm:cycle-index}) we have that an irrep $V$ of $S_n$ occurs in the permutation representation $S_n / \Gamma$ if and only if $V$, when restricted to $\Gamma$, contains a non-zero $\Gamma$-invariant vector.  To show that the occurrence information for $\Aut(G_n)$ and $\Aut(H_n)$ are the same, it is necessary and sufficient to show that for every irrep $V_{\lambda}$ of $S_n$, $V_\lambda$ contains $\Aut(G_n)$-invariants if and only if it contains $\Aut(H_n)$ invariants. Since in our case $\Aut(H_n) \leq \Aut(G_n)$, for each $S_n$-irrep that contains an $S_{n-4}$ invariant, we merely need to check that it also contains, among the $S_{n-4}$-invariants, an invariant for the additional copy of $C_2$ in $\Aut(G_n)$, generated by $\sigma := (n-3,n)(n-2,n-1)$.

Now, if $v \in V_{\lambda}$ is any vector, then $v + \sigma(v)$ is $C_2$-invariant, and is non-zero if and only if $\sigma(v) \neq -v$. Furthermore, since the action of $C_2=\langle \sigma \rangle$ commutes with the action of $S_{n-4}$, if $v$ is $S_{n-4}$-invariant, then $v + \sigma(v)$ is also $S_{n-4}$-invariant (hence, $S_{n-4} \times C_2$-invariant). Thus, whenever an $S_n$ irrep contains an $S_{n-4}$-invariant vector, it suffices to find such a vector $v$ such that, additionally, $\sigma(v) \neq -v$. 

From the branching rule for restriction from $S_n$ down to $S_{n-1}$ (see Theorem~\ref{thm:branching}), and then iterated down to $S_{n-4}$, an $S_n$ irrep contains an $S_{n-4}$ invariant if and only if the corresponding partition $\lambda$ has $\lambda_1 \geq n-4$. Leaving us exactly 12 cases to check. 

In each of these cases, we first construct a vector $v$ that is $S_{n-4}$-invariant to which we can apply $\sigma$. Then we show by explicit computation that $\sigma(v) \neq -v$. To simplify this computation, it turns out we don't need to compute the entire invariant vector, just its ``leading term,'' in the following sense.

Here we follow the presentation by Peel \cite[Theorem~5.1]{peel} (covered nicely in Sagan's textbook \cite[Theorem~2.8.3]{Sagan}). Namely, when we restrict an $S_n$-irrep $V$ indexed by partition $\lambda$ down to $S_{n-1}$, we define the following chain of subspaces: suppose that among standard Young tableaux, $n$ can appear in $k$ distinct places (these are sometimes called the inner corners of the shape $\lambda$). Let $V_i$ be the subspace of $V$ spanned by the standard Young tableaux in which $n$ occurs among the top-most $i$ places among those places it can occur. Then each $V_i$ is an $S_{n-1}$-submodule, and $V_i / V_{i-1}$ is the $S_{n-1}$-irrep indexed by the shape which is $\lambda$ minus the position of $n$ in $V_i$. Now, since $V_i$ is an $S_{n-1}$-submodule, note that when we apply a permutation $\pi \in S_{n-1}$ to some standard Young tableaux in $V_i$, we get back a linear combination of standard tableaux in $V_i \backslash V_{i-1}$ and those in $V_{i-1}$; we refer to the former as the ``leading terms'' and to the latter as ``lower-order terms.''

We may then iterate this to get down from $S_n$ to $S_{n-4}$. From the preceding paragraph we see that in order to reach an $S_{n-4}$-invariant, any boxes below the first row must be filled with a subset of $\{n-3,n-2,n-1,n\}$, with the remainder of that set occurring on the right-hand side of the first row. Since our $C_2$ action is $(n-3,n)(n-2,n-1)$, we only focus on the position of these four; when we write our Young tableaux we leave the other boxes blank or implicit. Note that if $\sigma(v)=-v$, then in particular the leading term(s) of $\sigma(v)$ must be the negative of the leading term(s) of $v$ (since the standard Young tableaux are linearly independent from one another). So rather than compute $v$ explicitly and the action of $C_2$ on it, we focus on its leading terms. We write each calculation as the presentation of the selected $v$ followed by the $C_2$ action $\sigma$ applied to $v$.

To determine the action of $S_n$ on Young tableaux of a given shape, we use the standard ``straightening rule'' for the action of $S_n$ on Young tableaux; we refer the reader to \cite[Section~7.4]{Fulton} for details.

For this calculation, we assume $n \geq 8$; we discuss this further, as well as the cases $n=6,7$ after the calculations:

\begin{enumerate}
\item $\lambda = (n)$. This is the trivial representation of $S_n$, so restriction to any subgroup is still trivial, and thus contains a non-zero invariant.

\newcommand{\blank}{\none}

\ytableausetup
{mathmode, boxframe=normal, boxsize=2em, notabloids} 

\item $\lambda = (n-1,1)$. 
\[
v=\begin{ytableau}
    \blank & \blank & \scriptstyle n-3 & \scriptstyle n-2 & \scriptstyle n-1 \\
    n
\end{ytableau}\;;\quad\sigma(v)=\begin{ytableau}
   \blank & \blank & n & \scriptstyle n-1 & \scriptstyle n-2 \\
   \scriptstyle n-3 
\end{ytableau}=\begin{ytableau}
   \blank & \blank & \scriptstyle n-2 & \scriptstyle n-1 & n \\
   \scriptstyle n-3 
\end{ytableau}\ne -v
\]

\item $\lambda = (n-2,2)$.
\[
v=\begin{ytableau}
    \blank & \blank & \blank & \scriptstyle n-3 & \scriptstyle n-2 \\
    \scriptstyle n-1 & n
\end{ytableau}\;;\quad\sigma(v)=\begin{ytableau}
    \blank & \blank & \blank & n & \scriptstyle n-1 \\
    \scriptstyle n-2 & \scriptstyle n-3
\end{ytableau}=\begin{ytableau}
   \blank & \blank & \blank & \scriptstyle n-1 & n \\
   \scriptstyle n-3 & \scriptstyle n-2
\end{ytableau}\ne -v
\]

\item $\lambda = (n-2,1,1)$.
\[
v=\begin{ytableau}
    \blank & \blank & \scriptstyle n-3 & \scriptstyle n-2 \\
    \scriptstyle n-1 \\
    n
\end{ytableau}\;;\quad\sigma(v)=\begin{ytableau}
    \blank & \blank & n & \scriptstyle n-1 \\
    \scriptstyle n-2 \\
    \scriptstyle n-3
\end{ytableau}=-\begin{ytableau}
   \blank & \blank & \scriptstyle n-1 & n \\
   \scriptstyle n-3 \\
   \scriptstyle n-2
\end{ytableau}\ne -v
\]

\item $\lambda = (n-3,3)$.
\[
v=\begin{ytableau}
    \blank & \blank & \blank & \blank & \scriptstyle n-3 \\
    \scriptstyle n-2 & \scriptstyle n-1 & n
\end{ytableau}\;;\quad\sigma(v)=\begin{ytableau}
    \blank & \blank & \blank & \blank & n \\
    \scriptstyle n-1 & \scriptstyle n-2 & \scriptstyle n-3
\end{ytableau}=\begin{ytableau}
   \blank & \blank & \blank & \blank & n \\
   \scriptstyle n-3 & \scriptstyle n-2 & \scriptstyle n-1
\end{ytableau}\ne -v
\]

\item $\lambda = (n-3,2,1)$.
\[
v=\begin{ytableau}
    \blank & \blank & \blank & \scriptstyle n-3 \\
    \scriptstyle n-2 & \scriptstyle n-1 \\
    n
\end{ytableau}\;;\quad\sigma(v)=\begin{ytableau}
    \blank & \blank & \blank & n \\
    \scriptstyle n-1 & \scriptstyle n-2 \\
    \scriptstyle n-3
\end{ytableau}=-\begin{ytableau}
   \blank & \blank & \blank & n \\
   \scriptstyle n-3 & \scriptstyle n-2 \\
   \scriptstyle n-1
\end{ytableau}\ne -v
\]

\item $\lambda = (n-3,1,1,1)$.
\[
v=\begin{ytableau}
    \blank & \blank & \scriptstyle n-3 \\
    \scriptstyle n-2 \\
    \scriptstyle n-1 \\ 
    n
\end{ytableau}\;;\quad\sigma(v)=\begin{ytableau}
    \blank & \blank & n \\
    \scriptstyle n-1 \\
    \scriptstyle n-2 \\ 
    \scriptstyle n-3
\end{ytableau}=-\begin{ytableau}
   \blank & \blank & n \\
   \scriptstyle n-3 \\
   \scriptstyle n-2 \\ 
   \scriptstyle n-1
\end{ytableau}\ne -v
\]

\item $\lambda = (n-4,4)$.
\[
v=\begin{ytableau}
    \scriptstyle n-3 & \scriptstyle n-2 & \scriptstyle n-1 & n
\end{ytableau}\;;\quad\sigma(v)=\begin{ytableau}
    n & \scriptstyle n-1 & \scriptstyle n-2 & \scriptstyle n-3
\end{ytableau}=\begin{ytableau}
   \scriptstyle n-3 & \scriptstyle n-2 & \scriptstyle n-1 & n
\end{ytableau}\ne -v
\]

\item $\lambda = (n-4,3,1)$. For this partition, we do not start with the standard filling, since the standard filling is sent to minus itself (plus lower-order terms, but we want to avoid having to calculate lower-order terms).
\[
v=\begin{ytableau}
    \scriptstyle n-3 & \scriptstyle n-2 & n \\
    \scriptstyle n-1
\end{ytableau}\;;\quad\sigma(v)=\begin{ytableau}
    n & \scriptstyle n-1 & \scriptstyle n-3 \\
    \scriptstyle n-2
\end{ytableau}=-\begin{ytableau}
   \scriptstyle n-3 & \scriptstyle n-2 & \scriptstyle n-1 \\
   n
\end{ytableau}+\begin{ytableau}
   \scriptstyle n-3 & \scriptstyle n-1 & n \\
   \scriptstyle n-2
\end{ytableau}\ne -v
\]

\item $\lambda = (n-4,2,2)$.
\[
v=\begin{ytableau}
    \scriptstyle n-3 & \scriptstyle n-2 \\
    \scriptstyle n-1 & n
\end{ytableau}\;;\quad\sigma(v)=\begin{ytableau}
    n & \scriptstyle n-1 \\
    \scriptstyle n-2 & \scriptstyle n
\end{ytableau}=\begin{ytableau}
   \scriptstyle n-3 & \scriptstyle n-2 \\
   \scriptstyle n-1 & n
\end{ytableau}\ne -v
\]

\item $\lambda = (n-4,2,1,1)$. 
\[
v=\begin{ytableau}
    \scriptstyle n-3 & \scriptstyle n-2 \\
    \scriptstyle n-1 \\
    n
\end{ytableau}\;;\quad\sigma(v)=\begin{ytableau}
    n & \scriptstyle n-1 \\
    \scriptstyle n-2 \\
    \scriptstyle n-3
\end{ytableau}=-\begin{ytableau}
   \scriptstyle n-3 & \scriptstyle n-1 \\
   \scriptstyle n-2 \\
   n
\end{ytableau}\ne -v
\]

\item $\lambda = (n-4,1,1,1,1)$. 
\[
v=\begin{ytableau}
    \scriptstyle n-3 \\ 
    \scriptstyle n-2 \\ 
    \scriptstyle n-1 \\
    n
\end{ytableau}\;;\quad\sigma(v)=\begin{ytableau}
    n \\ 
    \scriptstyle n-1 \\ 
    \scriptstyle n-2 \\
    \scriptstyle n-3
\end{ytableau}=\begin{ytableau}
   \scriptstyle n-3 \\
   \scriptstyle n-2 \\ 
   \scriptstyle n-1 \\
   n
\end{ytableau}\ne -v
\]

\end{enumerate}

Note that, in order for the above calculation to not run into any ``unexpected adjacencies'' between the boxes, even diagonal adjacencies, it suffices for $n \geq 8$ (the worst case is the partition $(n-3,3)$). For $n=6,7$, we use Sage code to explicitly check that there are multiplicity obstructions but not occurrence obstructions; we include the code in Appendix~\ref{app:code}. This completes the proof.
\end{proof}

\section{Separating invariants and the Reconstruction Conjectures} \label{sec:invariant}
In this section we give explicit invariant polynomials that identify a graph, and we show that improving the support-degree (resp. degree) of the invariants we construct is equivalent to the (Edge) Reconstruction Conjecture.

By a general theorem in invariant theory (e.\,g., \cite[Thm.~2.3.6]{DerksenKemper}), since the orbits of the $S_n$ action on graphs are finite sets (or more generally Zariski-closed, but in our case we have the stronger property of actually being finite), for any two non-isomorphic graphs $G,H$ there exists an invariant polynomial whose value on $G$ (and its orbit) and is different than its value on $H$ (and its orbit). In fact, there exists an invariant polynomial whose value on the orbit of $G$ differs from its value on  \emph{all} graphs non-isomorphic to $G$; we say such a polynomial \emph{identifies} (the isomorphism class of) $G$. However, as we are interested in complexity-theoretic aspects of these invariant polynomials, we would like to see them explicitly. We furnish two such explicit polynomials in the next two results. We doubt these are genuinely new, but being unable to find them in the literature we include their proofs here.

We can also get a similar polynomial of slightly higher degree by using ``generalized monomials'' that correspond to induced (rather than general) subgraphs. The degree is higher by a factor of at least 2, but the proof is simpler, so we start there:

\begin{proposition} \label{prop:invariant}
    For any $n$-vertex graph $G$, there exists an invariant polynomial that identifies $G$, of degree $N$ (where $N$ is the number of edges in the complete graph on $V(G)$) and support-degree $n$.
\end{proposition}

(We phrase it as above so it is clear the same proof works in the undirected case, the directed case, the bipartite case, etc., \emph{mutatis mutandis}.)

\begin{proof}
Define 
\[
I_G(X) := \frac{1}{|\Aut(G)|}\sum_{\sigma \in S_n} \sigma \left[\left(\prod_{(i,j)\in E(G)} x_{ij}\right)\left(\prod_{(i,j) \notin E(G)} (1-x_{ij})\right)\right].
\]
(``$I$'' is for ``induced subgraph.'') Each term in the above sum is a product of variables and one-minus-variables, with each variable occurring in exactly one factor of the product; it is readily apparent that such a term is $1$ on exactly one graph and $0$ on all other graphs (not just up to isomorphism). And since $I_G$ is a sum over all isomorphic copies, we get $I_G(H)=1$ if and only if $H \cong G$, and otherwise $I_G(H)=0$.
\end{proof}

Using the $S_G$ polynomials  as a building block instead of $I_G$, we can improve the degree by a factor of 2, with a slightly more complicated proof. 

\begin{proposition} \label{prop:invariant-smaller}
    For any graph $G$, there exists:
\begin{enumerate}
\item an invariant polynomial that identifies $G$  and has degree at most $\max\{2, \min\{|E(G)|, |E(G^c)|\}\}$

\item an invariant polynomial that identifies $G$ and, using $V_0(G)$ to denote the set of non-isolated vertices of $G$, has support-degree at most $\max\{4,  \min\{|V_0(G)|, |V_0(G^c)|\}\}$.
\end{enumerate}
\end{proposition}

\begin{proof}
    1. Suppose without loss of generality that the number of edges of $G$ is less than the number of non-edges (if it's not, apply the ring automorphism that effectively takes the complement of graphs: $x_{ij} \leftrightarrow 1-x_{ij}$ for all $(i,j)$). Then we claim the following polynomial satisfies the conclusion of the proposition:
    \[
    f(X) := -1 + S_G(X) + 2\left(|E(G)| - \sum_{i,j \in [n]} x_{ij}\right)^2.
    \]
    (Note that here $|E(G)|$ is an integer constant, independent of the input variables to $f$.)    We have already seen that $S_G$ is invariant; the other summand is the square of a polynomial that is readliy seen to be invariant, being the sum of a constant and \emph{all} the variables, with the same coefficient.

    Now, if $H \cong G$, then we have
    \begin{align*}
    f(H) & = -1 + S_G(H) + 2(|E(G)| - \sum x_{ij}(H))^2 \\
     & = -1 + 1 + 2(|E(G)| - |E(H)|)^2 = 0.
     \end{align*}

     On the other hand, suppose $H$ is any $n$-vertex graph that is not isomorphic to $G$. If $H$ contains an isomorphic copy of $G$ as a subgraph, then since $H \not\cong G$, $G$ must be a proper subgraph of $H$, so $|E(H)| > |E(G)|$. Thus the sum inside the second term is a nonzero integer, so its square is at least 1. We then have:
     \begin{align*}
    f(H) & = -1 + S_G(H) + 2(|E(G)| - \sum x_{ij}(H))^2 \\
     & = -1 + 1 + 2(|E(G)| - |E(H)|)^2 \\
     & \geq  0 +  2 \cdot 1 \geq 2,
     \end{align*}
    so in particular $f(H) \neq 0$.

    On the other hand, if $|E(H)| = |E(G)|$, then since $H \not\cong G$, $H$ cannot contain $G$ as a subgraph, and in that case we have:
     \begin{align*}
    f(H) & = -1 + S_G(H) + 2(|E(G)| - \sum x_{ij}(H))^2 \\
     & = -1 + 0 + 2(|E(G)| - |E(H)|)^2 \\
     & = -1 + 0 + 0 = -1 \neq 0.
     \end{align*}
    Thus, for any $H \not\cong G$, we have $f(H) \neq 0$.

    2. The same construction yields the statement on support-degree, but now we choose $G$ or its complement based on which one has fewer non-isolated vertices. The 4 comes from the fact that $(|E(G)| - \sum x_{ij})^2$ includes terms of the form $x_{ij}x_{k\ell}$.
    
\end{proof}

Whether or not the degree bounds in the preceding results can be improved is a question closely related to (standard generalizations of) various Reconstruction Conjectures. We next recall these conjectures and then make this connection explicit.

\newcommand{\multiset}[1]{\{\!\{#1\}\!\}}
The $k$-vertex deck of a graph $G$ is the multiset of isomorphism types gotten by removing $k$ vertices from $G$: $VD_k(G) := \multiset{[G \backslash U] : U \in \binom{V(G)}{k}}$, where we use $[G]$ for the isomorphism type; similarly, the $k$-edge deck of a graph is $ED_k(G) := \multiset{[G \backslash E] : E \in \binom{E(G)}{k}}$. Given some isomorphism-invariant function $f$ on the set of graphs, a graph $G$ is said to be \emph{reconstructible from $f$} if $f(G)=f(H)$ implies $G \cong H$. In this language, we have various reconstruction conjectures:

\begin{conjecture}[{Reconstruction Conjecture \cite{kellyThesis,kelly,ulam}}]
Every finite simple undirected graph on at least 3 vertices is reconstructible from its $1$-vertex deck.
\end{conjecture}

\begin{conjecture}[{Edge-Reconstruction Conjecture \cite{harary}}]
Every finite simple undirected graph with at least 4 edges is reconstructible from its $1$-edge deck.
\end{conjecture}

As recently as 2020, Farhadian proved that 
a random graph is reconstructible from the multiset of its $3\log n$-vertex subgraphs with high probability \cite{farhadian} (in the above language, reconstructible from its $(n-3\log n)$-deck). On the other hand, for any $0 < \varepsilon < 1$, there exist arbitrarily large pairs of non-isomorphic graphs that have identical multisets of $\lfloor \varepsilon n \rfloor$-vertex subgraphs \cite{nydl}.
A directed version of the Reconstruction Conjecture turned out to be false \cite{stockmeyer}, but there is a so-called New Digraph Reconstruction Conjecture that includes for each vertex $v$ the triple $([G\backslash v], \deg_{in}(v), \deg_{out}(v))$ \cite{ramachandran}. We are not aware of work generalizing the New Digraph Reconstruction Conjecture to removing more than 1 vertex.
For complexity-theoretic aspects of these conjectures and further references, see \cite{HHRT}.

\begin{theorem}[{Relating invariants and edge reconstruction, cf. \cite[Thm.~4.6]{forman}}] \label{thm:edge-recon} 
Let $k$ be a positive integer, and let $G$ be a finite simple undirected graph with $k < |E(G)|$. Then $G$ is $k$-edge reconstructible if and only if it is identified by invariant polynomials of degree at most $|E(G)|-k$.
\end{theorem}

This theorem is essentially equivalent to Forman's Theorem~4.6, though we discovered the result independently before becoming aware of Forman's paper; see Section~\ref{sec:related} for a comparison.

\begin{proof}
Suppose $G$ is identified by invariant polynomials of degree at most $|E(G)|-k$. Then for each graph $H$ on the same number of vertices as $G$, there is some graph $K$ with $|E(K)| \leq |E(G)|-k$ such that $S_K(G) \neq S_K(H)$ (for otherwise all invariants of degree at most $|E(G)|-k$ take the same values on both $G$ and $H$, by Observation~\ref{obs:invariants}). Since the value of $S_K(G)$ is the same as the multiplicity of $K$ as an edge-deleted subgraph of $G$, and the counts of $k$-edge-deleted subgraphs are determined by $ED_k(G)$, we have that for each graph $H$ not isomorphic to $G$, $ED_k(G) \neq ED_k(H)$, and thus $G$ is $k$-edge-reconstructible.

Conversely, suppose $G$ is $k$-edge-reconstructible. Then for every graph $H$ on the same number of vertices that is not isomorphic to $G$, we have $ED_k(G) \neq ED_k(H)$. In particular, there exists a graph $K$ with $|E(K)| \leq |E(G)|-k$ whose multiplicity in $ED_k(G)$ differs from its multiplicity in $ED_k(H)$. These multiplicities are precisely $S_K(G)$ and $S_K(H)$, respectively, so we get that $S_K(G) \neq S_K(H)$. Thus the invariant polynomial $S_K(X) - S_K(G)$ vanishes on $G$ but not on $H$, so $G$ is identified by invariants of degree at most $|E(G)|-k$.
\end{proof}

\begin{theorem}[{Relating invariants and vertex reconstruction, cf. \cite[Theorem~11.4]{forman}}] \label{thm:vertex-recon}
Let $k$ be a positive integer, and let $G$ be a finite simple undirected graph with $k < |V(G)|$. Then $G$ is $k$-vertex reconstructible if and only if it is identified by invariant polynomials of support-degree at most $|V(G)|-k$.
\end{theorem}

\begin{proof}
Suppose $G$ is identified by invariant polynomials of support-degree at most $|V(G)|-k$. Then for each graph $H$ on the same number of vertices as $G$, there is some graph $K$ with $|V(K)| \leq |V(G)|-k$ such that $S_K(G) \neq S_K(H)$, for otherwise all invariants of support-degree at most $|V(G)|-k$ take the same values on both $G$ and $H$, by Observation~\ref{obs:invariants}. By Möbius inversion, the census of non-induced subgraphs on $\leq |V(G)|-k$ vertices and the census of induced subgraphs on $\leq |V(G)|-k$ vertices determine one another. But the latter is precisely $VD_k(G)$. Thus we have that for each graph $H$ not isomorphic to $G$, $VD_k(G) \neq VD_k(H)$, and thus $G$ is $k$-reconstructible.

Conversely, suppose $G$ is $k$-reconstructible. Then for every graph $H$ on the same number of vertices that is not isomorphic to $G$, we have $VD_k(G) \neq VD_k(H)$. In particular, there exists a graph $K$ with $|V(K)| \leq |V(G)|-k$ whose multiplicity in $VD_k(G)$ differs from its multiplicity in $VD_k(H)$. By what we've said above, there must be some graph $K'$ on $\leq |V(G)|-k$ vertices such that $S_{K'}(G) \neq S_{K'}(H)$, and thus $G$ is identified by invariants of support-degree at most $|V(G)|-k$.
\end{proof}

\section*{Acknowledgment}
We thank Benedikt Pago for pointing us to \cite{DPS26}---in an earlier version of this paper we had posed an ``open'' question that it turned out they had already answered! We thank Nathan Lindzey for pointing out theorems on graphical regular representations that filled in the last key step for us in the proof of Proposition~\ref{prop:nonisocycleindex}.\footnote{\label{fn:llm}In fact, in an earlier version we had also posed this as an open question, which he then asked an LLM the answer for; it reconstructed what the authors already knew and then produced an example of order $3^5$, the first example of which we had become aware. The proof of Proposition~\ref{prop:nonisocycleindex} nonetheless consists almost entirely of ideas the authors had come up with before seeing the output of the LLM, with the one exception of the use of theorems on GRRs; the main new information we gained from the LLM, which we could/would have gained from our emails with Lindzey directly, was the existence of Godsil's GRR Theorem. Which we then looked up ourselves, leveraging community resources such as \url{zbmath.org} and the lovely librarians at CU Boulder InterLibrary Loan (the paper itself proved impossible to find online). This is the only place we are aware that LLMs impinged on our work in any way. 

Of this interaction, here is what N. Lindzey wrote in his email to us, reproduced here (very lightly edited) with his permission: 
\begin{quotation}
\noindent \texttt{as i was reading open question 1.6 it felt like that ought to be no, since finding graphs whose automorphism groups realize some abstract group is a pretty [continued]  \linebreak 
well-understood  area 
of algebraic graph theory -- Godsil (my co-advisor) gave several talks about this and his result on graphical regular representations of graphs. so, we could have exchanged emails, thought about this open problem together, looked up literature on graphical regular representations of groups, and then eventually smoke out a negative answer perhaps in GAP that shows open question 1.6 is negative. this would take some time, but it would have been nice and we'd (re)learn some things along the way.}

\noindent \texttt{or after lifting a couple of fingers, you send it off to ChatGPT and it's finished in minutes (see attached). it seems correct.}

\noindent \texttt{deeply unsatisfying.}
\end{quotation}
}
JU thanks Anuj Dawar for the invitation to speak at Cambridge, at which the aforementioned conversation occurred. We thank Sandra Kiefer for conversations related to graphs characterized by their symmetries. JAG thanks Emma Church for teaching him about representation stability and FI-modules, Jonah Blasiak for useful conversations over several years, and Maria Gillespie for interesting related conversations about representations of the symmetric group. Initial parts of this project were carried out as JU's master's research at U. Massachusetts, Amherst; JU thanks Cameron Musco for being his master's advisor. We thank several anonymous FOCS '26 reviewers for suggestions that significantly simplified some of our arguments and improved the paper. 
JAG was partially funded through NSF CAREER award CCF-2047756; JU was partially funded by Grochow sabbatical funds and a fellowship from the University of Colorado Boulder Department of Computer Science.

\section*{LLM Usage Statement}
To the best of our knowledge, no LLM was used in the research, writing, nor any other aspect of this paper, with the one exception of Footnote~\ref{fn:llm}.

\appendix

\section{Code for the calculation in Proposition~\ref{prop:mult-vs-occurrence}} \label{app:code}
We used the following code in SageMath version 10.7 on 22 June 2026. This code constructs the two graphs $G_n,H_n$ from the proof of Proposition~\ref{prop:mult-vs-occurrence}, then retrieves their multiplicity information as vectors---using the connection with cycle indices as in the proof of Theorem~\ref{thm:cycle-index}---and checks that there are no occurrence obstructions.

\begin{verbatim}
def graph_to_multiplicity_vec(G,n):
  Sn = SymmetricGroup(n)
  group = G.automorphism_group()
    
  # gets a list of cycle types in the order of columns of the character table
  cycle_types = [g.cycle_type() 
                    for g in Sn.conjugacy_classes_representatives()] 
        
  polynomial = list(group.cycle_index()) 
  # NB: .cycle_index() returns a polynomial which is already
  # divided by the order of the group (works with our proof of the theorem)
    
  vec = [0]*len(cycle_types) # list of 0s of length len(cycle_types)
  for monomial in polynomial:
    vec[cycle_types.index(monomial[0])] = monomial[1] # mon[0]==term, mon[1]==coeff
  M = Sn.character_table() 
  return(M*vector(vec)) # multiply character table by normalized cycle index vector

n = 6 # number of vertices; replace as needed

# dandelion graph
H = Graph({0:[1],1:[2],2:[3],3:[i for i in range(4,n)]})

# dandelion with seeds blown off
G = Graph({0:[1],1:[2],2:[3]}) 
for i in range(4,n):
  G.add_vertex()

G_mult = graph_to_multiplicity_vec(G,n)
H_mult = graph_to_multiplicity_vec(H,n)

G.show()
print(G_mult)
H.show()
print(H_mult)

for i in range(len(G_mult)):
  if (G_mult[i] == 0) != (H_mult[i] == 0):
    print(False) # this line should never be reached
\end{verbatim}

\section{Code for the calculation in Proposition~\ref{prop:nonisocycleindex} and Example~\ref{ex:nonisocycleindex}} \label{app:code2}
We used the following code in SageMath version 10.7 on 19 August 2026. This code checks all non-abelian groups of a given order and compares their cycle types.

\begin{verbatim}
from itertools import filterfalse
from math import log2

def nonisomorphic_GRR_groups_with_same_cycle_index(o):
  # Get all non-abelian small groups of order o. Requires GAP
  # Filter out generalized dicyclic groups and the 13 exceptions
  # Add back in C_2^k if o is a power of 2
  Gs = list(gap(o).AllSmallGroups("IsAbelian",False)) 
  Gs = list(filterfalse(is_generalized_dicyclic,Gs)) 
  Gs = list(filterfalse(is_exceptional,Gs)) 
  if isinstance(log2(o),int):
    Gs.append(gap.DirectProduct(gap.ListWithIdenticalEntries(log2(o),gap.SmallGroup(2,1))))

  # Turn group into regular representation and tag with small group ID. Requires GAP    
  Gs2 = [(PermutationGroup(gap_group = \
    gap.Image(gap.RegularActionHomomorphism(g))),gap.IdGroup(g)) for g in Gs] 
    
  # Group all groups with matching cycle index
  cycle_indexes = {}
  for g in Gs2:
    cis = str(g[0].cycle_index())
    if cis in cycle_indexes:
      cycle_indexes[cis].append(g[1])
    else:
      cycle_indexes[cis] = [g[1]]
  return dict(filter(lambda item: len(item[1])>1, cycle_indexes.items()))

# Function checks if a group is generalized dicyclic
def is_generalized_dicyclic(gap_g): 
  # Get all normal subgroups and filter out non-abelian and non-index-2 groups
  ns = gap.NormalSubgroups(gap_g)
  ansi2 = filter(lambda sg: is_abelian_index_2(sg,gap_g),list(ns))
    
  # Checking definition of generalized dicyclic
  for sg in ansi2:
    elements = gap.Difference(gap_g,sg)
    for x in elements:
      if x^2 in sg and gap.Order(x^2)==2 \
      and all(x*a*gap.Inverse(x) == gap.Inverse(a) for a in gap.Elements(sg)):
        return True
  return False

# Function checks if a subgroup is an abelian subgroup of index 2
def is_abelian_index_2(gap_sg,gap_g): 
  return gap.IsAbelian(gap_sg) and gap.Index(gap_g,gap_sg)==2

# Function checks if a group is one of the 13 exceptions
def is_exceptional(gap_g): 
  exceptions = ['[ 4, 2 ]', # C_2^2
                '[ 8, 5 ]', # C_2^3
                '[ 16, 14 ]', # C_2^4
                '[ 6, 1 ]', # S_3
                '[ 8, 3 ]', # D_4
                '[ 10, 1 ]', # D_5
                '[ 12, 3 ]', # A_4
                '[ 16, 13 ]', # C_4 o D_4
                '[ 16, 6 ]', # M_4(2)
                '[ 18, 4 ]', # C_3 s D_3
                '[ 27, 3 ]', # He_3
                '[ 24, 11 ]', # C_3 x Q_8
                '[ 32, 26 ]'] # C_4 x Q_8
  return str(gap.IdGroup(gap_g)) in exceptions

examples = nonisomorphic_GRR_groups_with_same_cycle_index(32) # should be smallest index with example
for elem in examples.values():
  print(elem)
\end{verbatim}

\section{Corrections to the 13 Exceptional non-GRR Groups}
Both Hetzel's Diplomarbeit and Godsil's GRR classifications \cite{Hetzel,Godsil} include typos surrounding the 13 exceptional groups that do not have GRRs, compared to the original papers that proved those groups were non-GRR. Hetzel provides small group IDs for all of the groups he lists, but in the current Small Groups Database as available in GAP through SageMath, they do not match the descriptions he gives. (Some of the IDs are only off by one, so we hypothesize that the small group IDs are from an older version of the small groups database.) Here is the full list with typos corrected. The groups that were presented in \cite{Godsil} as finite presentations were checked with Sage and GAP. The code for which is below.

\begin{enumerate}
    \item \href{https://people.maths.bris.ac.uk/~matyd/GroupNames/1/C2%5E2.html}{$C_2^2$} (4,2)
    \item \href{https://people.maths.bris.ac.uk/~matyd/GroupNames/1/C2%5E3.html}{$C_2^3$} (8,5)
    \item \href{https://people.maths.bris.ac.uk/~matyd/GroupNames/1/C2%5E4.html}{$C_2^4$} (16,14)
    \item \href{https://people.maths.bris.ac.uk/~matyd/GroupNames/1/S3.html}{$D_3$} (6,1) (Godsil uses the notation $D_6$; here we stick with the notation $D_n$ for the dihedral group of order $2n$)
    \item \href{https://people.maths.bris.ac.uk/~matyd/GroupNames/1/D4.html}{$D_4$} (8,3) 
    \item \href{https://people.maths.bris.ac.uk/~matyd/GroupNames/1/D5.html}{$D_5$} (10,1) 
    \item \href{https://people.maths.bris.ac.uk/~matyd/GroupNames/1/A4.html}{$A_4$} (12,3)
    \item \href{https://people.maths.bris.ac.uk/~matyd/GroupNames/1/C4oD4.html}{$C_4\circ D_4$} (16,13) a.\,k.\,a. the Pauli group (Godsil presents this group as $\langle a,b,c\mid a^2=b^2=c^2=e,abc=bca=cab\rangle$)
    \item \href{https://people.maths.bris.ac.uk/~matyd/GroupNames/1/M4(2).html}{$M_4(2)$} (16,6) (Godsil presents this group as $\langle a,b\mid a^8=b^2=e,b^{-1}ab=a^5\rangle$)
    \item \href{https://people.maths.bris.ac.uk/~matyd/GroupNames/1/C3sS3.html}{$C_3\rtimes D_3$} (18,4) (Godsil presents this group as $\langle a,b,c\mid a^3=b^3=c^2=e,ab=ba,(ac)^2=(bc)^2=e\rangle$)
    \item \href{https://people.maths.bris.ac.uk/~matyd/GroupNames/1/He3.html}{$He_3$} (27,3) a.\,k.\,a. the Heisenberg group (Godsil's presentation has a typo. Nowitz and Watkins give a fairly standard presentation of this group as $\langle a,b,c\mid a^3=b^3=c^3=e,ac=ca,bc=cb,ab=bac\rangle$)
    \item \href{https://people.maths.bris.ac.uk/~matyd/GroupNames/1/C3xQ8.html}{$C_3\times Q_8$} (24,11)
    \item \href{https://people.maths.bris.ac.uk/~matyd/GroupNames/1/C4xQ8.html}{$C_4\times Q_8$} (32,26)
\end{enumerate}

We used the following code in SageMath version 10.7 on 19 August 2026. This code constructs the groups 8--11 above in GAP and then prints their small group ID's.

\begin{verbatim}
# Groups 8, 10, and 11
F.<a,b,c> = FreeGroup()

G8 = F/[a^2,b^2,c^2,a*b*c*(b*c*a)^-1,b*c*a*(c*a*b)^-1,c*a*b*(a*b*c)^-1]
print(G8.group_id())

G10 = F/[a^3,b^3,c^2,a*b*(b*a)^-1,(a*c)^2,(b*c)^2]
print(G10.group_id())

G11 = F/[a^3,b^3,c^3,a*c*(c*a)^-1,b*c*(c*b)^-1,a*b*(b*a*c)^-1]
print(G11.group_id())

# Group 9
H.<a,b> = FreeGroup()
G9 = H/[a^8,b^2,b*a^5*(a*b)^-1]
print(G9.group_id())
\end{verbatim}

\bibliographystyle{alphaurl}
\bibliography{refs.bib}

\end{document}